\RequirePackage{paralist}
\documentclass{theoretics}

\title{Fast Symbolic Algorithms for Omega-Regular Games under Strong Transition Fairness}

\ThCSauthor[affil1]{Tamajit Banerjee}{cs1190408@iitd.ac.in}
\ThCSauthor[affil2]{Rupak Majumdar}{rupak@mpi-sws.org}[0000-0003-2136-0542]
\ThCSauthor[affil3]{Kaushik Mallik}{kaushik.mallik@ist.ac.at}[0000-0001-9864-7475]
\ThCSauthor[affil2]{Anne-Kathrin Schmuck}{akschmuck@mpi-sws.org}[0000-0003-2801-639X]
\ThCSauthor[affil4]{Sadegh Soudjani}{sadegh.soudjani@ncl.ac.uk}[0000-0003-1922-6678]

\ThCSaffil[affil1]{Department of Computer Science and Engineering, Indian Institute of Technology Delhi, India}
\ThCSaffil[affil2]{Max Planck Institute for Software Systems, Germany}
\ThCSaffil[affil3]{Institute of Science and Technology Austria, Austria}
\ThCSaffil[affil4]{Newcastle University, UK}

\ThCSyear{2023}
\ThCSarticlenum{4}
\ThCSreceived{Feb 16, 2022}
\ThCSrevised{Aug 1, 2022}
\ThCSaccepted{Nov 21, 2022}
\ThCSpublished{Feb 23, 2023}
\ThCSkeywords{Symbolic fixpoint algorithm, graph games, strong transition fairness, turn-based stochastic games}
\ThCSdoi{10.46298/theoretics.23.4}
\ThCSshortnames{T. Banerjee, R. Majumdar, K. Mallik, A. Schmuck, S. Soudjani}
\ThCSshorttitle{Symbolic Algorithms for $\omega$-Regular Games under Strong Transition Fairness}
\ThCSthanks{A previous version of this paper has appeared in TACAS 2022. Authors ordered alphabetically. T.~Banerjee was interning with MPI-SWS when this research was conducted. R.~Majumdar and A.-K.~Schmuck are partially supported by DFG project 389792660 TRR 248–CPEC. A.-K.~Schmuck is additionally funded through DFG project (SCHM 3541/1-1). K.~Mallik is supported by the ERC project ERC-2020-AdG 101020093.}

\usepackage{amscd}
\usepackage{amsmath}
\usepackage{mathtools}

\usepackage{xcolor}
 \usepackage{array}

\usepackage{bbm}
\usepackage{comment}
\usepackage{enumerate}

\usepackage{xspace}
\usepackage{paralist}
\usepackage{xifthen}
\usepackage{url}
\usepackage{csquotes}
\usepackage{wrapfig}
\usepackage{multirow}
\usepackage[binary-units=true]{siunitx} 

\usepackage{tikz}
\usetikzlibrary{trees,decorations,arrows,automata,shadows,positioning,plotmarks,backgrounds,shapes}
\usetikzlibrary{calc,matrix,fit,petri,decorations.markings,decorations.pathmorphing,patterns,intersections,decorations.text}
\usepackage{pgfplots}
\usepackage{pgfplotstable}

\tikzstyle{mystate}=[state,inner sep=3pt,minimum size=20pt,line width=0.2mm]
\tikzstyle{fstate}=[state,accepting,inner sep=2pt,minimum size=3pt]
\tikzstyle{istate}=[state,initial,inner sep=2pt,minimum size=3pt]
\tikzstyle{mysquare}=[inner sep=3pt,minimum size=15pt,line width=0.2mm]
\tikzstyle{fmysquare}=[inner sep=3pt,minimum size=15pt,line width=0.5mm,accepting]
\newcommand{\SFSAutomatEdge}[5]{\path[->](#1) edge[#4,line width=0.2mm] node[#5] {\ensuremath{#2}} (#3);}
\usepackage{subcaption}
\usepackage{booktabs}
\usepackage{xfrac}

\usepackage{etoc}
\etocsettocdepth{3}

\usepackage{titletoc}

\usetikzlibrary{arrows.meta}
\ThCStikzlinewidths
\tikzset{every picture/.style={
   thin,
= {Straight Barb[scale=0.85,sep]},
   scale=\ThCSscalefactor
}}

\newcommand{\half}{\sfrac{1}{2}}
\newcommand{\fairsyn}{\textsf{Fairsyn}\xspace}

\newcommand{\game}{\mathcal{G}}
\newcommand{\El}{E^{\ell}}
\newcommand{\Vl}{V^{\ell}}
\newcommand{\Gl}{\game^{\ell}}
\newcommand{\twohalf}{2\sfrac{1}{2}}
\newcommand{\onehalf}{1\sfrac{1}{2}}
\newcommand{\Inf}{\mathrm{Inf}}

\newcommand{\Dr}{\mathit{Derand}}

\newcommand{\cpre}{\mathrm{Cpre}}

\newcommand{\apre}{\mathrm{Apre}}

\newcommand{\epre}{\mathrm{Pre}^{\exists}}
\newcommand{\elpre}{\mathrm{Lpre}^{\exists}}
\newcommand{\alpre}{\mathrm{Lpre}^{\forall}}
\newcommand{\eapre}{\mathrm{Pre}^{\forall}}

\newcommand{\set}[1]{\lbrace #1 \rbrace}
\newcommand{\tup}[1]{\langle #1 \rangle}
\newcommand{\dom}{\mathsf{dom}}

\newcommand\tuple[1]{{\langle #1 \rangle}}

\newcommand{\FR}{\mathcal{R}}

\newcommand{\FgR}{\widetilde{\mathcal{R}}}

\newcommand{\FP}{\mathcal{C}}

\newcommand{\Fc}{\mathcal{F}}

\newcommand{\Gc}{\mathbf{G}}

\newcommand{\vl}{\alpha}
\newcommand{\vR}{\varphi}

\newcommand{\vP}{\varphi}

\newcommand{\vGRo}{\varphi}

\newcommand{\WlR}{\mathcal{W}}

 \newcommand{\ps}[1]{\ensuremath{~{}^{#1}\!}}

\newcommand{\rank}[1]{\mathop{\mathrm{rank}\ifthenelse{\isempty{#1}}{}{(#1)}}}
\newcommand{\ranko}[1]{\mathop{\overline{\mathrm{rank}}\ifthenelse{\isempty{#1}}{}{(#1)}}}

\newcommand{\Ct}{\widetilde{\mathcal{C}}}
\newcommand{\Pt}[1]{\widetilde{P}_{\setminus #1}}
\newcommand{\Xt}{\widetilde{X}}
\newcommand{\Yt}{\widetilde{Y}}
\newcommand{\Zt}{\widetilde{Z}}
\newcommand{\Xo}{\overline{X}}
\newcommand{\Yo}{\overline{Y}}
\newcommand{\Zo}{\overline{Z}}

\newcommand{\Zc}{\check{Z}}
\newcommand{\Dt}{\widetilde{D}}
\newcommand{\delt}{\widetilde{\delta}}

\newcommand{\To}{\overline{T}}
\newcommand{\Qo}{\overline{Q}}

\newcommand{\REFlem}[1]{\text{Lemma}~\ref{#1}}
\newcommand{\REFthm}[1]{\text{Theorem}~\ref{#1}}
\newcommand{\REFdef}[1]{Definition~\ref{#1}}

\newcommand{\REFrem}[1]{Remark~\ref{#1}}
\newcommand{\REFsec}[1]{Section~\ref{#1}}

\newcommand{\REFprop}[1]{Proposition~\ref{#1}}
\newcommand{\REFfig}[1]{Figure~\ref{#1}}

\newcommand{\REFapp}[1]{Appendix~\ref{#1}}

\newcommand{\REFtab}[1]{Table~\ref{#1}}

\makeatletter 
\newif\ifFIRST
\newif\ifSECOND
\let\LISTOP\relax
\newcommand{\List}[4][\;]{#3#1%
        \FIRSTtrue
        \@for\i:=#2\do{%
        \ifFIRST\LISTOP{\i}\FIRSTfalse\else,\LISTOP{\i}\fi%
        }%
        #1#4%
        \let\LISTOP\relax
}
\makeatother

\makeatletter

\newcommand{\propNeg}{\@ifstar\propNegStar\propNegNoStar}
\newcommand{\propNegStar}[1]{\ensuremath{\left(\propNegNoStar{#1}\right)}}
\newcommand{\propNegNoStar}[2][\cdot]{\ensuremath{\neg\ifthenelse{\isempty{#2}}{#1}{#2}}}

\newcommand{\propConj}{\@ifstar\propConjStar\propConjNoStar}
\newcommand{\propConjStar}[2]{\ensuremath{\left(\propConjNoStar{#1}{#2}\right)}}
\newcommand{\propConjNoStar}[3][\cdot]{\ensuremath{\ifthenelse{\isempty{#2}}{#1}{#2}\wedge\ifthenelse{\isempty{#3}}{#1}{#3}}}

\newcommand{\propDisj}{\@ifstar\propDisjStar\propDisjNoStar}
\newcommand{\propDisjStar}[2]{\ensuremath{\left(\propDisjNoStar{#1}{#2}\right)}}
\newcommand{\propDisjNoStar}[3][\cdot]{\ensuremath{\ifthenelse{\isempty{#2}}{#1}{#2}\vee\ifthenelse{\isempty{#3}}{#1}{#3}}}

\newcommand{\propImp}{\@ifstar\propImpStar\propImpNoStar}
\newcommand{\propImpStar}[2]{\ensuremath{\left(\propImpNoStar{#1}{#2}\right)}}
\newcommand{\propImpNoStar}[3][\cdot]{\ensuremath{\ifthenelse{\isempty{#2}}{#1}{#2}\Rightarrow\ifthenelse{\isempty{#3}}{#1}{#3}}}

\newcommand{\propAequ}{\@ifstar\propAequStar\propAequNoStar}
\newcommand{\propAequStar}[2]{\ensuremath{\left(\propAequNoStar{#1}{#2}\right)}}
\newcommand{\propAequNoStar}[3][\cdot]{\ensuremath{\ifthenelse{\isempty{#2}}{#1}{#2}\Leftrightarrow\ifthenelse{\isempty{#3}}{#1}{#3}}}

\newcommand{\AllQ}{\@ifstar\AllQStar\AllQNoStar}
\newcommand{\AllQStar}[3][\;]{\ensuremath{\left(\forall #2#1.#1#3\right)}}
\newcommand{\AllQNoStar}[3][\;]{\ensuremath{\forall #2#1.#1#3}}
\newcommand{\AllQu}{\@ifstar\AllQuStar\AllQuNoStar}
\newcommand{\AllQuStar}[3][\;]{\ensuremath{\left(\forall^{\infty} #2#1.#1#3\right)}}
\newcommand{\AllQuNoStar}[3][\;]{\ensuremath{\forall^{\infty} #2#1.#1#3}}

\newcommand{\ExQ}{\@ifstar\ExQStar\ExQNoStar}
\newcommand{\ExQStar}[3][\;]{\ensuremath{\left(\exists #2#1.#1#3\right)}}
\newcommand{\ExQNoStar}[3][\;]{\ensuremath{\exists #2#1.#1#3}}

\newcommand{\NExQ}{\@ifstar\NExQStar\NExQNoStar}
\newcommand{\NExQStar}[3][\;]{\ensuremath{\left(\nexists #2#1.#1#3\right)}}
\newcommand{\NExQNoStar}[3][\;]{\ensuremath{\nexists #2#1.#1#3}}

\newcommand{\UniqueQ}{\@ifstar\UniqueQStar\UniqueQNoStar}
\newcommand{\UniqueQStar}[3][\;]{\ensuremath{\left(\exists! #2#1.#1#3\right)}}
\newcommand{\UniqueQNoStar}[3][\;]{\ensuremath{\exists! #2#1.#1#3}}

  \newlength{\SFS@HEIGHT}
  \newlength{\SFS@WIDTH}
  \newcommand{\SplitX}[2]{
            \settoheight{\SFS@HEIGHT}{$#2$}
            \settowidth{\SFS@WIDTH}{$#2$}
            \mbox{\begin{tikzpicture}[baseline=(current bounding box.center)]
            \node[] (E) at (0,0) {$#1$};
            \node[inner sep=0pt] (F) at ($(E.south west)+(1ex,-1ex)+(3ex+.5\SFS@WIDTH,-\SFS@HEIGHT)$) {$#2$};
            \node[] (E) at (0,0) {\phantom{$#1$}};
            \draw[fill] ($(E.east)+(1ex,0ex)$) circle (.2ex);
            \draw[-] ($(E.east)+(1ex,0ex)$) -- ($(E.south east)+(1ex,-0.5ex)$) -- ($(E.south west)+(1ex,-0.5ex)$) -- ($(E.south west)+(1ex,-1ex)-(0,\SFS@HEIGHT)$) -- ($(E.south west)+(2.5ex,-1ex)-(0,\SFS@HEIGHT)$);
            \draw[fill] ($(E.south west)+(2.5ex,-1ex)-(0,\SFS@HEIGHT)$) circle (.2ex);
            \end{tikzpicture}}}
  \newcommand{\SplitS}[2]{
            \settoheight{\SFS@HEIGHT}{$#2$}
            \settowidth{\SFS@WIDTH}{$#2$}
            \mbox{\begin{tikzpicture}[baseline=(current bounding box.center)]
            \node[] (E) at (0,0) {$#1$};
            \node[inner sep=0pt] (F) at ($(E.south west)+(1ex,0.5ex)+(0ex+.5\SFS@WIDTH,-\SFS@HEIGHT)$) {$#2$};
            \end{tikzpicture}}}     

                     \newcommand{\SetComp}[3][]{\{#1#2#1\mid#1#3#1\}}

\newcommand{\Set}[2][]{\List[#1]{#2}{\{}{\}}}
\newcommand{\VSet}[2][]{\let\LISTOP\val\List[#1]{#2}{\{}{\}}}

\newcommand{\VTuple}[2][]{\let\LISTOP\val\List[#1]{#2}{(}{)}}

\newcommand{\UNION}{\@ifstar\UNIONStar\UNIONNoStar}
\newcommand{\UNIONStar}[2]{\ensuremath{\left(\UNIONNoStar{#1}{#2}\right)}}
\newcommand{\UNIONNoStar}[2]{\ensuremath{\ifthenelse{\isempty{#1}}{\cdot}{#1}\cup\ifthenelse{\isempty{#2}}{\cdot}{#2}}}

\newcommand{\UNIOND}{\@ifstar\UNIONDStar\UNIONDNoStar}
\newcommand{\UNIONDStar}[2]{\ensuremath{\left(\UNIONDNoStar{#1}{#2}\right)}}
\newcommand{\UNIONDNoStar}[2]{\ensuremath{\ifthenelse{\isempty{#1}}{\cdot}{#1}\uplus\ifthenelse{\isempty{#2}}{\cdot}{#2}}}

\newcommand{\SETMINUS}{\@ifstar\SETMINUSStar\SETMINUSNoStar}
\newcommand{\SETMINUSStar}[2]{\ensuremath{\left(\SETMINUSNoStar{#1}{#2}\right)}}
\newcommand{\SETMINUSNoStar}[2]{\ensuremath{\ifthenelse{\isempty{#1}}{\cdot}{#1}\setminus\ifthenelse{\isempty{#2}}{\cdot}{#2}}}

\newcommand{\INTERSECT}{\@ifstar\INTERSECTStar\INTERSECTNoStar}
\newcommand{\INTERSECTStar}[2]{\ensuremath{\left(\INTERSECTNoStar{#1}{#2}\right)}}
\newcommand{\INTERSECTNoStar}[2]{\ensuremath{\ifthenelse{\isempty{#1}}{\cdot}{#1}\cap\ifthenelse{\isempty{#2}}{\cdot}{#2}}}

\newcommand{\CARTPROD}{\@ifstar\CARTPRODStar\CARTPRODNoStar}
\newcommand{\CARTPRODStar}[2]{\ensuremath{\left(\CARTPRODNoStar{#1}{#2}\right)}}
\newcommand{\CARTPRODNoStar}[2]{\ensuremath{\ifthenelse{\isempty{#1}}{\cdot}{#1}\times\ifthenelse{\isempty{#2}}{\cdot}{#2}}}

\newcommand{\FINCOUNT}{\@ifstar\FinCountStar\FinCountNoStar}
\newcommand{\FinCountStar}[1]{\ensuremath{\#(\ifthenelse{\isempty{#1}}{\cdot}{#1})}}
\newcommand{\FinCountNoStar}[1]{\ensuremath{\#\left(\ifthenelse{\isempty{#1}}{\cdot}{#1}\right)}}

\makeatother

\newtheorem*{theorem*}{Theorem}

\newcommand{\p}[1]{\ensuremath{\text{Player}~#1}}

\pgfplotsset{compat=1.17}  
\begin{document}
\maketitle



 
	\begin{abstract}
%
We consider fixpoint algorithms for two-player games on graphs with $\omega$-regular winning conditions, where 
the environment is constrained by a \emph{strong transition fairness} assumption.
Strong transition fairness is a widely occurring special case of strong fairness.
It requires that any execution is strongly 
fair with respect to a specified set of live edges: whenever the source vertex of a live edge is visited infinitely often along a play, 
the edge itself is traversed infinitely often along the play as well.

We show that, surprisingly, \emph{strong transition fairness} retains the 
algorithmic characteristics of the fixpoint algorithms for $\omega$-regular games---the new algorithms have the same alternation depth as the classical algorithms but invoke a new type of predecessor operator.
For example, for Rabin games with $k$ pairs under strong transition fairness, the complexity of the new algorithm is $O(n^{k+2}k!)$ symbolic steps, which 
is independent of the number of live edges in the strong transition fairness assumption. 
In contrast, \emph{strong fairness} necessarily requires increasing the alternation depth depending on the number of fairness assumptions.

We get symbolic algorithms for (generalized) Rabin, parity, and GR(1) objectives under strong transition fairness assumptions as well as a 
direct symbolic algorithm for qualitative winning in stochastic $\omega$-regular games that runs in $O(n^{k+2}k!)$ symbolic steps,
improving the state of the art.
Previous approaches for handling fairness assumptions would either increase the alternation depth of the fixpoint algorithm 
or require an up-front automata-theoretic construction that would increase the state space, or both.

We have implemented a BDD-based synthesis engine based on our algorithm. We show on a set of synthetic and real benchmarks that
our algorithm is scalable, parallelizable, and outperforms previous algorithms by orders of magnitude.

	\end{abstract}


\section{Introduction}

Symbolic algorithms for two-player graph games are at the heart of 
many problems in the automatic synthesis of correct-by-construction hardware, software, and
cyber-physical systems from logical specifications.
The problem has a rich pedigree, going back to Church \cite{Church62} and a sequence of seminal
results \cite{BL69,Rabin69,GH82,PR89,EJ88,EJ91,Zielonka98,KVSafraless}.
A chain of reductions can be used to reduce the synthesis problem for $\omega$-regular
specifications to finding winning strategies in two-player games on graphs, for which 
(symbolic) algorithms are known (see, e.g., \cite{PR88,EJ91,Zielonka98,PitermanPnueli_RabinStreett_2006}).
These reductions and algorithms form the basis for algorithmic reactive synthesis.

In practice, it is often the case that no solution exists to a given synthesis problem, but for ``uninteresting'' reasons.
For example, consider synthesizing a mutual exclusion protocol from a specification that requires (1) that at most one of two 
processes can be in the critical section at any time and (2) that a process wishing to enter the critical section is eventually allowed to do so.
As stated, there may not be a feasible solution to the problem because a process within the critical section may decide to stay there forever.
Similarly, in a synthesis problem involving concurrent threads, no solution may exist simply because the scheduler may decide never to pick a particular thread.
Fairness assumptions rule out such uninteresting conditions by constraining the possible behaviors of the environment.
The winning condition under fairness is of the form
\begin{equation}
\label{eq:win}
\mathsf{Fairness~Assumption} \ \Rightarrow \ \omega\mathsf{-regular~Specification}.
\end{equation}
For example, a fairness constraint can state that whenever a process is in its critical section, it must eventually leave it or that, 
if a thread is enabled infinitely often, then it is picked by the scheduler infinitely often. 
Similarly, a mobile robot can assume that a narrow passage is always eventually freed by other robots if it is known that all robots have distant goals they need to reach.
These examples, and many other practical instances of fairness, fall into a particular subclass of 
fairness assumptions, called \emph{strong transition fairness} 
\cite{QueilleSifakis,Francez,baierbook}.
A strong transition fairness assumption can be modeled by a set of \emph{live} environment transitions in the underlying two-player game graph.
Whenever the source vertex of a live transition is visited infinitely often,
the transition will be taken infinitely often by the environment.
Unfortunately, despite the widespread prevalence of strong transition fairness, current symbolic algorithms for solving games
do not take advantage of their special structure in the winning condition in \eqref{eq:win} and no algorithm better than those for
general (Streett) liveness assumptions is known. 

In this paper, we consider $\omega$-regular games under \emph{strong transition fairness} assumptions, which we call \emph{fair adversarial games}.  We show a surprisingly simple syntactic transformation that modifies the well-known symbolic fixpoint algorithm for Rabin games \emph{without} fairness assumptions, 
such that the modified fixpoint algorithm solves the \emph{fair adversarial} Rabin game.
To appreciate the simplicity of our modification, let us consider the well-known fixpoint algorithms for Büchi and co-B\"uchi games---particular classes of Rabin games---given by the following $\mu$-calculus formulas:
	\begin{subequations}\label{equ:intro:Parity}
	 \begin{align}\label{equ:intro:SafeBuechi:old}
	 \begin{array}{l l}
	 	\textbf{B\"uchi:} & \qquad \nu Y.~\mu X.~ 
	 \left( G\cap \cpre(Y)\right)
	 \cup \cpre(X),\\
	 	\textbf{Co-B\"uchi:} & \qquad \mu X.~\nu Y.~ 
	 	 \left(G\cap\cpre(Y)\right)\cup \cpre(X),
	 \end{array}
	\end{align}
	where $\cpre(\cdot)$ denotes the controllable predecessor operator and $G$ denotes the set of states that should be visited always eventually (Büchi) and eventually always (co-Büchi), respectively.
	In the presence of strong transition fairness assumptions on the environment, the new algorithm becomes
	\begin{align}
	\begin{array}{l l}
		\textbf{B\"uchi:}  &\qquad \phantom{\nu W.~}\nu Y.~\mu X.~
	 \left( G\cap \cpre(Y)\right)
	 \cup \textcolor{blue}{\apre}(Y,X),\label{equ:intro:SafeBuechi:new}\\
	 	\textbf{Co-B\"uchi:}  &\qquad \textcolor{blue}{\nu W.}~\mu X.~\nu Y.~
    \left(G\cap\cpre(Y)\right)\cup\textcolor{blue}{\apre}(W,X).
	\end{array}
	\end{align}
	\end{subequations}
	The only syntactic change (highlighted in blue) we make is to substitute the controllable predecessor for the $\mu$ variable $X$ by a new 
	\emph{almost sure predecessor operator} $\apre(Y,X)$ incorporating also the previous $\nu$ variable $Y$; if the fixpoint starts with a $\mu$ variable (as for co-B\"uchi), we add one outermost $\nu$ variable. 
	For the general class of Rabin games which are solved by a deeply nested fixpoint algorithm, we perform this substitution for every $\cpre(\cdot)$ operator over a $\mu$ variable.

	We prove the correctness of the outlined syntactic fixpoint transformation for fair adversarial Rabin and generalized Rabin games. This immediately results in correct algorithms for fair adversarial Safety-, (generalized) B\"uchi-, (generalized) Co-B\"uchi-, GR(1)-, and Muller games as special cases. While all mentioned reductions result in a modified fixpoint algorithm which can be obtained by directly applying the outlined syntactic transformation to the respective well known fixpoint algorithm for normal games (as shown for B\"uchi and Co-B\"uchi in \eqref{equ:intro:Parity}), we show that for fair adversarial parity games, which are also a subclass of Rabin games, the resulting fixpoint algorithm is slightly more complex than the syntactic transformation suggests. 
However, the alternation depth of both fixpoint algorithms still coincide.
	
\smallskip
Our syntactic transformation is inspired by the work of \cite{dAHK98} on symbolic fixpoint algorithms for \emph{concurrent} two-player games on finite graphs. In concurrent games, both players \emph{simultaneously and independently} choose their actions from a given vertex, and the transition relation defines a probability distribution over the set of successor vertices, given the current state and
the chosen actions. 
It was shown by \cite{dAHK98} that for B\"uchi games the set of almost-sure winning vertices 
(i.e., vertices from which the system player wins the game with probability one) can be computed by the symbolic fixpoint algorithm in \eqref{equ:intro:SafeBuechi:new}. 
The reason why the fixpoint algorithms coincide for concurrent and fair adversarial B\"uchi games is rather subtle. 
For concurrent games, it is known that optimal winning strategies may require randomization, and it is this randomization (in winning strategies) that 
\emph{induces} strong transition fairness on plays compliant with the chosen strategies.
In contrast, in fair adversarial games the environment player is \emph{constrained} by a \emph{given} strong transition fairness assumption, and computed 
(deterministic) winning strategies condition their moves on this fair behavior. 
In both cases, the fixpoint algorithm has to take possible transition fairness into account (witnessed by the use of the same $\apre(\cdot)$ operator), 
however, the conclusion drawn for the resulting winning regions for the subsequent strategy construction are substantially different in both game types. 

This observation also explains why the fixpoint algorithms for concurrent and fair adversarial games no longer coincide for \emph{co-}B\"uchi games. Here, randomized strategies 
introduce a different type of co-fairness constraint---now certain transitions are ensured to be taken only \emph{finitely often}, 
leading to yet another pre-operator used in the symbolic fixpoint algorithm for concurrent co-B\"uchi games. 
For fair adversarial co-B\"uchi games, however, we still restrict the environment player with strong transition fairness constraints 
(which might not be as helpful for a co-B\"uchi objective as for a B\"uchi objective), and by this, the fixpoint algorithm again only has to utilize the $\apre{}$ operator. 

Our main contribution in this paper is to show that the use of the $\apre(\cdot)$ operator to incorporate strong transition fairness in symbolic algorithms extends from B\"uchi games to all other types of $\omega$-regular games while retaining the algorithmic characteristics of the respective algorithms. 
It is this generalization of strong transition fairness to the full class of omega-regular games, that
allows us to obtain direct symbolic algorithms for \emph{simple stochastic games} as a byproduct.
Simple stochastic games generalize two-player graph games with an additional category of ``random'' vertices: whenever the game reaches a random vertex, a random process picks one of the outgoing edges (uniformly at random, w.l.o.g.). 
Interestingly, one can replace random vertices in simple stochastic games by environment vertices constrained by \emph{extreme fairness} (\citet{Pnueli83}). However, extreme fairness is a special case of strong transition fairness---a run is extremely fair if it is strongly transition fair for \emph{every} outgoing edge from a vertex---showing that simple stochastic games are a special case of fair adversarial games.

\smallskip
In a nutshell, the new direct symbolic algorithms for fair adversarial games developed in this paper show that, in contrast to \emph{full strong fairness}, strong \emph{transition} fairness retains algorithmic efficiency in game solving for all $\omega$-regular objectives. 
This leads to three, conceptually rather different contributions that substantially improve the state of the art.

\textbf{(I)} 
In the context of \emph{reactive synthesis under environment assumptions}, our new fair adversarial game solver enables many expressive fairness assumptions on the environment player in combination with \emph{full LTL} objectives for the system player. This extends existing work in this context. The GR(1) fragment of LTL, for example, was introduced by \citet{piterman2006synthesis} explicitly to rule out strong fairness constraints because of the absence of suitable low-depth fixpoint algorithms. Over the years, the GR(1) fragment has been extensively used as a useful logical fragment of LTL for reactive synthesis, especially in the cyber-physical
and robotics domains \cite{kress2007s,kress2009temporal,DBLP:conf/fmcad/AlurMT13,maoz2015synthesizing,Belta17}. 
Our new fair adversarial game solver enables expressive fairness assumptions for properties that go way beyond the ones expressible in GR(1). 
On the other hand, we extend the results of \citet{thistle1998control} who showed that \emph{extreme} fairness assumptions on the environment allow efficient synthesis of supervisory controllers for non-terminating processes\footnote{Supervisory controller synthesis for non-terminating processes is conceptually similar to reactive synthesis under environment assumptions but utilizes different solution algorithms \cite{schmuck2020relation}.} under Rabin specifications. 

\textbf{(II)} 
In the context of \emph{games with randomized strategies}, we show that simple \emph{stochastic} two-player games (also known as $2\sfrac{1}{2}$-player games) can be reduced to fair adversarial games.
We show that, to solve a qualitative stochastic (generalized) Rabin game, we can equivalently solve the (generalized) Rabin game under extreme fairness which is a particular fair adversarial (generalized) Rabin game. This results in a direct symbolic algorithm for this problem. Our algorithm, which runs in $O(n^{k+2}k!)$ symbolic steps for an  $n$-vertex $k$-pair stochastic Rabin game, 
improves the best known algorithm for such games given in \citet{DBLP:conf/icalp/ChatterjeeAH05}.
Their algorithm  is based on a reduction to a  $O\left(n(k+1)\right)$-vertex $(k+1)$-pair (deterministic) Rabin game and a simple analysis
indicates that it requires $O\left((n(k+1))^{k+2}(k+1)!\right)$ symbolic steps.

\textbf{(III)} In the context of \emph{efficient solutions} of $\omega$-regular games, we obtain \emph{symbolic algorithms} which solve two-player games by finding the set of states of the underlying game graph from which the game can be won. 
The benefit of symbolic approaches is that they allow efficient implementations based on manipulations of formulas (often represented using data structures such as BDDs). 
 Indeed, these fixpoint expressions are the cornerstone of many reactive synthesis tools \cite{brenguier2014abssynthe,ehlers2016slugs,michaud2018reactive}.
 Due to the simplicity of our syntactic transformation from the fixpoint algorithm for usual games to the one for fair adversarial games, existing symbolic implementations of 
reactive synthesis can be slightly modified to incorporate strong transition fairness assumptions.

We have implemented our algorithm in a symbolic reactive synthesis tool called \fairsyn.
\fairsyn uses a multi-threaded BDD library \cite{van2015sylvan} and implements an acceleration technique for the fixpoints \cite{long1994improved}.
We show on a number of synthetic benchmarks from the very large transition systems benchmark suite \cite{vlts_benchmark} that our algorithm, with
the improvements, can scale to large Rabin games and the performance scales with the number of cores.
Additionally, we evaluate our tool on two case studies, one from software synthesis \cite{rupak_krishnendu} and the other from stochastic control
synthesis \cite{dutreix2020abstraction}.
We show that \fairsyn scales well on these case studies, and outperforms a state-of-the-art stochastic game solver by an order of magnitude.
In contrast, a solver that treats transition fairness as Streett fairness does not finish on the considered case studies.
%


\section{Preliminaries}

    \paragraph{Notation:}
	We use the notation $\mathbb{N}_0$ to denote the set of natural numbers including ``$0$.''
	Given $a,b\in \mathbb{N}_0$, we use the notation $[a;b]$ to denote the set $\set{n\in \mathbb{N}_0 \mid a\leq n\leq b}$.
	Observe that, by definition, $[a;b]$ is an empty set if $a > b$.
	For any set $A \subseteq U$ defined on the universe $U$, we use the notation $\overline{A}$ to denote the complement of $A$.
	
	Let $A$ and $B$ be two sets and $R\subseteq A\times B$ be a relation.
	We use the notation $\dom(R)$ to denote the domain of $R$, which is the set $\set{a\in A\mid \exists b\in B\;.\;(a,b)\in R}$.
	For any element $a\in A$, we use the notation $R(a)$ to denote the set $\set{b\in B\mid (a,b)\in R}$, and for any element $b\in B$, we use the notation $R^{-1}(b)$ to denote the set $\set{a\in A\mid (a,b)\in R}$.
	We generalize $R(\cdot)$ to operate on sets in the following way: 
	for any $A'\subseteq A$, we write $R(A')\coloneqq \cup_{a\in A'} R(a)$, and for any $B'\subseteq B$, we write $R^{-1}(B')\coloneqq \cup_{b\in B'} R^{-1}(b)$.
	
	Given an alphabet $A$, we use the notation $A^*$ and $ A^\omega$ to denote respectively the set of all finite words and the set of all infinite words formed using the letters of the alphabet $A$. We use~$A^\infty$ to denote the set $A^*\cup A^\omega$. Given two words $a\in A^*$ and $b\in A^\infty$, we use $a\cdot b$ to denote their concatenation.

\subsection{Two-Player Games}
 \paragraph{Game Graphs:}
	We define a \emph{two-player game graph} as a tuple $\game=\tup{V,V_0,V_1,E}$, where 
	\begin{inparaenum}[(i)]
		\item $V = V_0\uplus V_1$ is a finite set of vertices\footnote{We use the terms ``vertex'' and ``state'' interchangeably in this paper.} that is partitioned into the sets $V_0$ and $V_1$;
		\item $E\subseteq (V\times V)$ is a relation denoting the set of (directed) edges;
	\end{inparaenum}
	The two players are called $\p{0}$ and $\p{1}$, who control the vertices $V_0$ and $V_1$ respectively.
		
      \paragraph{Strategies:}
	A strategy of $\p{0}$ is a function $\rho_0\colon V^*\cdot V_0\to V$ with the constraint $\rho_0(w\cdot v)\in E(v)$ for every $w\cdot v\in V^*\times V_0$.
	Likewise, a strategy of $\p{1}$ is a function $\rho_1\colon V^*\cdot V_1 \to V$ with the constraint $\rho_1(w\cdot v)\in E(v)$ for every $w\cdot v\in V^*\times V_1$.
	Of special interest is the class of memoryless strategies: 
	a strategy $\rho_0$ of $\p{0}$ is \emph{memoryless} if for every $w_1\cdot v, w_2\cdot v\in V^*\times V_0$, we have $\rho_0(w_1\cdot v)=\rho_0(w_2\cdot v)$.
	
    \paragraph{Plays:}
	Consider an infinite sequence of vertices $\pi = v^0v^1v^2\allowbreak\ldots\in V^\omega$.
	The sequence $\pi$ is called a \emph{play} over $\game $ starting at the vertex $v^0$ if for every $i\in \mathbb{N}_0$, we have $v^i\in V$ and $(v^i,v^{i+1}) \in E$.
	In our convention for denoting vertices, superscripts (ranging over $\mathbb{N}_0$) will denote the position of a vertex within a given play, whereas subscripts, either $0$ or $1$, will denote the membership of a vertex in the sets $V_0$ or $V_1$ respectively.
	Let $\rho_0$ and $\rho_1$ be a given pair of strategies of $\p{0}$ and $\p{1}$, respectively, and let $v^0$ be a given initial vertex.
	The play \emph{compliant with  $\rho_0$ and $\rho_1$} is the unique play $\pi=v^0v^1v^2\ldots$ for which for every $i\in \mathbb{N}_0$, if $v^i\in V_0$ then $v^{i+1}= \rho_0(v^0 \ldots v^i)$, and if $v^i\in V_1$ then $v^{i+1}=\rho_1(v^0\ldots v^i)$.

    \paragraph{Winning Conditions:}	
	A \emph{winning condition} $\varphi$ is a set of infinite plays over $\game $, i.e., $\varphi\subseteq V^\omega$.	
	We adopt Linear Temporal Logic (LTL) notation for describing winning conditions.
	The atomic propositions for the LTL formulae are sets of vertices, i.e., elements of the set $2^V$.
	We use the standard symbols for the Boolean and the temporal operators: ``$\lnot$'' for \emph{negation}, ``$\wedge$'' for \emph{conjunction}, ``$\vee$'' for \emph{disjunction}, ``$\rightarrow$'' for \emph{implication}, ``$\mathcal{U}$'' for \emph{until} ($A \, \mathcal{U} \, B$ means ``the play remains inside the set $A$ until it moves to the set $B$''), ``$\bigcirc$'' for \emph{next} ($\bigcirc A$ means ``the next vertex is in the set $A$''), ``$\lozenge$'' for \emph{eventually} ($\lozenge A$ means ``the play will eventually visit a vertex from the set~$A$''), and ``$\square$'' for \emph{always} ($\square A$ means ``the play will only visit vertices from the set $A$'').
	The syntax and semantics of LTL can be found in standard textbooks \cite{baierbook}.
	By slightly abusing notation, we will use $\varphi$ interchangeably to denote both the LTL formula and the set of plays satisfying~$\varphi$. Hence, we write $\pi\in\varphi$ (instead of $\pi \models \varphi$) to denote the satisfaction of the formula $\varphi$ by the play~$\pi$.

%
\paragraph{Winning Regions:}
$\p{0}$ wins a two-player game over the game graph $\game $ for a winning condition $\varphi$ from a vertex $v^0\in V$ if there is a $\p{0}$ strategy $\rho_0$ such that
for every $\p{1}$ strategy $\rho_1$, the play $\pi$ from $v^0$ compliant with $\rho_0$ and $\rho_1$ satisfies $\varphi$,
i.e., $\pi \in \varphi$.
The \emph{winning region} $\mathcal{W}\subseteq V$ for $\p{0}$ is the set of vertices from which $\p{0}$ wins the game.

\subsection{Fair Adversarial Games}\label{subsec:FAG}

Let $\game $ be a two-player game graph and let $\El\subseteq (V_1 \times V) \cap E$ be a given set of \emph{live} edges.
Let $\Vl\coloneqq \dom(\El)$ denote the set of $\p{1}$ vertices in the domain of $\El$.
Intuitively, the edges in~$\El$ represent \emph{fairness assumptions} on $\p{1}$: for \emph{every} edge $(v,v') \in \El$, if $v$ is
visited infinitely often along a play, we expect that the edge $(v,v')$ is picked infinitely often by $\p{1}$. I.e., if a vertex $v$ is visited infinitely often, \emph{every} outgoing live edge of $v$ is expected to be taken infinitely often.

We write $\Gl = \tup{\game , \El}$ to denote a game graph with live edges, and extend notions such as plays, strategies,
winning conditions, winning region, etc., from game graphs to those with live edges. 
A play $\pi$ over $\Gl$ is \emph{strongly transition fair} if it satisfies the LTL formula:
	\begin{align}\label{equ:vl}
	 	\vl\coloneqq\textstyle\bigwedge_{(v,v')\in \El} \left(\square\lozenge v\rightarrow \square\lozenge (v\wedge \bigcirc v')\right).
	\end{align}
Given $\Gl$ and a winning condition $\varphi$, $\p{0}$ wins the \emph{fair adversarial game} over $\Gl$ for the winning
condition $\varphi$  from a vertex $v^0\in V$ if $\p{0}$ wins the game over $\Gl$ for the winning condition $\vl \rightarrow \varphi$ from $v^0$.

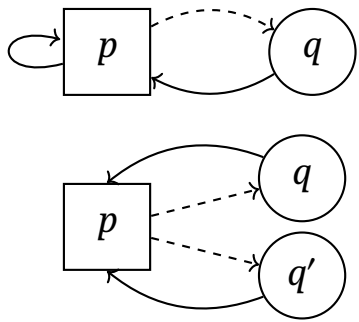
\begin{figure}
	\centering
	\begin{tikzpicture}
		\node[state,rectangle]	(a)	at	(0,0)		{$p$};
		\node[state]			(b)	[right=1.2 \ThCScm of a]		{$q$};
		\path[->]	(a)	edge[loop left]		()
						(b)	edge[bend left]			(a);
		\path[->,dashed]		(a)		edge[bend left]		(b);
		\node[state,rectangle]	(a)	at	(0,-1.8)		{$p$};
		\node[state]		(b)	at (2,-1.3)	{$q$};
		\node[state]		(c)	at (2,-2.3)	{$q'$};
		\path[->,dashed]	(a) edge[] (b) edge[] (c);
		\path[->]	(b) edge[bend right] (a.north);
		\path[->]	(c) edge[bend left] (a.south);
	\end{tikzpicture}
	\vspace{-0.1cm}
	\caption{Two fair adversarial games.}
	\label{fig:simple}
\end{figure}
We have two interesting observations about fair adversarial games.
First, live edges allow to rule out particular strategies of $\p{1}$, 
making it easier for $\p{0}$ to win in certain situations.
Consider for example a game graph (\REFfig{fig:simple} (top)) with two vertices $p$ and $q$.
Vertex $p$ (square) is a \p{1} vertex and vertex $q$ is a \p{0} vertex (circle).
The edge $(p,q)$ is a live edge (dashed).
Suppose the specification for $\p{0}$ is $\varphi=\square\lozenge q$.
If the edge $(p,q)$ were non-live,
$\p{0}$ would not win for this specification from $p$, because 
$\p{1}$ would be able to trap the game in $p$ by always choosing $p$ itself as the successor.
In contrast, $\p{0}$ wins from $p$ in the fair adversarial game, because the liveness assumption on the edge $(p,q)$ forces $\p{1}$ 
to infinitely often choose the transition to $q$.

Second, fairness assumptions modeled by live edges restrict the strategy choices of $\p{1}$ less than assuming that $\p{1}$ chooses probabilistically between these edges. 
Consider for example a fair adversarial game with one $\p{1}$ vertex $p$ (square) which has two outgoing live edges to states $q$ and $q'$; see \REFfig{fig:simple} (bottom). If $\p{1}$ chooses randomly between edges $(p,q)$ and $(p,q')$, \emph{every} finite sequence of visits to states $q$ and $q'$ will happen infinitely often with probability one. 
This is not true in the fair adversarial game. Here  $\p{1}$ is allowed to choose a particular sequence of visits to states $q$ and $q'$ (e.g., only $qq'qq'qq'qq'\hdots$), as long as both are visited infinitely often. 

\subsection{Symbolic Computations over Game Graphs}	
   \paragraph{Set Transformers:}
	Our goal is to develop symbolic fixpoint algorithms to characterize the winning region of a fair adversarial game over a game graph with live edges.
	As a first step, given~$\Gl$, we define the required symbolic transformers of sets of states.
	We define the existential, universal, and controllable predecessor operators as follows. 
	For $S\subseteq V$, we have
	\begin{subequations}\label{equ:pre}
	\begin{align}
	 \epre_0(S)&\coloneqq \SetComp{v\in V_0}{E(v)\cap S\neq \emptyset },\\
	 \eapre_1(S)&\coloneqq \SetComp{v\in V_1}{E(v)\subseteq S},~\text{and}\\
	 \cpre(S)&\coloneqq \epre_0(S)\cup\eapre_1(S)\label{equ:cpre}.
	\end{align}
\end{subequations}
Intuitively, the controllable predecessor operator $\cpre(S)$ computes the set of all states that can be controlled by $\p{0}$ to stay in $S$ after one step
regardless of the strategy of $\p{1}$.
Additionally, we define two operators which take advantage of the fairness assumption on the live edges. 
   	Given two sets $S,T\subseteq V$, we define the live-existential and almost sure predecessor operators:
	\begin{subequations}\label{equ:ell_pre}
	\begin{align}
	 \elpre(S)&\coloneqq \SetComp{v\in \Vl}{\El(v)\cap S\neq \emptyset },~\text{ and}\\
	 \apre(S,T)&\coloneqq \cpre(T)\cup\left(\elpre(T)\cap\eapre_1(S)\right).\label{equ:apre}
	\end{align}
\end{subequations}
Intuitively, the almost sure predecessor operator\footnote{We will justify the naming of this operator later in \REFrem{rem:apre}.} $\apre(S,T)$ computes the set of all states that can be controlled by 
$\p{0}$ to stay in $T$ (via $\cpre(T)$) \emph{as well as} all $\p{1}$ states in $\Vl$ that (a) will \emph{eventually} make progress towards $T$ if 
$\p{1}$ obeys its fairness-assumptions encoded in $\vl$ (through $\elpre(T)$) and (b) will never leave $S$ in the \enquote{meantime} (through $\eapre_1(S)$).
We see that all set transformers are monotonic with respect to set inclusion. Further, $\cpre(T)\subseteq \apre(S,T)$ always holds,  $\cpre(T)=\apre(S,T)$ if $\Vl=\emptyset$, and $\apre(S,T)\subseteq\cpre(S)$ if $T\subseteq S$ (see \REFlem{lem:AsubsetC} in the appendix for a proof).

\paragraph{Fixpoint Algorithms in the $\mu$-calculus:}
We use the  $\mu$-calculus \cite{Kozen83} as a convenient logical notation used to define a symbolic algorithm (i.e., an algorithm that manipulates sets of states rather then
individual states) for computing a set of states with a particular
property over a given game graph $\game $. 
The formulas of the $\mu$-calculus, interpreted over a two-player game graph~$\game $, 
are given by the grammar
\[
\varphi\; \Coloneqq \; p \mid X \mid \varphi \cup \varphi \mid \varphi \cap \varphi \mid \mathit{pre}(\varphi) \mid \mu X.\varphi \mid \nu X.\varphi
\]
where $p$ ranges over subsets of $V$, $X$ ranges over a set of formal variables,
$\mathit{pre}$ ranges over monotone set transformers in $\set{\epre_0,\eapre_1,\cpre, \elpre,\apre}$, and $\mu$ and $\nu$ denote, respectively, the least and the greatest fixed-point of the functional
defined as $X \mapsto \varphi(X)$.
Since the operations $\cup$, $\cap$, and the set transformers $\mathit{pre}$ are all monotonic, the fixed-points are guaranteed to exist.
A $\mu$-calculus formula evaluates to a set of states over $\game $, and the set can be computed
by induction over the structure of the formula, where the fixed-points are evaluated by iteration.
We omit the (standard) semantics of formulas (see \cite{Kozen83}).



\section{Fair Adversarial Rabin Games}
\label{sec:RabinGames}

This section presents the main result of this paper, which
is a symbolic fixpoint algorithm that computes the winning region of $\p{0}$ in the fair adversarial game over $\Gl$ with respect to  any $\omega$-regular property formalized as a Rabin winning condition. 

Our new fixpoint algorithm has multiple unique features. \\
\begin{inparaenum}[(I)]
\item It works directly over $\Gl$, without requiring any pre-processing step to reduce $\Gl$ to a ``normal'' two-player game. This feature allows us to obtain a direct symbolic algorithm for stochastic games as a by-product (see \REFsec{sec:stochastic}).\\
\item Conceptually, our symbolic algorithm is \emph{not} more complex than the known algorithm solving Rabin games over \enquote{normal} two-player game graphs by \citet{PitermanPnueli_RabinStreett_2006} (see \REFsec{sec:Rabin:Complexity}).\\
\item Our new fixpoint algorithm is obtained from the known algorithm of \citet{PitermanPnueli_RabinStreett_2006} by a simple syntactic change (as previewed in \eqref{equ:intro:Parity}). 
We simply replace all controllable predecessor operators over least fixpoint variables by the almost sure predecessor operator invoking the preceding maximal fixpoint variable. This makes the proof of our new fixpoint algorithm conceptually simple (see \REFsec{sec:Rabin:proofoutline}).
\end{inparaenum}

At a higher level, our syntactic change is a very simple yet efficient transformation to incorporate environment assumptions expressible by live edges into reactive synthesis while retaining computational efficiency. 
Most remarkably, this transformation also works \emph{directly} for fixpoint algorithms solving  reachability, safety,  B\"uchi, (generalized) co-B\"uchi, Rabin-chain and parity games, as these can be formalized as particular instances of a Rabin game (see \REFsec{sec:SimpleRabinGames}). 
Moreover, it also works for generalized Büchi and GR(1) games. However, as these games are particular instances of a \emph{generalized} Rabin game, 
we prove these special cases separately in \REFsec{sec:GenRabinGames} after formally introducing generalized Rabin games.

\subsection{The Symbolic Algorithm}
\label{subsec:new_algorithm}
   \noindent\textbf{Fair adversarial Rabin Games:}
   A \emph{Rabin} winning condition is defined by the set \linebreak $\FR = \set{\tuple{G_1,R_1},\hdots,\tuple{G_k,R_k}}$, where $G_i,R_i\subseteq V$ for all $i\in[1;k]$. 
   We say that $\FR$ has index set $P=[1;k]$.
   A play $\pi$ satisfies the \emph{Rabin condition} $\FR$ if $\pi$ satisfies the LTL formula
   \begin{align}\label{equ:vR}
      \vR\coloneqq\textstyle \bigvee_{i\in P}\left(\Diamond\Box\overline{R}_{i}\wedge \Box\Diamond G_{i}\right).
   \end{align}
We now present our new  symbolic fixpoint algorithm to compute the winning region of $\p{0}$ in the fair adversarial game over $\Gl$ with respect to a Rabin winning condition $\FR$. 

\begin{theorem}\label{thm:FPsoundcomplete}
Let $\Gl =\tup{\game, \El}$ be a game graph with live edges and 
$\FR$ be a Rabin condition over $\game$ with index set $P=[1;k]$. 
Further, let $Z^*$ denote the fixed-point of the following $\mu$-calculus formula:
 \begin{subequations}\allowdisplaybreaks
 \label{equ:Rabin_all}
   \begin{align}
 & 
\nu Y_{p_0}.\mu X_{p_0}. \bigcup_{p_1\in P} \nu Y_{p_1}.\mu X_{p_1}.\label{equ:Rabin_all_C}
 \bigcup_{p_2\in P\setminus\set{p_1}} \nu Y_{p_2}.\mu X_{p_2}.
\hdots
 \bigcup_{p_k\in P\setminus\set{p_1,\hdots, p_{k-1}}}\nu Y_{p_k}.\mu X_{p_k}.
\left[\bigcup_{j=0}^k \mathcal{C}_{p_j}\right],\\
&\quad\text{where}\quad
 \mathcal{C}_{p_j}\coloneqq
\left(\textstyle\bigcap_{i=0}^{j} \overline{R}_{p_i}\right)\cap
 \left[ 
 \left( G_{p_j}\cap \cpre(Y_{p_j})\right)
 \cup \left(\apre(Y_{p_j},X_{p_j})\right)
 \right],\label{equ:Rabin_all_Cpj}
\end{align}
 \end{subequations}
with\footnote{The Rabin pair $\tuple{G_{p_0},R_{p_0}}=\tuple{\emptyset,\emptyset}$ in \eqref{equ:Rabin_all} is artificially introduced to 
make the fixpoint representation more compact. 
It is not part of $\FR$.} $p_0=0$, $G_{p_0}\coloneqq\emptyset$ and $R_{p_0}\coloneqq\emptyset$. 
Then $Z^*$ is equivalent to the winning region $\WlR$ of $\p{0}$ in the fair adversarial game over~$\Gl$ for the Rabin winning condition $\FR$.  
Moreover, the fixpoint algorithm runs in $O(n^{k+2}k!)$ symbolic steps, and a memoryless winning strategy for $\p{0}$ can be extracted from it.
\end{theorem}	

\subsection{Proof Outline}\label{sec:Rabin:proofoutline}
Given a Rabin winning condition over a \enquote{normal} two-player game, 
\citet{PitermanPnueli_RabinStreett_2006} provided a symbolic fixpoint algorithm which computes the winning region for $\p{0}$. 
The fixpoint algorithm in their paper is almost identical to our fixpoint algorithm in  \eqref{equ:Rabin_all}:
it only differs in the last term of the constructed $\mathcal{C}$-terms in \eqref{equ:Rabin_all_Cpj}. \citet{PitermanPnueli_RabinStreett_2006} 
define the term $\mathcal{C}_{p_j}$ as 
\begin{equation*}
\left(\textstyle\bigcap_{i=0}^{j} \overline{R}_{p_i}\right)\cap
 \left[ 
 \left( G_{p_j}\cap \cpre(Y_{p_j})\right)
 \cup \left(\cpre(X_{p_j})\right)
 \right].
\end{equation*}
Intuitively, a single term $\mathcal{C}_{p_j}$ computes the set of states that always remain within $Q_{p_j}:=\bigcap_{i=0}^{j} \overline{R}_{p_i}$ while always re-visiting $G_{p_j}$. 
I.e, given the simpler (local) winning condition 
\begin{equation}\label{equ:SR:psi}
 \psi:=\Box Q\wedge \Box\Diamond G
\end{equation}
for two sets $Q,G\subseteq V$, the set 
\begin{equation}\label{equ:SR:FP:orig}
 \nu Y.~\mu X.~Q\cap
 \left[ 
 \left( G\cap \cpre(Y)\right)
 \cup \left(\cpre(X)\right)
 \right]
\end{equation}
is known to define exactly the states of a \enquote{normal} two-player game $\game$ from which $\p{0}$ has a strategy to win the game with winning condition $\psi$ \cite{MPS95}. Such winning conditions are typically called \emph{Safe B\"{u}chi winning conditions}, written as $\tup{G,Q}$. 
The key insight in the proof of \REFthm{thm:FPsoundcomplete} is to show that the new definition of $\mathcal{C}$-terms in \eqref{equ:Rabin_all_Cpj} using the new \emph{almost sure predecessor operator $\apre$} actually computes the winning state sets of \emph{fair adversarial} safe B\"{u}chi games.
Subsequently, we generalize this intuition to the fixpoint for the Rabin games.

\smallskip
\noindent\textbf{Fair Adversarial Safe B\"{u}chi Games:}
Solution of a fair adversarial safe B\"{u}chi game is formalized in the following theorem.

\begin{theorem}\label{thm:SingleRabin}
 Let $\Gl =\tup{\game, \El}$ be a game graph with live edges and $\tup{G,Q}$ be a safe B\"uchi winning condition. 
 Further, let
\begin{equation}\label{equ:SR:FP}
 Z^*\coloneqq\nu Y.~\mu X.~Q\cap
 \left[ 
 \left( G\cap \cpre(Y)\right)
 \cup \left(\apre(Y,X)\right)
 \right].
\end{equation}
Then $Z^*$ is equivalent to the winning region of $\p{0}$ in the fair adversarial safe B\"uchi game over~$\Gl$ with the winning condition $\tup{G,Q}$.  
Moreover, the fixpoint algorithm runs in $O(n^2)$ symbolic steps, and a memoryless winning strategy for $\p{0}$ can be extracted from it.
\end{theorem}

Intuitively, the fixpoint algorithms in \eqref{equ:SR:FP:orig} and \eqref{equ:SR:FP} consist of two parts: 
\begin{inparaenum}[(a)]
 \item a smallest fixpoint over $X$ which computes (for any fixed value of $Y$) the set of states that can \emph{reach} the \enquote{target state set} $T\coloneqq Q\cap G\cap \cpre(Y)$ while staying inside the safe set $Q$, and  
 \item a greatest fixpoint over $Y$ which ensures that the only states considered in the target $T$ are those 
that allow to re-visit a state in $T$ while staying in~$Q$.
\end{inparaenum}

By comparing \eqref{equ:SR:FP:orig} and \eqref{equ:SR:FP} we see that our syntactic transformation only changes part (a). Hence, in order to prove \REFthm{thm:SingleRabin} it essentially remains to show that this transformation works for the even simpler \emph{safe reachability games}.

\smallskip
\noindent\textbf{Fair Adversarial Safe Reachability Games:}
A safe reachability condition is a tuple $\tuple{T,Q}$ with $T,Q\subseteq V$ and a play $\pi$ satisfies the \emph{safe reachability condition } $\tuple{T,Q}$ if $\pi$ satisfies the LTL formula
  \begin{equation}\label{equ:Reach:psi}
 \psi:=Q\,\mathcal{U}\,T.
\end{equation}
A safe reachability game is often called a \emph{reach-avoid} game, where the safe sets are specified by an unsafe set $R:=\overline{Q}$ that needs to be avoided. 
The solution to fair adversarial reach-avoid games is formalized in the following theorem, and is proved in \REFapp{app:prop:Reachability}.

\begin{theorem}\label{thm:Reachability}
  Let $\Gl =\tup{\game, \El}$ be a game graph with live edges and $\tuple{T,Q}$ be a safe reachability winning condition. Further, let
\begin{equation}\label{equ:Reach:FP}
 Z^*\coloneqq\nu Y.~\mu X.~T \cup (Q\cap\apre(Y,X)).
\end{equation}
Then $Z^*$ is equivalent to the winning region of $\p{0}$ in the fair adversarial safe reachability game over $\Gl$ with the winning condition $\tup{T,Q}$. 
Moreover, the fixpoint algorithm runs in $O(n^2)$ symbolic steps, and a memoryless winning strategy for $\p{0}$ can be extracted from it.
\end{theorem}

To gain some intuition on the correctness of \REFthm{thm:Reachability}, let us recall that the fixpoint algorithm for safe reachability games \emph{without} live edges is given by: 
\begin{equation}\label{eq:normal reach fp}
\mu X.~T \cup (Q\cap\cpre(X)).
\end{equation}
Intuitively, the fixpoint in \eqref{eq:normal reach fp} is initialized with $X^0=\emptyset$ and computes a sequence $X^0,X^1,\hdots,X^k$ of increasingly larger sets until $X^k=X^{k+1}$. We say that $v$ has rank $r$ if $v\in X^r\setminus X^{r-1}$. All states contained in $X^r$ allow \p{0} to force the play to reach $T$ in at most $r-1$ steps while staying in~$Q$. 
The corresponding \p{0} strategy $\rho_0$ is known to be winning w.r.t.\ \eqref{equ:Reach:psi}, and 
along every play $\pi$ compliant with $\rho_0$, the path $\pi$ remains in $Q$ and the rank is always decreasing. 

To see why the same strategy is also \emph{sound} in the \emph{fair adversarial} safe reachability game~$\Gl$, first recall that for vertices $v\notin \Vl$ of $\Gl$, the almost sure pre-operator $\apre(X,Y)$ simplifies to $\cpre(X)$. 
With this, we see that for every $v\notin \Vl$ a $\p{0}$ winning strategy $\widetilde{\rho}_0$ in $\Gl$ can always force plays to stay in $Q$ and to decrease their rank, similar to $\rho_0$.
With this, we see that plays $\pi$ which are compliant with such a strategy $\widetilde{\rho}_0$ and visit a vertex in $\Vl$ only finitely often satisfy~\eqref{equ:Reach:psi}.

The only interesting case for soundness of \REFthm{thm:Reachability} are therefore plays $\pi$ that visits states in $\Vl$ infinitely often. However, as the number of vertices is finite, we only have a finite number of ranks and hence a certain vertex $v\in\Vl$ with a finite rank $r$ needs to get visited by $\pi$ infinitely often. Due to the definition of $\apre$ we however know that only states $v\in\Vl$ are contained in~$X^r$ if $v$ has an outgoing \emph{live} edge reaching $X^k$ with $k<r$. With this, reaching $v$ infinitely often implies that also a state with rank $k$ s.t.\ $k<r$ will get visited infinitely often. As $X^1=T$ we can show by induction that $T$ is eventually visited along $\pi$ while $\pi$ always remains in $Q$ until then.

In order to prove \emph{completeness} of \REFthm{thm:Reachability} we need to show that all states in $V\setminus Z^*$ are loosing for Player 0. Here, again the reasoning is equivalent to the \enquote{normal} safe reachability game with $v\notin\Vl$. For vertices $v\in\Vl$, we see that $v$ is not added to $Z^*$ via $\apre$ if $v\notin T$ and either (i) all its outgoing live transitions do not make progress towards $T$, or (ii) it has some outgoing edge (not necessarily a live one) that makes it leave $Z^*$. One can therefore construct a \p{1} strategy that for (i)-vertices always chooses a live transition and thereby never makes progress towards $T$ (also if $v$ is visited infinitely often), and for (ii)-vertices ensures that they are only visited once on plays which remain in $Q$. This ensures that (ii)-vertices never make progress towards $T$ via their possibly existing rank-decreasing live edges.

A detailed soundness and completeness proof of \REFthm{thm:Reachability} along with the respective \p{0} and \p{1} strategy construction is provided in \REFapp{app:prop:Reachability}. In addition, \REFthm{thm:SingleRabin} is proven in \REFsec{app:prop:SingleRabin} by a reduction to  \REFthm{thm:Reachability} for every iteration over $Y$.

\begin{figure}
\begin{center}
 \includegraphics[width=0.5\linewidth]{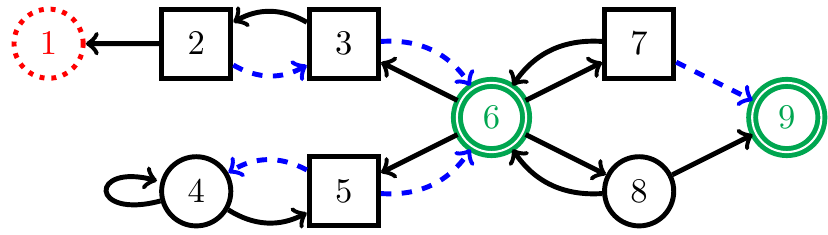}
\end{center}
\vspace{-0.4cm}
\caption{Fair adversarial game graph discussed in Examples~\ref{example_simple} and \ref{example_simple_buechi} with vertex sets $G=\{6,9\}$ (double circled, green), $\overline{Q}=\{1\}$ (red,dotted), and live edges $\El=\set{(2,3),(3,6),(5,4),(5,6),(7,9))}$ (dashed, blue). $\p{0}$ and $\p{1}$ vertices are indicated by circles and squares, respectively. }
\label{fig:example_simple}
\end{figure}

\begin{example}[Fair adversarial safe reachability game]\label{example_simple}
We consider a fair adversarial safe reachability game over the game graph depicted in \REFfig{fig:example_simple} with target vertex set $T=G=\{6,9\}$ and safe vertex set $Q=V\setminus \{1\}$. 

We denote by $Y^m$ the $m$-th iteration over the fixpoint variable $Y$ in \eqref{equ:Reach:FP}, where $Y^0=V$. Further, we denote by $X^{mi}$ the set computed in the $i$-th iteration over the fixpoint variable $X$ in \eqref{equ:Reach:FP} during the computation of  $Y^m$ where  $X^{m0}=\emptyset$.
We further have $X^{m1}=T=\{6,9\}$ as $\apre(\cdot,\emptyset)=\emptyset$. Now we compute
\begin{align}
 X^{12}&=T \cup (Q\cap\apre(Y^0,X^{11}))\notag\\
 &=\{6,9\}\cup (V\setminus \{1\}\cap[\underbrace{\cpre(X^{11})}_{\{8\}}\cup\underbrace{(\elpre(X^{11})\cap\eapre_1(V))}_{\{3,5,7\}}])=\{5,6,7,8,9\}\label{ex:aaa}
\end{align}
We observe that the only vertex added to $X$ via the $\cpre $ term is vertex $8$.  States $\{3,5,7\}$ are added due to the existing live edge leading to a target vertex. Here, we note that vertex $7$ is added due to its live edge to vertex $9$. The additional requirement  $\eapre_1(V)$ in $\apre(Y^0,X^{11})$ is trivially satisfied for all vertices at this point as $Y^0=V$ and can therefore be ignored.
Doing one more iteration over $X$ we see that now vertex $4$ gets added via the $\cpre{}$ term (as it is a $\p{0}$ vertex that allows progress towards $5$) and vertex $2$ is added via the $\apre{}$ term (as it allows progress to $3$ via a live edge). 
The iteration over $X$ terminates with $Y^1=X^{1*}=V\setminus \{1\}$. 

Re-iterating over $X$ for $Y^1$ gives $X^{22}=X^{12}=\{5,6,7,8,9\}$ as before. However, now vertex~$2$ does not get added to $X^{23}$ because vertex $2$ has an edge leading to $V\setminus Y^1=\{1\}$. Therefore the iteration over $X$ terminates with $Y^2=X^{2*}=V\setminus \{1,2\}$. 
When we now re-iterate over~$X$ for $Y^2$ we see that vertex $3$ is not added to $X^{32}$ any more, as vertex $3$ has a transition to $V\setminus Y^2=\{1,2\}$. Therefore the iteration over $X$ now terminates with $Y^3=X^{3*}=V\setminus \{1,2,3\}$.
Now re-iterating over~$X$ does not change the vertex set anymore and the fixpoint terminates with $Y^*=Y^3=V\setminus \{1,2,3\}$.

We note that the $\mu$-calculus formula \eqref{eq:normal reach fp} for \enquote{normal} safe reachability games terminates after two iterations over $X$ with $X^*=\{6,8,9\}$, as vertex $8$ is the only vertex added via the $\cpre{}$ operator in \eqref{ex:aaa}. Due to the stricter notion of $\cpre{}$ requiring that \emph{all} outgoing edges of $\p{0}$ vertices make process towards the target, \eqref{eq:normal reach fp} does not require an outer largest fixed-point over~$Y$ to \enquote{trap} the play in a set of vertices which allow progress when \enquote{waiting long enough.} This \enquote{trapping} required in \eqref{equ:Reach:FP} via the outer fixed-point over $Y$ actually fails for vertices~$2$ and~$3$ (as they are excluded form the winning set of \eqref{equ:Reach:FP}). Here, $\p{1}$ can enforce to \enquote{escape} to the unsafe vertex $1$ in two steps before $2$ and $3$ are visited infinitely often (which would imply progress towards $6$ via the existing live edges).

We see that the winning region in the \enquote{normal} game is significantly smaller than the winning region for the fair adversarial game, as adding live transitions restricts the strategy choices of $\p{1}$, making it easier for $\p{0}$ to win the game.
\end{example}

\begin{example}[Fair adversarial safe Büchi game]\label{example_simple_buechi}
 We now consider a fair adversarial safe Büchi game over the game graph depicted in \REFfig{fig:example_simple} with sets $G=\{6,9\}$ and $Q=V\setminus \{1\}$. 
 
 We first observe that we can rewrite the fixpoint in \eqref{equ:SR:FP} as
 \begin{align}
 \nu Y.~\mu X.~\left[Q\cap G\cap \cpre(Y)\right] \cup
 \left[ Q \cap \left(\apre(Y,X)\right)\right].\label{equ:SR:FP:rewrite}
 \end{align}
 Using \eqref{equ:SR:FP:rewrite} we see that for $Y^0=V$ we can define $T^0:=Q\cap G\cap \cpre(V)=G=\{6,9\}$. Therefore the first iteration over $X$ is equivalent to \eqref{ex:aaa} and terminates with $Y^1=X^{1*}=V\setminus \{1\}$. 
 
 Now, however, we need to re-compute $T$ for the next iteration over $X$ and obtain 
  $T^1=Q\cap G\cap \cpre(Y^1)=V\setminus\{1\}\cap\{6,9\}\cap V\setminus\{1,2,9\}=\{6\}$.
 This re-computation of $T^1$ checks which target vertices are re-reachable, as required by the Büchi condition. As vertex $9$ has no outgoing edge it is trivially not re-reachable.
 
With this, we see that for the next iteration over $X$ we only have one target vertex $T^1=\{6\}$. If we recall that vertex $7$ is added to $X^{22}$ due to its live edge to $9$, we see that it is now not added anymore. Intuitively, we have to exclude $7$ as $\p{1}$ can always decide to take the live edge towards $9$ from $7$ (also if $7$ only gets visited once), and therefore prevents to re-visit a target state.

Now, vertices $2$ and $3$ get eliminated for the same reason as in the safe reachability game within the second and third iteration over $Y$. The overall fixpoint computation therefore terminates with $Y^*=Y^3=\{4,5,6,8\}$. 
\end{example}

\smallskip
\begin{proof}[Proof of \REFthm{thm:FPsoundcomplete}]
  With \REFthm{thm:Reachability} and \REFthm{thm:SingleRabin} in place, the proof of\\ \REFthm{thm:FPsoundcomplete} is essentially equivalent to the proof of \citet{PitermanPnueli_RabinStreett_2006} while utilizing \REFthm{thm:Reachability} and \REFthm{thm:SingleRabin} at all suitable places. For completeness, we give the full proof of \REFthm{thm:FPsoundcomplete}, including the memoryless strategy construction, in \REFapp{appD}. 
In addition, we illustrate the steps of the fixpoint algorithm in \eqref{equ:Rabin_all} with a simple fair adversarial Rabin game (depicted in \REFfig{fig:example}) which has two acceptance pairs in \REFapp{app:ExpFP}. 
\end{proof}

\begin{remark}\label{rem:apre}
We remark that the fixpoint \eqref{equ:Reach:FP}, as well as the $\apre$ operator, are similar in structure to the solution of
almost surely winning states in concurrent reachability games \cite{dAHK98,dAH00,ChatterjeeAH11}.
In concurrent games, the fixed-point captures the largest set of states in which the game can be trapped while maintaining
a positive probability of reaching the target.
In our case, the fixed-point captures the largest set of states in which $\p{0}$ can keep the game while ensuring a visit to the target
either directly or through the live edges.
The commonality justifies our notation and terminology for $\apre$.

However, concurrent games are fundamentally different from fair adversarial games. 
In concurrent games, the two players simultaneously and independently choose their actions from a given vertex, and the next vertex
is chosen probabilistically (given the current vertex and the choice of actions).
It is known that optimal winning strategies in concurrent games may require randomization.
The randomization in strategies induces progress conditions similar to our live edges. 
In contrast, in fair adversarial games, the live edges are given as an assumption on the environment and are fixed once and for all, that is,
the set of live edges cannot be modified based on particular strategies of the players. 
To see the difference from concurrent games, consider co-B\"uchi winning conditions.
Almost sure winning regions for co-B\"uchi concurrent games can be characterized as fixpoints \cite{dAH00}; however, the
characterization requires an additional predecessor operator.
The additional operator provides a ``dual'' of live edges, whereby a player can ensure that some edges are taken finitely often in the long run.
Again, the choice of these edges is based on the strategies chosen by the players.
Thus, fixpoint algorithms for co-B\"uchi (and also Rabin) concurrent games are quite different from fair adversarial games, and both the reasons for their 
correctness and constructions of optimal strategies are more intricate. 
\end{remark}

\begin{remark}
\citet{AminofBK04} studied fair CTL and LTL model checking where the fairness condition is given by a transition fairness 
with \emph{all} edges of the transition system live. 
They show that CTL model checking under this all-live fairness condition, 
can be syntactically transformed to \emph{non-fair} CTL model checking.
A similar transformation is possible for fair model checking of B\"uchi, Rabin, and Streett formulas. 
The correctness of their transformation is based on reasoning similar to our $\apre$ operator.
For example, a state satisfies the CTL formula $\forall\Diamond p$ under fairness iff all paths starting from
the state either eventually visits $p$ or always visits states from which a visit to $p$ is possible.
\end{remark}

\subsection{Complexity}\label{sec:Rabin:Complexity}

\smallskip
\noindent\textbf{Complexity Analysis of \eqref{equ:Rabin_all}:} 
For Rabin games with $k$ Rabin pairs, \citet{PitermanPnueli_RabinStreett_2006} show a fixpoint formula with alternation depth $2k+1$ .
Using the accelerated fixpoint computation technique of \citet{long1994improved}, they deduce a bound of $O(n^{k+1}k!)$ symbolic steps. We show in \REFapp{app:acceleration} that this accelerated fixpoint computation can also be applied to \eqref{equ:Rabin_all} yielding a bound of $O(n^{k+2} k!)$ symbolic steps. 
(The additional complexity is because of an additional outermost $\nu$-fixpoint.)
Thus our algorithm is almost as efficient as the original algorithm for Rabin games without environment 
assumptions---\emph{independent} of the number of strong transition fairness assumptions!

\smallskip
\noindent\textbf{Comparison with a Na\"ive Solution:} 
We show a na\"ive reduction from fair adversarial Rabin games to usual Rabin games. 
Suppose $\Gl = \tup{\game,\El}$ is a game graph with live edges, $\FR = \set{\tuple{G_1,R_1},\hdots,\tuple{G_k,R_k}}$ is a Rabin winning condition defined over $\Gl$, and $\varphi$ is the corresponding LTL specification as defined in \eqref{equ:vR}.
Let $\widehat{\game} = \tup{\widehat{V},\widehat{V}_0,\widehat{V}_1,\widehat{E}}$ be a game graph obtained by just replacing every live edge of $\Gl$ with a gadget shown in \REFfig{fig:fair adversarial to normal rabin} and explained next.
For every live edge $(v,v')\in \El$ we introduce a new intermediate vertex named $vv'\in \widehat{V}$, and without loss of generality we assume that $vv'\in \widehat{V}_0$.
(We could have equivalently used the convention that $vv'\in \widehat{V}_1$.)
Then we replace the edge $(v,v')$ with a pair of new edges $(v,vv')\in \widehat{E}$ and $(vv',v')\in \widehat{E}$; the rest remains the same as in $\game$.
Assuming that $|\El|=l$ and $|V| = n$,
the number of vertices of~$\widehat{\game}$ is $n+l$.

Intuitively, the event of the newly introduced vertices being reached in $\widehat{\game}$ simulates the event of the corresponding live edge being taken in $\Gl$, and vice versa.
We are now ready to transfer the specification $\vl\rightarrow \varphi$ to a new Rabin winning condition  $\widehat{\FR}$ for $\widehat{\game}$.
First observe that $\vl\rightarrow \varphi$ is equivalent to $\lnot \vl \vee \varphi$, and $\lnot \vl$ can be expressed in LTL as
	$\bigvee_{(v,v')\in \El} (\square\lozenge \set{v} \wedge \lozenge\square \overline{\set{vv'}}$),
and is therefore equivalent to the Rabin winning condition $\FR^\ell \coloneqq \set{\tup{\set{v},\set{vv'}} \mid (v,v')\in \El}$.
Since Rabin winning conditions are closed under union, we obtain the new Rabin condition $\widehat{\FR}\coloneqq \FR \cup \FR^\ell$.

Once $\widehat{\game}$ and $\widehat{\FR}$ are obtained, one can use the fixpoint algorithm of \citet{PitermanPnueli_RabinStreett_2006} for ``normal'' two-player Rabin games. 
This whole process yields a symbolic algorithm for fair adversarial Rabin games with $2(k+l)+1$ alternations of fixpoint operators on a set of $(n+l)$ vertices that runs in time $O((n+l)^{k+l+1}(k+l)!)$.
In contrast, our main theorem shows that we get a symbolic fixpoint expression with $2(k+1)$ alternations that runs in $O(n^{k+2}k!)$ symbolic steps.
In many applications, we expect $l = \Theta(n)$, for which our algorithm is significantly faster. 

\begin{figure}
	\begin{tikzpicture}[xscale=0.8]
	\begin{footnotesize}
		\node[mysquare,draw]		(A)		at	(0,0)		{$v$};
		\node[mystate]		(B)		[right=of A]		{$v'$};
		\path[->]	(A)	edge		(B);
		
		\node	(Arr)		[right=0.45\ThCScm of B]		{\Huge $\Rightarrow$};
		
		\node[mysquare,draw]		(C)		[right=0.45\ThCScm of Arr]		{$v$};
		\node[mystate]		(D)		[right=of C]				{$vv'$};
		\node[mystate]		(E)		[right=of D]				{$v'$};
		\path[->]	(C)	edge		(D)
						(D)	edge		(E);
	\end{footnotesize}
	\end{tikzpicture}
	\caption{Left: A live edge $(v,v')$ in $\Gl$. Right: The gadget used to replace $(v,v')$ in $\widehat{\game}$. The vertex named $vv'$ is a newly added vertex in $\widehat{\game}$; $v$ belongs to $\widehat{V}_1$, $vv'$ belongs to $\widehat{V}_0$, but $v'$ may belong to either $\widehat{V}_0$ or $\widehat{V}_1$.}
	\label{fig:fair adversarial to normal rabin}
\end{figure}
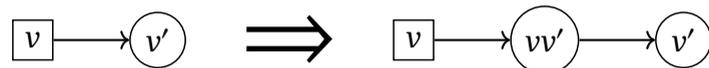 
\begin{cfigure}
\begin{center}
\begin{tikzpicture}[auto,scale=1.5]
  \begin{footnotesize}
          \node[mysquare,draw] (a) at (0,0) {$a$};
          \node[mystate] (b) at (1,0) {$b$};
          \node[mysquare,draw] (c) at (2,0) {$c$};
          \node[mystate] (d) at (3,0) {$d$};
          \SFSAutomatEdge{a}{}{b}{bend left}{}
          \SFSAutomatEdge{b}{}{a}{bend left}{}
          \SFSAutomatEdge{b}{}{c}{bend left}{}
          \SFSAutomatEdge{c}{}{b}{bend left}{}
          \SFSAutomatEdge{c}{}{d}{}{}
          \SFSAutomatEdge{d}{}{d}{loop right}{}
          
  \end{footnotesize}
\end{tikzpicture}
\end{center}
\caption{Counterexample to the equality of strong transition fairness and strong fairness (compassion).}

\label{fig:counterexample}
\end{cfigure}
\begin{remark}
As already mentioned in the introduction, not all strong fairness assumptions (Streett assumptions) can be translated into live edges (see e.g., \cite[p.264]{baierbook}). 
 As an example, consider the two-player game graph depicted in \REFfig{fig:counterexample}. $\p{0}$ and $\p{1}$ vertices are indicated by a circle and a box, respectively. Now consider the following one-pair Streett assumption
\begin{equation}\label{equ:Streett:ce}
 \varphi_A\coloneqq \Box\Diamond\set{a,b,c}\rightarrow \Box\Diamond \set{a}= 
 \Diamond\Box\set{d} \vee \Box\Diamond \set{a}.
\end{equation}
This fairness assumption states that it is not possible for a game to infinitely stay inside the set $\set{a,b,c}$ if $\p{0}$ decides to not transition from $b$ to $a$ anymore from some point onward. 
We see that we cannot model this behavior by a fair edge leaving a $\p{1}$ (square) state. 
If we mark the edge $(c,d)$ live, any fair play will transition to $d$ no matter if $a$ is visited infinitely often or not. 
Let us call this fair edge assumption $\vl_A$. Then we see that $\vl_A\rightarrow \varphi_A$ but not vice versa. 
\end{remark}



\subsection{Specialized Rabin Games}
\label{sec:SimpleRabinGames}

This section shows that the known fixpoint algorithms for Rabin chain, parity, and generalized co-B\"{u}chi winning conditions allow for the same \enquote{syntactic transfomation} as in the Rabin case to get the right algorithm for their fair adversarial version. We prove these claims by reducing the fixpoint algorithm in \eqref{equ:Rabin_all} to the special cases induced by the aforementioned winning conditions.  

We note that the fixpoint algorithm for fair adversarial Rabin games in \eqref{equ:Rabin_all} reduces to the normal fixpoint for Rabin games if $\El=\emptyset$. Therefore, our reductions of \eqref{equ:Rabin_all} to fixpoint algorithms for other winning conditions also proves these reductions in the usual case. We are not aware of such reductions proved elsewhere in the literature. 

\smallskip
\noindent\textbf{Fair Adversarial Rabin Chain Games:}
   A \emph{Rabin chain} winning condition \cite{Mostowski84} is a \emph{Rabin} condition $\FR = \set{\tuple{G_1,R_1},\hdots,\allowbreak\tuple{G_k,R_k}}$, with the additional \emph{chain condition}
 \begin{align}\label{equ:RCprop_1}
 &R_1\supseteq R_2\supseteq\hdots\supseteq R_k~&&\text{and}~
 &&G_1\supseteq G_2 \supseteq \hdots\supseteq G_k.
 \end{align}
Intuitively, the fixpoint algorithm computing $Z^*$ in \eqref{equ:Rabin_all} simplifies to a single permutation sequence, namely $p_1=k$, $p_2=k-1$, $\hdots$, $p_k=1$, if \eqref{equ:RCprop_1} holds. This is formalized in the following theorem which is proved in \REFapp{app:RabinChain}.

\begin{theorem}\label{thm:RabinC_all}
Let $\Gl =\tup{\game, \El}$ be a game graph with live edges and 
$\FR$ be a Rabin chain winning condition over $\game$ with $k$ pairs. 
Further, let 
   \begin{subequations}\label{equ:RabinC_all}
   \begin{align}
 \textstyle Z^*\coloneqq&\nu Y_{0}.~\mu X_{0}.~\nu Y_{k}.~\mu X_{k}.~\nu Y_{k-1}.~\hdots\mu X_{1}.~\textstyle\bigcup_{j=0}^k \Ct_{j},\label{equ:Rabin_all_FP}\\
\text{where }\quad\Ct_{j}\coloneqq&
\overline{R}_j\cap
 \left[ 
 \left( G_{j}\cap \cpre(Y_{j})\right)
 \cup \apre(Y_{j},X_{j})
 \right]\label{equ:RabinC_all_Cpj}
\end{align}
with $G_{p_0}\coloneqq\emptyset$ and $R_{p_0}\coloneqq\emptyset$. 
Then $Z^*$ is equivalent to the winning region $\WlR$ of $\p{0}$ in the fair adversarial Rabin chain game over $\Gl$ for the winning condition $\FR$. 
Moreover, the fixpoint algorithm runs in $O(n^{k+2})$ symbolic steps, and a memoryless winning strategy for $\p{0}$ can be extracted from it.
 \end{subequations}
 \end{theorem}

\smallskip
   \noindent\textbf{Fair Adversarial Parity Games:}
 A \emph{parity} winning condition \cite{EmersonJutla89} is defined by a set $\FP=\set{C_1,C_2,\hdots C_{2k}}$ of colors, where each $C_i\subseteq V$ is the set of vertices of $\game$ with color $i$. 
   Further, $\FP$ partitions the state space, i.e., $\bigcup_{i\in[1;2k]}C_i=V$ and $C_i\cap C_j=\emptyset$ for all $i,j\in[1;2k]$ with $i\neq j$. 
   A play $\pi$ satisfies the \emph{parity condition} $\FP$ if $\pi$ satisfies the LTL formula 
   \begin{equation}\label{equ:vP}
	\vP \coloneqq  \textstyle\bigwedge_{i\in[1;k]}\left(\square\lozenge C_{2i-1} \rightarrow \bigvee_{j\in [i;k]} \square\lozenge C_{2j}\right).
\end{equation}
That is, the maximal color visited infinitely often along $\pi$ is even.
A parity winning condition $\FP$ with $2k$ colors corresponds to the Rabin chain winning condition 
\begin{align}\label{equ:Fcond_}
 &\set{\tuple{F_{2},F_{3}},\hdots,\tuple{F_{2k},\emptyset}}\quad\quad\text{s.t.}~F_{i}:=\textstyle\bigcup_{j=i}^{2k} C_{j},
\end{align}
which has $k$ pairs. Due to $\FP$ forming a partition of the state space  one can further simplify the Rabin chain fixpoint algorithm in \eqref{equ:RabinC_all}. 
Interestingly, the resulting fixpoint looks slightly different from the one we would obtain by mechanically applying our syntactic transformation. While the usual fixpoint algorithm for parity games is given as
\begin{align}\label{equ:Parity_normal}
  \textstyle Z^*\coloneqq&\nu Y_{2k}.~\mu X_{2k-1}\ldots \nu Y_{2}.~\mu X_1 . \\
                &\quad (C_1 \cap \cpre(X_1)) \cup (C_2 \cap \cpre(Y_2)) \cup 
                (C_3\cap \cpre(X_3)) \ldots \cup  (C_{2k} \cap \cpre(Y_{2k} )),\notag
\end{align}
the fixpoint algorithm for fair adversarial parity games, formalized in the following theorem, looks slightly different. 

   \begin{theorem}\label{thm:Parity_all}
Let $\Gl =\tup{\game, \El}$ be a game graph with live edges and 
$\FP$ be a parity condition over $\game$ with $2k$ colors. 
Further, let 
\begin{align}\label{equ:Parity_live}
 \textstyle Z^*:=&\nu Y_{2k}.~\mu X_{2k-1}.\hdots\nu Y_{2}.~\mu X_{1}.\\
  &~\cup(C_{2k}\cap \cpre(Y_{2k}))\cup ((C_1\cup\hdots\cup C_{2k-1})\cap \apre(Y_{2k},X_{2k-1}))\notag\\
   &~\cup~\hdots\notag\\
   &~\cup(C_{4}\cap \cpre(Y_{4}))\cup ((C_1\cup C_2\cup C_{3})\cap \apre(Y_4,X_{3}))\notag\\
 &~\cup (C_2\cap \cpre(Y_2))\cup (C_1\cap \apre(Y_2,X_1))\notag
\end{align}
Then $Z^*$ is equivalent to the winning region $\WlR$ of $\p{0}$ in the fair adversarial parity game over~$\Gl$ with the set of colors~$\FP$. 
Moreover, the fixpoint algorithm runs in $O(n^{k+1})$ symbolic steps, and a memoryless winning strategy for $\p{0}$ can be extracted from it.
\end{theorem}

The intuition why the union of all colors $C_1 \hdots C_{2k-1}$ are intersected with $\apre(Y_{2k},X_{2k-1})$ in \eqref{equ:Parity_live} (in comparison to only the matching odd color $C_{2k-1}$ being intersected with $\cpre(X_{2k-1})$ in \eqref{equ:Parity_normal}) can be illustrated via the example in \REFfig{fig:counterexample_parity}. Here, the names of the vertices coincide with their color and we see that \p{0} wins as every path visits vertex $1$ infinitely often which implies that $\p{1}$ has to take the (dashed) live edge infinitely often, resulting in the maximum color seen infinitely often to be even (i.e., 4). We see that in order to infer that color $4$ is seen infinitely often whenever color $3$ is seen infinitely often, we need to understand that a lower color vertex (i.e., vertex $1$) enforces visits to vertex $4$ via its live edge. If $C_1$ would not be intersected with the $\apre(Y_4,X_3)$ term of the fixpoint algorithm, this conclusion cannot be made. The same reasoning applies if the color of the \p{1} vertex is $2$ in \REFfig{fig:counterexample_parity}, which shows that also lower even color vertex sets need to be intersected with the respective $\apre{}$ term.
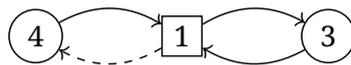
\begin{figure}
\begin{center}
\begin{tikzpicture}[auto,scale=1.5]
  \begin{footnotesize}
          \node[mystate] (a) at (0,0) {$4$};
          \node[mysquare,draw] (b) at (1,0) {$1$};
          \node[mystate] (c) at (2,0) {$3$};
          \SFSAutomatEdge{a}{}{b}{bend left}{}
          \SFSAutomatEdge{b}{}{a}{bend left,dashed}{}
          \SFSAutomatEdge{b}{}{c}{bend left}{}
          \SFSAutomatEdge{c}{}{b}{bend left}{}
  \end{footnotesize}
\end{tikzpicture}
\end{center}
\caption{Counterexample to the simple syntactic transformation for Parity games. The name of the vertex indicates its color. }
\label{fig:counterexample_parity}
\end{figure}

\smallskip
   \noindent\textbf{Fair Adversarial (Generalized) Co-B\"{u}chi Games:}
   A \emph{co-B\"{u}chi} winning condition is defined by a subset $A\subseteq V$ of vertices of $\game$. 
   A play $\pi$ satisfies the \emph{co-B\"{u}chi condition} $A$ if $\pi$ satisfies 
   \begin{equation}\label{equ:vCB}
	\textstyle\varphi\coloneqq  \lozenge\square A.
  \end{equation}
   A \emph{generalized co-B\"{u}chi} winning condition is defined by a set  $\mathcal{A}=\set{A_1,\hdots A_r}$, where each $A_i\subseteq V$ is a subset of  vertices of $\game$. 
   A play $\pi$ satisfies the \emph{generalized co-B\"{u}chi condition} $\mathcal{A}$ if $\pi$ satisfies 
   \begin{equation}\label{equ:vGCB}
	\textstyle\varphi\coloneqq \bigvee_{a\in[1;r]} \lozenge\square A_a.
  \end{equation}
     Generalized co-B\"{u}chi winning conditions correspond to a Rabin condition $\FR$ with $r$ pairs s.t.\
 \begin{align} \label{equ:Rtilde}
\AllQ{j\in[1;r]}{&R_j\coloneqq \overline{A}_{j}\quad\text{and}\quad G_j\coloneqq V}.
 \end{align}
Intuitively, the fact that $G_j\coloneqq V$ for all $j$ leads to a cancellation of all $\apre$ terms in $\mathcal{C}_j$ and all terms become ordered, i.e., we have $\mathcal{C}_{p_{j+1}}\subseteq\mathcal{C}_{p_j}$ for every permutation sequence used in \eqref{equ:Rabin_all}. 
As we take the union over all $\mathcal{C}_{p_j}$-s in \eqref{equ:Rabin_all_C}, the term $\mathcal{C}_{p_1}$ absorbs all others for every permutation sequence.  Hence, for every permutation sequence we only have two terms left, one for $j=0$ (over the artificially introduced Rabin pairs $G_{p_0}=R_{p_0}=\emptyset$) and one for the first choice $p_1$ made in this particular permutation. This is formalized in the following theorem which is proved in \REFapp{app:CoBuechi}.

\begin{theorem}\label{thm:GCB}
Let $\Gl =\tup{\game, \El}$ be a game graph with live edges and 
$\mathcal{A}$ be a generalized co-B\"{u}chi winning condition $\game$ with $r$ pairs. 
Further, let 
\begin{align}\label{equ:GCB}
  \textstyle Z^*\coloneqq&\nu Y_{0}.~\mu X_{0}.~\bigcup_{a\in [1;r]} \nu Y_{a}.~
\apre(Y_{0},X_{0})\cup (A_a\cap \cpre(Y_a)).
\end{align}
Then $Z^*$ is equivalent to the winning region $\WlR$ of $\p{0}$ in the fair adversarial generalized co-B\"uchi game over $\Gl$ for the winning condition $\mathcal{A}$. Moreover, 
the fixpoint algorithm runs in $O(rn^2)$ symbolic steps, and a memoryless winning strategy for $\p{0}$ can be extracted from it.
\end{theorem}


\section{Generalized Rabin Games}
\label{sec:GenRabinGames}

In this section, we slightly generalize our main result, \REFthm{thm:FPsoundcomplete}, to fair adversarial \emph{generalized} Rabin games. That is, for each Rabin pair, we allow the goal set $G_i$ to be a set of goal sets $\Gc_j=\set{\ps{1}G_j,\hdots, \ps{m_j}G_j}$. Then a play fulfills the winning condition if there exists one generalized Rabin pair $\tuple{\Gc_i,R_i}$ such that the play eventually remains in $\overline{R}_i$ and visits \emph{all} sets $\ps{l}G_i$ infinitely often.

The motivation of this generalization is to show that our syntactic transformation also works for fair adversarial games with a \emph{generalized reactivity winning condition of rank $1$} (GR(1) games for short) \cite{piterman2006synthesis}. \emph{Generalized} Rabin games allow us to see a GR(1) winning condition as a particularly simple instantiation of a Rabin game as shown in \REFsec{subsec:GR1_reduction}.

\subsection{Fair Adversarial Generalized Rabin Games}\label{sec:GenRabin}

   \noindent\textbf{Generalized Rabin Conditions:}
    A \emph{generalized Rabin condition} is defined by a set \linebreak $\FgR = \set{\tuple{\Gc_1,R_1},\hdots,\tuple{\Gc_k,R_k}}$ where each $\Gc_j=\set{\ps{1}G_j,\hdots, \ps{m_j}G_j}$ is a finite set s.t.\ $\ps{l}G_j\subseteq V$ for all $j\in[1;k]$ and all $l\in[1;m_j]$. We say that $\FgR$ has global index set $P=[1;k]$. A play $\pi$ satisfies the \emph{generalized Rabin condition} $\FgR$ if $\pi$ satisfies the LTL formula
\begin{align}\label{equ:vgR}
\varphi:=\textstyle\bigvee_{j\in P}\left(\Diamond\Box\overline{R}_{j}\wedge \bigwedge_{l\in[1;m_j]}\Box\Diamond \ps{l}G_{j}\right).
\end{align}

Recalling the discussion of \REFsec{subsec:new_algorithm}, we know that the proof of \REFthm{thm:FPsoundcomplete} fundamentally relies on the correctness of our transformation for safe Büchi (\REFthm{thm:SingleRabin}) and safe reachability (\REFthm{thm:Reachability}) games. 
Similarly, one needs to prove correctness of our syntactic transformation for safe \emph{generalized}  Büchi games in the case of generalized Rabin games.

\smallskip
\noindent\textbf{Safe Generalized Büchi Games}
A \emph{safe generalized Büchi condition} is defined by a tuple $\tuple{\Fc,Q}$ where $Q\subseteq V$ is a set of safe states and $\mathcal{F}=\set{\ps{1}F,\hdots,\ps{s}F}$ is a set of goal sets. A play $\pi$ satisfies the \emph{safe generalized Büchi condition} $\tuple{\Fc,Q}$ if $\pi$ satisfies the LTL formula
\begin{align}\label{equ:sGB:psi}
\varphi:=\textstyle\Box Q \wedge \bigwedge_{l\in[1;s]}\Box\Diamond \ps{l}F.
\end{align}
Now we can apply our syntactic transformation to the usual fixpoint algorithm for solving safe generalized Büchi games 
and prove its correctness for all fair adversarial plays. This is formalized in the next theorem and proved in \REFapp{app:proof:sGB}.

\begin{theorem}\label{thm:sGB}
Let $\Gl =\tup{\game, \El}$ be a game graph with live edges, and $\tuple{\Fc,Q}$ with $\mathcal{F}=\set{\ps{1}F,\hdots,\ps{s}F}$ be a safe generalized Büchi winning condition.
Further, let 
\begin{align}\label{equ:sGB:FP}
  \textstyle Z^*\coloneqq&\nu Y.\bigcap_{b\in [1;s]} \mu \ps{b}X.~
  Q\cap 
  \left[
(\ps{b}F\cap  \cpre(Y))\cup \apre(Y,\ps{b}X)
\right].
\end{align}
Then $Z^*$ is equivalent to the winning region $\WlR$ of $\p{0}$ in the fair adversarial safe generalized B\"uchi game over $\Gl$ for the winning condition $\tuple{\Fc,Q}$. 
Moreover, the fixpoint algorithm runs in $O(sn^2)$ symbolic steps, and a finite-memory winning strategy for $\p{0}$ can be extracted from it.
\end{theorem}

Intuitively, the proof of \REFthm{thm:sGB} reduces to \REFthm{thm:SingleRabin} in a similar manner as the proof of \REFthm{thm:SingleRabin} reduces to \REFthm{thm:Reachability}. However, the challenge in proving \REFthm{thm:sGB} is to show that it is indeed sound to use the fixpoint variable $Y$ which is actually the intersection of fixpoint variables $X$ both within $\cpre$ and $\apre$. The proof of this correctness essentially requires to show that upon termination we have $Y^*=\ps{b}X^*$ for all $b\in [1;s]$ (see \REFapp{app:proof:sGB} for a formal proof).

\smallskip
   \noindent\textbf{The Symbolic Algorithm:}
By knowing that \eqref{equ:sGB:FP} allows to correctly solve safe generalized Büchi games, we can immediately generalize this observation to Rabin games. This is formalized in the following theorem which is an immediate consequence of \REFthm{thm:FPsoundcomplete} and \REFthm{thm:sGB}.
   
\begin{theorem}\label{thm:GenRabin}
Let $\Gl =\tup{\game, \El}$ be a game graph with live edges and 
$\FgR$ be a generalized Rabin condition over $\game$ with index set $P=[1;k]$. 
Further, let
 \begin{subequations}\allowdisplaybreaks
 \label{equ:GenRabin_all}
   \begin{align}
    &Z^*:=\nu Y_{0}.~\mu X_{0}.~
         \bigcup_{p_1\in P}  \nu Y_{p_1}.~\bigcap_{l_1\in[1;m_{p_1}]}\mu \ps{l_1}X_{p_1}.~\hdots\hdots
\bigcup_{p_k\in P\setminus\set{p_1,\hdots, p_{k-1}}}\hspace{-0.3cm}\nu Y_{p_k}.~\bigcap_{l_k\in[1;m_{p_k}]}\hspace{-0.3cm}\mu \ps{l_k}X_{p_k}.~
 \bigcup_{j=0}^k \ps{l_j}\mathcal{C}_{p_j},\label{equ:GenRabin_all_C}\\
&\quad\text{where}\quad\quad
 \ps{l_j}\mathcal{C}_{p_j}  := 
\left(\textstyle\bigcap_{i=0}^{j} \overline{R}_{p_i}\right)\cap
 \left[ 
 \left( \ps{l_j}G_{p_j}\cap \cpre(Y_{p_j})\right)
 \cup \apre(Y_{p_j},\ps{l_j}X_{p_j})
 \right]\label{equ:GenRabin_all_Cpj}
\end{align}
 \end{subequations}
with\footnote{Again, the generalized Rabin pair $\tuple{\Gc_{p_0},R_{p_0}}$ in \eqref{equ:Rabin_all} is artificially introduced and not part of $\FgR$.} $p_0=0$, $G_{p_0}\coloneqq\set{\emptyset}$ and $R_{p_0}\coloneqq\emptyset$. 
Then $Z^*$ is equivalent to the winning region $\WlR$ of $\p{0}$ in the fair adversarial generalized Rabin game over $\Gl$ for the winning condition $\FgR$. Moreover, 
the fixpoint algorithm runs in $O(n^{k+2}k! m_{1}\ldots m_k)$ symbolic steps, and yields a finite-memory winning strategy for $\p{0}$.
\end{theorem}	

The proof of \REFthm{thm:GenRabin} is almost identical to the proof of \REFthm{thm:FPsoundcomplete} in \REFapp{appD}, when using \REFthm{thm:sGB} instead of \REFthm{thm:SingleRabin} in all appropriate places. This, yields a finite memory winning strategy by suitably \enquote{stacking} the individual finite-memory strategies constructed in the proof of \REFthm{thm:sGB}. (See \REFapp{app:GenRabinproof} for a complete proof of \REFthm{thm:GenRabin}.)

\subsection{Fair Adversarial Muller Games}
A Muller winning condition \cite{gradel2002automata} is defined by a set $\mathcal{F}=\{F_1,F_2,\hdots F_k\}$ and a play $\pi$ satisfies the Muller condition $\mathcal{F}$ if 
the set of vertices appearing infinitely often along $\pi$ is exactly $F_i$ for some $i\in \set{1,\ldots, k}$.
Equivalently, a play is winning if it satisfies
\begin{equation}
 \varphi\coloneqq \bigvee_{i\in[1;k]} \left(\Diamond\Box F_i~\wedge~\bigwedge_{q\in F_i} \Box\Diamond v  \right ).
\end{equation}
It is easy to see that a Muller winning condition can be written as the generalized Rabin winning condition 
$\FgR = \set{\tuple{\Gc_1,R_1},\hdots,\tuple{\Gc_k,R_k}}$ where $\Gc_i\coloneqq\SetComp{\set{v}}{v\in F_i}$ and $R_i\coloneqq\overline{F_i}$ for $i\in\set{1,\ldots, k}$. 
It therefore follows that fair adversarial Muller games can be solved via the fixpoint algorithm in \eqref{equ:GenRabin_all}.

\subsection{Fair Adversarial GR(1) Games}
\label{subsec:GR1_reduction}
Within this section, we show how fair adversarial Rabin games can be reduced to fair adversarial games with GR(1) winning conditions. 

\smallskip
   \noindent\textbf{GR(1) winning condition:}
A GR(1) winning condition is defined by two sets $\mathcal{A}=\Set{A_1,\hdots ,A_r}$ and  $\mathcal{F}=\Set{F_1,\hdots, F_s}$, where for every $i\in [1;r]$ and $j\in [1;s]$, $A_i,F_j\subseteq V$.
   A play $\pi$ satisfies the GR(1) condition $(\mathcal{A},\mathcal{F})$ if it satisfies the LTL formula
   \begin{align}\label{equ:vGRo}
    \vGRo\coloneqq &\textstyle\left(\bigwedge_{a\in[1;r]}\Box\Diamond A_{a}\right)
\rightarrow 
\left(\bigwedge_{b\in[1;s]}\Box\Diamond F_{b}\right)
=\textstyle\left(\bigvee_{a\in[1;r]}\Diamond\Box \overline{A}_{a}\right)
\vee    
\left(\bigwedge_{b\in[1;s]}\Box\Diamond F_{b}\right).
   \end{align}
By comparing $\vGRo$ in \eqref{equ:vGRo} with $\vR$ in \eqref{equ:vgR}, we see that a GR(1) condition $(\mathcal{A},\mathcal{F})$ can be transformed into a generalized Rabin condition $\FgR$ with $k=r+1$ pairs, such that
\begin{subequations}\label{equ:GR1toRabin}
 \begin{align}
\AllQ{j\in[1;r]}{&R_j\coloneqq A_{j}\quad\text{and}\quad \Gc_j\coloneqq \set{V}},\quad\text{and}\\
 &R_{k}\coloneqq \emptyset\quad\text{and}\quad  \Gc_{k}\coloneqq \mathcal{F}.\label{equ:GR1toRabin2}
\end{align}
\end{subequations}

\smallskip
   \noindent\textbf{Fixpoint Algorithm:}
We first observe that the first $r$ Rabin pairs with trivial goal sets actually correspond to a generalized co-Büchi condition (compare \eqref{equ:Rtilde}) which can be solved by the fixpoint in \REFthm{thm:GCB} (see \REFsec{sec:SimpleRabinGames}). Intuitively, the fixpoint in \REFthm{thm:GCB} only needs to consider single indices from $P=[1;r]$ rather then full permutation sequences as in \REFthm{thm:FPsoundcomplete}. By 
adding the last tuple $\tuple{\Gc_{k},R_{k}}$ to the winning condition, we essentially need to consider two indices in each conjunct of \eqref{equ:RabinC_all}, i.e., $p_j$ (with $j\in[1;r]$) and $p_k$. In principle, we would need to consider both possible orderings of these two indices (compare \eqref{equ:GenRabin_all}). However, by inspecting \eqref{equ:GR1toRabin} we see that the sets corresponding to these indices always fulfill a (generalized) chain condition (compare \eqref{equ:RCprop_1}). That is, we have $R_j\supseteq R_k$ and $V=\ps{1}G_j\supseteq \ps{b}F$ for any 
$j\in[1;r]$ and $b\in[1;s]$. Hence, we only need to consider the permutation sequence $p_kp_j$ (compare \eqref{equ:RabinC_all}). 
Using this insight, along with some additional simplifications, we indeed yield the fixpoint that we would obtain 
by simply applying our transformation to the well-known GR(1) fixpoint (compare e.g.,~\cite{piterman2006synthesis}). 
This observation is formalized in the next theorem and proved in \REFapp{app:GRoneproof}.

\begin{theorem}\label{thm:GR1}
Let $\Gl =\tup{\game, \El}$ be a game graph with live edges and $(\mathcal{A},\Fc)$ a GR(1) winning condition.
Further, let 
\begin{align}\label{equ:GR1:FP}
Z^*=&\nu Y_k.~\bigcap_{b\in[1;s]}\mu \ps{b}X_k.~\bigcup_{a\in [1;r]} \nu Y_{a}.
\quad(F_b\cap\cpre(Y_k))\cup \apre(Y_k,\ps{b}X_k)\cup (\overline{A}_a\cap \cpre(Y_a)).
\end{align}
Then $Z^*$ is equivalent to the winning region $\WlR$ of $\p{0}$ in the fair adversarial GR(1) game over~$\Gl$ for the winning condition $(\mathcal{A},\Fc)$. 
Moreover, the fixpoint algorithm runs in $O(n^2r s)$ symbolic steps, and a finite-memory winning strategy for $\p{0}$ can be extracted from it.
\end{theorem}

In particular, the strategy extraction is performed in the same way as by \citet{piterman2006synthesis} for a \enquote{normal} GR(1) game.

\begin{remark}
\citet{Belta17} presented a symbolic fixpoint algorithm for stochastic games 
(which can be modeled using fair adversarial games, see \REFsec{sec:stochastic}) with respect to GR(1) winning conditions.
While one can show that the output of their algorithm coincides with the output of our newly derived fixpoint algorithm in \eqref{equ:GR1:FP}, 
their algorithm is structurally more involved.
On a conceptual level, we feel our insight about simply \enquote{swapping} predecessor operators in the right manner is insightful even if one can also use their algorithm to find a solution to this problem.  
\end{remark}

\smallskip
   \noindent\textbf{Fair Adversarial vs. Environmentally-Friendly GR(1) Games:}
The idea of the simple \enquote{predecessor operator swapping trick} shares resemblance with environmentally-friendly GR(1) synthesis, proposed by \citet{MPS19}. 
There, the authors show a direct symbolic algorithm to compute $\p{0}$ strategies which do not win a given GR(1) game vacuously, by rendering the assumptions false. More precisely, given a synthesis game for the specification $\varphi\coloneqq (\varphi_A\rightarrow \varphi_G)$ with $\varphi_A$ and $\varphi_G$ being LTL formulas modeling respectively environment assumptions and system guarantees, $\p{0}$ can win by violating $\varphi_A$ and thereby satisfying~$\varphi$ vacuously. 
Environmentally-friendly synthesis rules out such undesired strategies by only computing so called non-conflicting winning strategies. 
Interestingly, the fixpoint algorithm introduced by \citet{MPS19} also swaps $\cpre{}$ and $\apre{}$ operators, but in a slightly different way.

The GR(1) fragment considered by \citet{MPS19} corresponds to a specification $\varphi_A\rightarrow \varphi_G$ where both $\varphi_A$ and $\varphi_G$ can be realized by a deterministic generalized B\"uchi automaton. 
Hence, they provide an algorithm to compute non-conflicting winning strategies in a deterministic generalized B\"uchi game under deterministic generalized B\"uchi assumptions. If the used deterministic B\"uchi assumptions can be translated into live edges over the same game graph, the resulting fair adversarial game is a generalized B\"uchi game (not a GR(1) game), solvable by the fixpoint in \eqref{equ:sGB:FP} for $Q=V$.
 
By reducing a GR(1) game to a fair adversarial game, one transforms the given assumption into one expressed by fair edges which cannot be falsified by $\p{0}$ and therefore yields a simpler algorithm to compute non-conflicting strategies. However, the direct relationship between deterministic generalized B\"uchi assumptions and live-edge assumptions is not known, i.e., we do not know if all environmentally-friendly GR(1) games can be reduced to fair adversarial generalized B\"uchi games.

Finally, we want to point out that fair adversarial GR(1) games compute winning strategies that are only non-conflicting with respect to  the environment assumptions encoded in the live edges. $\p{0}$ can still win a fair adversarial GR(1) game vacuously by falsifying $\varphi_A$, i.e., never visiting any set $A_i$ in $\mathcal{A}$ (see \eqref{equ:vGRo}) infinitely often. 

%

\section{Stochastic Generalized Rabin Games}
\label{sec:stochastic}

We present an important application of our fixpoint algorithm in solving stochastic two-player games, commonly known as $\twohalf$-player games.
$\twohalf$-player games form an important subclass of stochastic games, 
and have been studied quite extensively in the literature \cite{Condon92,DBLP:conf/icalp/ChatterjeeAH05,Zielonka04}.
They can be seen as a generalization of two-player games by additionally capturing the environmental randomness inside the game.
In order to do so, in addition to $\p{0}$ and $\p{1}$ vertices as in a two-player game, they include a new set of vertices called the \emph{random vertices}.
Whenever the game reaches a random vertex, one of the outgoing edges is picked uniformly at random.
$\p{0}$ is said to win a $\twohalf$-player game almost surely if she wins the game with probability~$1$; the respective $\p{0}$ strategy is called an almost sure winning strategy.
We only consider 
stochastic games with a uniform probability distribution over edges which originate from a random vertex. 
This is indeed without loss of generality since it is known that stochastic games with other probability distributions 
over random edges have exactly the same almost sure winning sets as $\twohalf$-player games \cite{DBLP:conf/icalp/ChatterjeeAH05}.

We present a reduction from the computation of almost sure winning strategies in $\twohalf$-player generalized Rabin games to the computation of winning strategies in fair adversarial generalized Rabin games.
This yields a direct symbolic algorithm for solving $\twohalf$-player generalized Rabin games.

\subsection{Preliminaries: \texorpdfstring{$\twohalf$}{2.5}-player games}\label{sec:twoandhalf prelims}
We introduce the basic setup of the $\twohalf$-player games.

\smallskip\noindent\textbf{The game graph:}
We consider usual $\twohalf$-player games played between $\p{0}$, $\p{1}$, and a third player representing \emph{environmental randomness}.
Formally, a $\twohalf$-player game graph is a tuple $\game = \tup{V, V_0,V_1,V_r,E}$ where
\begin{inparaenum}[(i)]
	\item $V$ is a finite set of vertices,
	\item $V_0$, $V_1$, and $V_r$ are subsets of~$V$ which form a partition of $V$, and
	\item $E\subseteq V\times V $ is the set of directed edges.
\end{inparaenum} 
The vertices in~$V_r$ are called \emph{random vertices}, and the edges originating in a random vertex are called \emph{random edges}.
The set of all random edges is denoted by $E_r\coloneqq E(V_r)$.

\smallskip\noindent
\textbf{Strategies and plays:}
We define strategies for $\p{0}$ and $\p{1}$ in exactly the same way as the strategies in two-player games.
While in principle, we could consider randomized strategies, it is known that optimal strategies for $\omega$-regular winning
conditions are pure \cite{DBLP:conf/icalp/ChatterjeeAH05}.
The new part is when the $\twohalf$-player game reaches a random vertex, the game chooses one of the random edges uniformly at random.
A play is, as usual, an infinite sequence of vertices $(v^0,v^1,\ldots)$ that satisfies the edge relation between two consecutive vertices in the sequence.
Due to the presence of random edges, given an initial vertex $v^0\in V$ and given a pair of strategies $\rho_0$ and~$\rho_1$ of $\p{0} $ and $\p{1}$ respectively, we will obtain a \emph{probability distribution over the set of plays}. 
We denote the set of strategies of $\p{0}$ and $\p{1}$ by $\Pi_0$ and $\Pi_1$, respectively.

\smallskip\noindent
\textbf{Almost sure winning:}
Let $\varphi$ be any $\omega$-regular specification over $V$.
Let us denote the event that the runs of a $\twohalf$-player game graph $\game$ satisfies $\varphi$ using the symbol $\game\models \varphi$
For a given initial vertex $v^0\in V$ and for a given pair of strategies $\rho_0$ and $\rho_1$ of $\p{0}$ and $\p{1}$, we denote the probability of the occurrence of the event $\game\models \varphi$ by $P_{v^0}^{\rho_0,\rho_1}(\game\models \varphi)$.
We define the set of almost sure winning states of $\p{0}$ for the specification $\varphi$ as the set of vertices $\mathcal{W}^{\mathit{a.s.}}\subseteq V$ such that for every $v\in \mathcal{W}^{\mathit{a.s.}}$,
\begin{equation}
	\textstyle\sup_{\rho_0\in \Pi_0}\inf_{\rho_1\in \Pi_1} P_v^{\rho_0,\rho_1}(\game\models \varphi) = 1.
\end{equation}

\subsection{The reduction}\label{sec:twoandhalf reduction}
Suppose $\game$ is a $\twohalf$-player game graph and $\FgR$ is a generalized Rabin winning condition.
To obtain the reduced two-player game graph, we simply reinterpret the random vertices as $\p{1}$ vertices and the random edges as live edges.
Let us first formalize this notion of the reduced game graph.

\begin{definition}[Reduction to two-player game with live edges]
Let $\game=\tup{V, V_0,V_1,V_r,E}$ be a $\twohalf$-player game graph.
Define $\Dr(\game) \coloneqq \tup{\tup{\widetilde{V},\widetilde{V}_0,\widetilde{V}_1,\widetilde{E}},\El}$ as follows:
\begin{itemize}
	\item $\widetilde{V} = V$,
	$\widetilde{V}_0 = V_0$,
	$\widetilde{V}_1 = V_1\cup V_r$,
	$\widetilde{E} = E$, and
	$\El = E_r$.
\end{itemize}
\end{definition}

It remains to show that the almost sure winning set of $\p{0}$ in $\game$ for the generalized Rabin winning condition $\FgR$ is the same as the winning set of $\p{0}$ in the fair adversarial game over $\Dr(\game)$ for the winning condition $\FgR$. This is formalized in the following theorem, which is proved in \REFapp{app:stoch}. 
The proof essentially shows that the random edges of $\game$ simulate the live edges of $\Dr(\game)$, and vice versa.

\begin{theorem}\label{thm:Reduction}
 Let $\game$ be a $\twohalf$-player game graph, $\FgR$ be a generalized Rabin condition, $\varphi \subseteq V^\omega$ be the corresponding LTL specification (Eq.~\eqref{equ:vgR}) over the set of vertices $V$ of $\game$, and $\Dr(\game)$ be the reduced two-player game graph.
 Let $\WlR\subseteq \widetilde{V}$ be the set of all the vertices from where $\p{0}$ wins the fair adversarial game over $\Dr(\game)$ for the winning condition $\varphi$, and $\mathcal{W}^{\mathit{a.s.}}$ be the almost sure winning set of $\p{0}$ in the game graph $\game$ for the specification $\varphi$.
 Then, $\WlR = \mathcal{W}^{\mathit{a.s.}}$.
 Moreover, a winning strategy in $\Dr(\game)$ is also a winning strategy in $\game$, and vice versa.
\end{theorem}

The above theorem generalizes \cite[Thmeorem 11.1]{GH19_Survey} from liveness properties to all LTL specifications on $\twohalf$-player games.
Together with our symbolic algorithm for fair adversarial Rabin games, the reduction implies a $O(n^{k+2}k!)$ algorithm
for stochastic Rabin games for a game with~$n$ vertices and $k$ Rabin pairs.
This improves the previous best algorithm from \cite{DBLP:conf/icalp/ChatterjeeAH05}, which reduces the problem to a normal two-player game by replacing every random vertex using a gadget with $O(k)$ vertices; similar gadgets are used to reduce other classes of stochastic games to their non-stochastic counterparts as well \cite{chatterjee2003simple,chatterjee2019combinations}.
The resulting two-player Rabin game has $O(n(k+1))$ vertices and $k+1$ Rabin conditions.
Plugging in the complexity of Rabin games, the resulting complexity is $O\left((n(k+1))^{k+2}(k+1)!\right)$.

\begin{remark}
The idea underlying this section is to replace random edges with live edges to compute almost sure winning states. We recall again that probabilistic choice is different from (i.e., stronger than) strong transition fairness studied in our paper. See \REFsec{subsec:FAG} for an illustrative example in \REFfig{fig:simple}.
\end{remark}


\section{Experimental Evaluation}
\label{sec:experiments}

We have developed a C++-based tool \fairsyn, which implements 
the symbolic fair adversarial Rabin fixpoint from Eq.~\eqref{equ:Rabin_all} using BDDs.
We developed two versions of \fairsyn: 
A single-threaded version using the (single-threaded) CUDD library \cite{somenzi2019cudd}, and 
a multi-threaded version using the (multi-threaded) Sylvan library \cite{van2015sylvan}. 

Our tool implements a well-known acceleration technique for fixpoint computations \cite{long1994improved}. 
It exploits certain monotonicity properties of the fixpoint variables, and ``warm-starts'' the inner fixpoint iterations by initializing 
them with earlier computed values for similar configurations of the leading fixpoint variables' iteration indices 
(see \REFapp{app:acceleration} for a formal explanation).
The acceleration procedure trades memory for time; it can avoid computations
if all the intermediate values of the fixpoint variables for all possible configurations of the fixpoint iteration indices are stored.
In practice, this creates an inordinate amount of overhead on the memory requirement:
The original algorithm would already run out of memory when solving the smallest instance of the case study reported in Table~\ref{tab:lego results} (first line) on a computer with $1.5$ TB of memory.
We have therefore adapted the acceleration technique to achieve a novel (space-)bounded acceleration algorithm that we utilize within \fairsyn.
Our new algorithm takes an \emph{acceleration parameter} $M$ as input, which bounds the extent to which intermediate values of fixpoint variables are cached (see \REFapp{app:acceleration} for details). 
Whenever no cached value is available during the computation, our algorithm falls back to the default way of initializing fixpoint variables and re-computations.

To show the effectiveness of our proposed symbolic algorithm for fair adversarial Rabin games, we performed various experiments with \fairsyn which fall into two different categories. 
First, in \REFsec{sec:vlts}, we demonstrate the merits of utilizing parallelization and acceleration within \fairsyn.
Second, in \REFsec{sec:casestudies}, we show the practical relevance of our algorithm by 
solving two large practical case-studies stemming from the areas of software engineering and control systems. 

The experiments in \REFsec{sec:vlts} and \REFsec{sec:resource management} were performed using \textsf{Sylvan}-based \fairsyn on a computer 
equipped with a $3$~GHz Intel Xeon E7 v2 processor with $48$ CPU cores and \num{1.5}\,\si{\tebi\byte} RAM.
The experiments in \REFsec{sec:stoch experiment} were performed using CUDD-based \fairsyn on a Macbook Pro (2015) laptop equipped with a $2.7$~GHz Dual-Core Intel Core i5 processor with \num{16}\,\si{\gibi\byte} RAM.

\subsection{Performance Evaluation}\label{sec:vlts}

This section discusses a benchmark suite used to  empirically evaluate the merits of the two important aspects of \fairsyn, namely the parallelization and the acceleration.
Our benchmark suite is build on transition systems taken from the Very Large Transition Systems (VLTS) benchmark suite \cite{vlts_benchmark}.
For each chosen transition system, we randomly generated benchmark instances of fair adversarial Rabin games with up to $3$ Rabin pairs. To transform a given transition systems into a fair adversarial Rabin game, we labeled (i) $50\%$ of randomly chosen vertices as system vertices, (ii) the remaining vertices as environment vertices, (iii) up to $5\%$ of randomly selected environment edges as live edges, and (iv) for every set in $\FR = \set{\tuple{G_1,R_1},\hdots,\tuple{G_k,R_k}}$ we randomly selected up to $5\%$ of all vertices to be contained. 
We have summarized the relevant details of all the randomly generated instances of the fair adversarial Rabin games in \REFtab{tab:vlts benchmark data 1} and \REFtab{tab:vlts benchmark data 2} in \REFapp{app:experiments}.
In these examples, the number of vertices were \num{289}--\num[group-separator={,}]{566639}, the number of BDD variables were \num{9}--\num{20}, the number of transitions were \num[group-separator={,}]{1224}--\num[group-separator={,}]{3984160}, and number of live edges were \num{1}--\num[group-separator={,}]{42757}.
For all benchmark instances with more than $4$ live edges, the na\"ive version of \fairsyn which treats live edges as Streett conditions and transforms them into additional Rabin pairs as discussed in \REFsec{sec:Rabin:Complexity}, did not terminate after $2$ hours.

\smallskip
\noindent\textbf{Merits of parallelization.}
We ran \fairsyn on $10$ different benchmark instances with $1$ or $2$ Rabin pairs, and varied the number of parallel worker threads used in \fairsyn between $1$--$48$, while keeping the acceleration enabled.
The left scatter plot in \REFfig{fig:BCG scatter plot} plots the computation times with $48$ threads (parallel) versus the computation times with $1$ thread (non-parallel).
Observe that in almost all the experiments, the parallelized version outperforms the non-parallelized version (points above the solid red line).
In addition, in many cases the speedup achieved due to the parallelization was more than one order of magnitude (points above the dashed red line).

A more fine-grained analysis of the benefits of parallelization is shown in \REFfig{fig:BCG trend plots}.(a). Here computation time (in logarithmic scale) is plotted over the number of worker threads used.
We observe that the saving due to parallelization is more significant for the curves lying in the top half which correspond to larger examples. This is due to the better utilization of the available pool of worker threads by the larger examples.

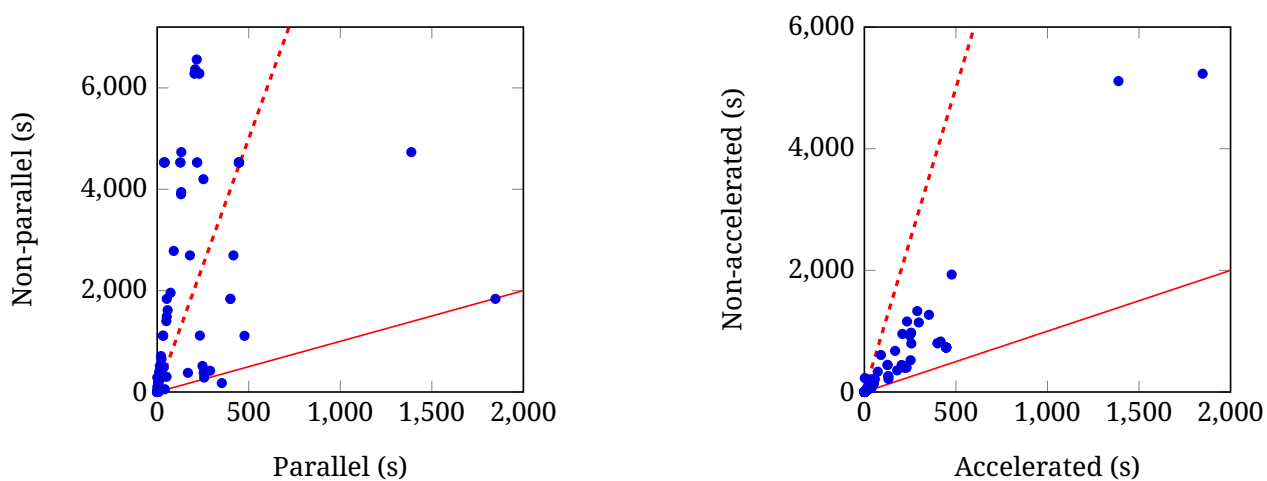
\begin{figure}
	\begin{tikzpicture}[scale=0.6]
\begin{axis}[xmin=0,ymin=0,xmax=2000,ymax=7200,xlabel={Parallel\ (\si{\second})}, ylabel={Non-parallel\ (\si{\second})},
					width=5\ThCScm,height=5\ThCScm,mark size=1.25\ThCSpt]


\draw[color=red]		(0,0)	--	(2000,2000);
\draw[color=red, dashed,  thick] 	(0,0)	--	(800,8000);

\addplot+[only marks]	table[x=acc_+_par,y=acc_+_npar] {DATA/par_vs_npar.txt};

\end{axis}

\end{tikzpicture}\qquad\qquad
	\begin{tikzpicture}[scale=0.6]
\begin{axis}[xmin=0,ymin=0,xmax=2000,ymax=6000,xlabel={Accelerated\ (\si{\second})}, ylabel={Non-accelerated\ (\si{\second})},
					width=5\ThCScm,height=5\ThCScm,mark size=1.25\ThCSpt]


\draw[color=red]		(0,0)	--	(2000,2000);
\draw[color=red, dashed,  thick] 	(0,0)	--	(700,7000);

\addplot+[only marks]	table[x=acc_+_par,y=nacc_+_par] {DATA/accl_vs_naccl.txt};

\end{axis}

\end{tikzpicture}
	\caption{
		(Left) Comparison between the computation times for the non-parallel (1 worker thread)\\ and parallel (48 worker threads) version of \fairsyn, with acceleration being enabled in both cases.
		(Right) Comparison between the computation times for the non-accelerated and the accelerated version of \fairsyn, with parallelization being enabled in both cases.
		(Both) The points on the solid red line represent the same computation time. The points on the dashed red line represent an order of magnitude improvement.
	}
	\label{fig:BCG scatter plot}
\end{figure}

\smallskip
\noindent\textbf{Merits of acceleration.}
We ran \fairsyn on $10$ different benchmark instances with $1$--$3$ Rabin pairs, and varied the acceleration parameter $M$ between $2$--$15$, while the number of worker threads was fixed to $48$.
The right scatter plot in \REFfig{fig:BCG scatter plot} plots the computation times with $M=15$ versus the computation times with no acceleration.
Observe that in almost all the experiments, the accelerated version outperformed the non-accelerated version (points above the solid red line), and in many cases the achieved speedup is close to an order of magnitude (points near the dashed red line). See \REFfig{fig:BCG scatter plot zoomed-in} in \REFapp{app:experiments} for a zoomed-in version of \REFfig{fig:BCG scatter plot}.

A more fine-grained analysis of the benefits of acceleration is shown in \REFfig{fig:BCG trend plots}.(b)--(e). Here we have plotted the total computation time (Plots~(b),(d)) and the initialization time (Plots~(c),(e)) in logarithmic scale over $M$ for benchmark instances with $2$ Rabin pairs (Plots~(b),(c)) and $3$ Rabin pairs (Plots~(d),(e)). Plots for instances with $1$ Rabin pair can be found in \REFfig{fig:BCG trend plots rp 1} in \REFapp{app:experiments}.

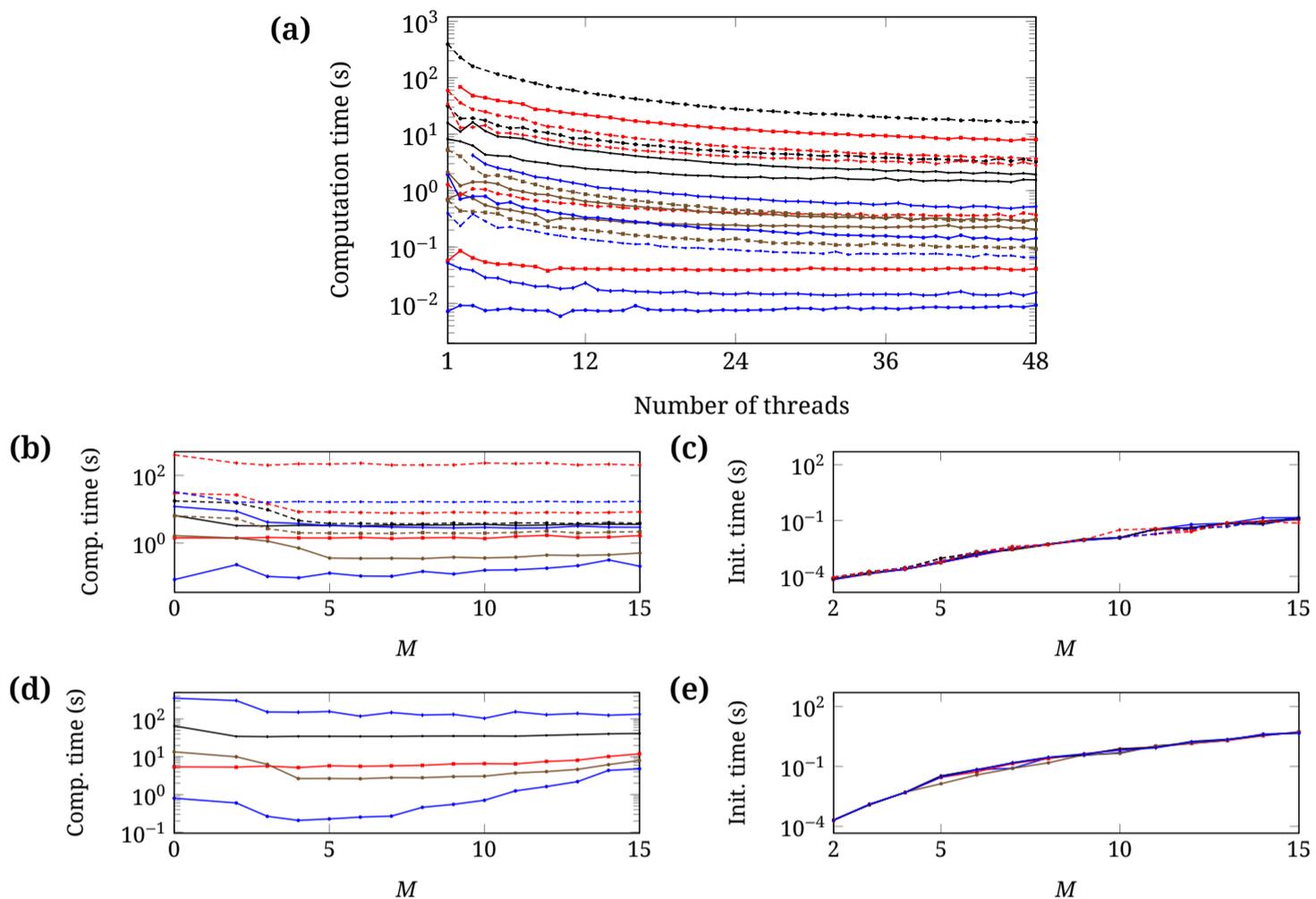
\begin{figure} 
%
\begin{tikzpicture}[scale=0.6]
\node at (-3,6) {\textbf{(a)}};
\begin{axis}[xmin=1,xmax=48,
	ymode=log,
	xtick={1,12,24,36,48},
    xlabel={Number of threads},
    ylabel={Computation time (\si{\second})},
    width=8\ThCScm,height=5\ThCScm,
    mark options={scale=0.25}
    ]
  \addplot table [x=nthreads,y=c11] {DATA/vary_nof_threads.txt};
  \addplot table [x=nthreads,y=c12] {DATA/vary_nof_threads.txt};
   \addplot table [x=nthreads,y=c21] {DATA/vary_nof_threads.txt};
  \addplot table [x=nthreads,y=c22] {DATA/vary_nof_threads.txt};
   \addplot table [x=nthreads,y=c31] {DATA/vary_nof_threads.txt};
  \addplot table [x=nthreads,y=c32] {DATA/vary_nof_threads.txt};
   \addplot table [x=nthreads,y=c41] {DATA/vary_nof_threads.txt};
  \addplot table [x=nthreads,y=c42] {DATA/vary_nof_threads.txt};
   \addplot table [x=nthreads,y=c51] {DATA/vary_nof_threads.txt};
  \addplot table [x=nthreads,y=c52] {DATA/vary_nof_threads.txt};
   \addplot table [x=nthreads,y=c61] {DATA/vary_nof_threads.txt};
  \addplot table [x=nthreads,y=c62] {DATA/vary_nof_threads.txt};
   \addplot table [x=nthreads,y=c71] {DATA/vary_nof_threads.txt};
  \addplot table [x=nthreads,y=c72] {DATA/vary_nof_threads.txt};
   \addplot table [x=nthreads,y=c81] {DATA/vary_nof_threads.txt};
  \addplot table [x=nthreads,y=c82] {DATA/vary_nof_threads.txt};
   \addplot table [x=nthreads,y=c91] {DATA/vary_nof_threads.txt};
  \addplot table [x=nthreads,y=c92] {DATA/vary_nof_threads.txt};
\end{axis}

\end{tikzpicture}
\begin{tikzpicture}[scale=0.55]
\node at  (-3,3) {\textbf{(b)}};
\begin{axis}[
	xmin=0,xmax=15,ymax=500,
	xtick={0,5,10,15},
	ymode=log,
    xlabel={$M$},
    ylabel={Comp.\ time (\si{\second})},
    width=0.45\textwidth,height=3\ThCScm,
    mark options={scale=0.25}
    ]
 	 \addplot table [x=M,y=c1] {DATA/vary_mem_rp_2.txt};
	\addplot table [x=M,y=c2] {DATA/vary_mem_rp_2.txt};
	\addplot table [x=M,y=c3] {DATA/vary_mem_rp_2.txt};
	\addplot table [x=M,y=c4] {DATA/vary_mem_rp_2.txt};
	\addplot table [x=M,y=c5] {DATA/vary_mem_rp_2.txt};
	\addplot table [x=M,y=c6] {DATA/vary_mem_rp_2.txt};
	\addplot table [x=M,y=c7] {DATA/vary_mem_rp_2.txt};
	\addplot table [x=M,y=c8] {DATA/vary_mem_rp_2.txt};
	\addplot table [x=M,y=c9] {DATA/vary_mem_rp_2.txt};
	\addplot table [x=M,y=c10] {DATA/vary_mem_rp_2.txt};
\end{axis}

\end{tikzpicture}
%
\begin{tikzpicture}[scale=0.55]
\node at  (-3,3) {\textbf{(c)}};
\begin{axis}[
	xmin=2,xmax=15,ymax=500,
	xtick={2,5,10,15},
	ymode=log,
	xlabel={$M$},
    ylabel={Init.\ time (\si{\second})},
    width=0.45\textwidth,height=3\ThCScm,
    mark options={scale=0.25}
    ]
 	 \addplot table [x=M,y=c1] {DATA/vary_mem_init_rp_2.txt};
	\addplot table [x=M,y=c2] {DATA/vary_mem_init_rp_2.txt};
	\addplot table [x=M,y=c3] {DATA/vary_mem_init_rp_2.txt};
	\addplot table [x=M,y=c4] {DATA/vary_mem_init_rp_2.txt};
	\addplot table [x=M,y=c5] {DATA/vary_mem_init_rp_2.txt};
	\addplot table [x=M,y=c6] {DATA/vary_mem_init_rp_2.txt};
	\addplot table [x=M,y=c7] {DATA/vary_mem_init_rp_2.txt};
	\addplot table [x=M,y=c8] {DATA/vary_mem_init_rp_2.txt};
	\addplot table [x=M,y=c9] {DATA/vary_mem_init_rp_2.txt};
	\addplot table [x=M,y=c10] {DATA/vary_mem_init_rp_2.txt};
\end{axis}

\end{tikzpicture}

\begin{tikzpicture}[scale=0.55]
\node at  (-3,3) {\textbf{(d)}};
\begin{axis}[
	xmin=0,xmax=15,ymax=500,
	xtick={0,5,10,15},
	ymode=log,
	xlabel={$M$},
    ylabel={Comp.\ time (\si{\second})},
  	width=0.45\textwidth,height=3\ThCScm,
    mark options={scale=0.25}
    ]
 	 \addplot table [x=M,y=c1] {DATA/vary_mem_rp_3.txt};
	\addplot table [x=M,y=c2] {DATA/vary_mem_rp_3.txt};
	\addplot table [x=M,y=c3] {DATA/vary_mem_rp_3.txt};
	\addplot table [x=M,y=c4] {DATA/vary_mem_rp_3.txt};
	\addplot table [x=M,y=c5] {DATA/vary_mem_rp_3.txt};
\end{axis}

\end{tikzpicture}
%
\begin{tikzpicture}[scale=0.55]
\node at  (-3,3) {\textbf{(e)}};
\begin{axis}[
	xmin=2,xmax=15,ymax=500,
	xtick={2,5,10,15},
	ymode=log,
	xlabel={$M$},
    ylabel={Init.\ time (\si{\second})},
    width=0.45\textwidth,height=3\ThCScm,
    mark options={scale=0.25}
    ]
 	 \addplot table [x=M,y=c1] {DATA/vary_mem_init_rp_3.txt};
	\addplot table [x=M,y=c2] {DATA/vary_mem_init_rp_3.txt};
	\addplot table [x=M,y=c3] {DATA/vary_mem_init_rp_3.txt};
	\addplot table [x=M,y=c4] {DATA/vary_mem_init_rp_3.txt};
	\addplot table [x=M,y=c5] {DATA/vary_mem_init_rp_3.txt};
<	
\end{axis}

\end{tikzpicture}
	\caption{
		(a) Effect of parallelization on computation time, with the acceleration enabled.
		(b,d) Effect of variation of the acceleration parameter $M$ on the total computation time (parallelization being enabled) for $2$ and $3$ Rabin pairs respectively.
		(c,e) Effect of variation of the acceleration parameter $M$ on the initialization time for $2$ and $3$ Rabin pairs respectively.
		The computation time (Y-axis) is always shown in the logarithmic scale.
	}
	\label{fig:BCG trend plots}
\end{figure}

The plotted initialization time is needed by the accelerated algorithm for allocating memory to store intermediate fixpoint values.
We observe that this initialization time grows exponentially with $M$, which is due to the $\mathcal{O}(M^{k+1}k!)$ space complexity of the acceleration algorithm.
As a result, the computational savings due to the use of acceleration get undermined by the high initialization cost for large $M$. 
We note that, due to their random generation, the considered benchmark instances are not well structured. This results in low iteration numbers over involved fixpoint variables. Due to this, the allocated memory gets underutilized for large values of $M$. In the practically relevant examples discussed in \REFsec{sec:casestudies} the game graph is naturally structured, resulting in a large number of fixpoint iterations and thereby showing superior performance for larger values of $M$.

\subsection{Practical Benchmarks}\label{sec:casestudies}

This section shows that  \fairsyn is able to efficiently solve two practical case studies stemming from the areas of software engineering (\REFsec{sec:resource management}) and control systems (\REFsec{sec:stoch experiment}). 

\subsubsection{Code-Aware Resource Management}
\label{sec:resource management}
We consider a case study introduced by \citet{rupak_krishnendu}. It considers the problem of synthesizing a \emph{code-aware resource manager} for a network protocol, i.e., multi-threaded program running on a single CPU.
The task of the resource manager is to grant different threads access to different shared synchronization resources (mutexes and counting semaphores).
The specification is deadlock freedom across all threads at all time while assuming a \emph{fair scheduler} (scheduling every thread always eventually) and \emph{fair progress} in every thread (i.e., taking every existing execution branch always eventually).
By making the resource manager \emph{code-aware}, it can avoid deadlocks by utilizing its knowledge about the require and release characteristics of all treads for different resources.

\citet{rupak_krishnendu} showed that the problem of synthesizing a code-aware resource manager can be approximated using a $1\half$-player game\footnote{A $1\half$-player game is a $2\half$-player game without any $\p{1}$ vertices.} generated from the known require and release characteristics of all threads. 
We used \fairsyn to synthesize a code-aware resource manager for this problem, where the live edges model the aforementioned fairness conditions imposed on the scheduler and the threads.

%
\begin{sidefigure}
\begin{center}
 		\begin{tikzpicture}[scale=0.6] 
 		\begin{footnotesize}
			\node[draw, ellipse]	(a)	at	(0,0)	{generator};
			\node[draw, ellipse]	(b)	[below=of a, node distance=\ThCS0.2cm]	{sender};
			\node[draw, rectangle]	(c)	[right=of a]	{broadcast};
			\node[draw, rectangle]	(d)	at	(b -| c)	{output};
			\node[draw, ellipse]	(e)	at	($(c)!0.5!(d) + (3,0)$)	{delay};
			
			\path[->]	(a)	edge	(c)	
							edge	(d)
						(b)	edge	(c)
						(c)	edge	(e)
						(d)	edge	(b)
						(e)	edge	(d);
						
			\draw[->]	(b.south)	--	($(b)+(0,-1)$)	node[below]	{to network};
                    \end{footnotesize}
		\end{tikzpicture}
		\caption{Structure of network protocol.}
		\label{fig:network protocol}
\end{center}
\end{sidefigure}
Motivated by the case study conducted by \citet{rupak_krishnendu}, we consider a network protocol consisting of $3$ threads and $2$ queues of bounded capacity, as depicted in \REFfig{fig:network protocol}.
The threads (shown as oval-shaped nodes) are called \textit{generator}, \textit{sender}, and \textit{delay}, and the queues (shown as rectangular nodes) are called \textit{broadcast} and \textit{output}.
The generator generates data packets and dispatches them to either the broadcast queue or the output queue.
Packets from the broadcast queue are added to the output queue after a random delay, introduced by the delay thread.
The purpose of this delay is to avoid packet collisions during broadcasting.
The packets in the output queue are in transit and get processed by the sender process.
The sender process attempts to transmit packets from the output queue via the network, and when the transmission fails, it adds the respective data packet back to the broadcast queue, so that another transmission attempt can be made after a delay.
Access to all queues is protected by mutexes and semaphores.
Each queue has one mutex and two semaphores, one for counting the number of empty places and another for counting the number of packets present.

As discussed by \citet{rupak_krishnendu}, the outlined network protocol may deadlock when both queues are full, a transmission via sender fails, and the sender tries to insert the packet back to the broadcast queue. In this case, due to the output queue being full, the broadcast queue will not be able to make space for the incoming packet, leading to a deadlock situation.
The correct strategy for the resource manager to prevent this deadlock is to ensure that  the generator never adds packets to the broadcast queue if the output queue is full.

We used the parallel and accelerated version of \fairsyn with $M=15$ to automatically synthesize the resource manager for the outlined network protocol case study. Indeed, \fairsyn was successful in discovering the outlined managing strategy. To showcase \fairsyn's performance on this case study,
we report the number of vertices of the problem instance and \fairsyn's computation time to solve it for different queue capacities in \REFtab{tab:lego results}; an extended version of the table with more number of cases has been included in \REFtab{tab:resource management extended} in \REFapp{app:experiments}.
In all cases, \fairsyn was able to provide expected strategies within a reasonable amount of time.
Note that treating the live edges as Streett conditions would result in a game with several million Rabin pairs, making all these examples go far beyond the scope of any synthesis tool for Rabin games.

\begin{table}
	\begin{tabular}{
		|>{\centering\arraybackslash}m{1.6\ThCScm}|
		>{\centering\arraybackslash}m{1.4\ThCScm}|
		>{\centering\arraybackslash}m{2.2\ThCScm}|
		>{\centering\arraybackslash}m{2.2\ThCScm}|
		>{\centering\arraybackslash}m{2.2\ThCScm}|
		>{\centering\arraybackslash}m{1.5\ThCScm}|
		>{\centering\arraybackslash}m{1.5\ThCScm}|}
		 \hline
          Broadcast Queue Capacity	&	Output Queue Capacity	&	Number of Vertices	&	Number of Transitions	 &	Number of Live edges	&	Number of BDD variables	&	Time (seconds)	\\
		\hline
		1	&	1	&	\tablenum[group-separator={,},table-format=9.0]{5307840}	&	\tablenum[group-separator={,},table-format=9.0]{10135300}	&	\tablenum[group-separator={,},table-format=9.0]{5124100}	&	25	&	\tablenum[table-format=3.2]{7.38} \\
		2	&	1	&	\tablenum[group-separator={,},table-format=9.0]{21231400}	&	\tablenum[group-separator={,},table-format=9.0]{40541200}	&	\tablenum[group-separator={,},table-format=9.0]{20496400}	&	27	&	\tablenum[table-format=3.2]{24.90}	\\
		3	&	1	&	\tablenum[group-separator={,},table-format=9.0]{21414100}	&	\tablenum[group-separator={,},table-format=9.0]{42080300}	&	\tablenum[group-separator={,},table-format=9.0]{21265900}	&	27	&	\tablenum[table-format=3.2]{28.98}	\\
		1	&	2	&	\tablenum[group-separator={,},table-format=9.0]{21340800}	&	\tablenum[group-separator={,},table-format=9.0]{40879100}	&	\tablenum[group-separator={,},table-format=9.0]{20834300}	&	27	&	\tablenum[table-format=3.2]{38.26}	\\
		1	&	3	&	\tablenum[group-separator={,},table-format=9.0]{21559400}	&	\tablenum[group-separator={,},table-format=9.0]{42756100}	&	\tablenum[group-separator={,},table-format=9.0]{21772800}	&	27	&	\tablenum[table-format=3.2]{51.56}	\\
		2	&	2	&	\tablenum[group-separator={,},table-format=9.0]{85363200}	&	\tablenum[group-separator={,},table-format=9.0]{163516000}	&	\tablenum[group-separator={,},table-format=9.0]{83337200}	&	29	&	\tablenum[table-format=3.2]{133.20}	\\
		3	&	2	&	\tablenum[group-separator={,},table-format=9.0]{86061400}	&	\tablenum[group-separator={,},table-format=9.0]{169673000}	&	\tablenum[group-separator={,},table-format=9.0]{86415400}	&	29	&	\tablenum[table-format=3.2]{144.28}	\\
		2	&	3	&	\tablenum[group-separator={,},table-format=9.0]{86237400}	&	\tablenum[group-separator={,},table-format=9.0]{171024000}	&	\tablenum[group-separator={,},table-format=9.0]{87091200}	&	29	&	\tablenum[table-format=3.2]{163.62}	\\
		3	&	3	&	\tablenum[group-separator={,},table-format=9.0]{86870100}	&	\tablenum[group-separator={,},table-format=9.0]{177181000}	&	\tablenum[group-separator={,},table-format=9.0]{90169300}	&	29	&	\tablenum[table-format=3.2]{203.15}	\\
		\hline
	\end{tabular}
	\caption{Performance of \fairsyn on the code-aware resource management benchmark experiment.}
	\label{tab:lego results}
\end{table}

\subsubsection{Controller Synthesis for Stochastically Perturbed Dynamical Systems}
\label{sec:stoch experiment}

Synthesizing verified symbolic controllers for continuous dynamical systems is an active area in cyber-physical systems research \cite{tabuada2009verification}.
Recently, it was shown by \citet{Majumdar2021}, that the symbolic controller synthesis problem for stochastic continuous dynamical systems can be approximated using a strategy synthesis problem over a (finite) $\twohalf$-player game graph.
This result, together with our reduction in \REFsec{sec:stochastic}, enables us to use \fairsyn to synthesize a symbolic controller for stochastic continuous dynamical systems.
We show in this section, that on different instances of an established case study for this synthesis problem, \fairsyn outperforms state-of-the art synthesis techniques by margins varying between $1$ order of magnitude to up to $2.5$ orders of magnitude.

In the following, we first formalize the case study, which was proposed by \citet{dutreix2020abstraction}.
Consider the dynamic model of a bistable switch which is a tuple $\Sigma=(X,U,W,f)$ with a two-dimensional compact state space $X = [0,4]\times [0,4]\in \mathbb{R}^2$, 
a finite input space $U = \set{-0.5,0,0.5} \times \set{-0.5,0,0.5}$, a two-dimensional bounded disturbance space $W = [-0.4,-0.2]\times [-0.4,-0.2]\in \mathbb{R}^2$, and a transition function $f\colon X\times U\to X$.
Suppose $x\colon \mathbb{N}\to X$, $u\colon \mathbb{N}\to U$, and $w\colon \mathbb{N} \to W$ denote the system's state, input, and disturbance trajectories, given as functions of (discrete) time.
Note that the functions $x$, $u$, $w$, and $f$ are vector-valued, and we will denote each element of vectors using the element index in the suffix.
For instance, $x_1, x_2$ are the first and the second element of the state trajectory $x$ respectively, and $f_1(x,u), f_2(x,u)$ are the first and the second element of the valuation of the transition function $f(x,u)$ respectively.
At each time step $k$, we assume that $w(k)\in W$ is drawn from a probability distribution with the support $W$; for our purpose, the shape of the distribution is irrelevant.
The state evolution of the system is modeled using a set of difference equations of the following form:
\begin{align}
	x_1(k+1) &= f_1(x(k),u(k)) + w_1(k) = x_1(k) + 0.05\left(-1.3x_1(k) + x_2(k)\right) + u_1(k) + w_1(k),\label{equ:switch dynamics} \\
	x_2(k+1) &= f_2(x(k),u(k)) + w_2(k) =  x_2(k) + 0.05\left( \frac{(x_1(k))^2}{(x_1(k))^2+1} - 0.25x_2(k) \right) + u_2(k) + w_2(k). \nonumber
\end{align}

A controller for a dynamical system $\Sigma$ is a function $C\colon X\to U$ that determines the control inputs $u_1(k):=C_1(x(k))$ and $u_2(k):=C_2(x(k))$ in \eqref{equ:switch dynamics} for all time steps $k$. Recalling that $w(k)\in W$ is drawn from a probability distribution with the support $W$ in every time step, we see that, for a given initial state $x(0)=\mathsf{init}\in X$, a fixed controller $C$ induces a probability measure $P_{\mathsf{init}}^C$ over all state trajectories 
starting at $x(0)=\mathsf{init}$ and evolving in accordance to \eqref{equ:switch dynamics}. 

\begin{sidefigure}
\begin{center}
 	\begin{tikzpicture}[scale=0.8] 
		\draw [draw=black!50!white] (0,0) grid  (4,4) rectangle (0,0);
		\draw [fill=red,opacity=0.3]	(1,1)	--	(3,1)	--	(3,3)	--	(2,3)	--	(2,2)	--	(1,2)	--	cycle;
		\draw [fill=green,opacity=0.3]	(0,0)	rectangle	(1,1);
		\draw [fill=blue,opacity=0.3]	(1,1)	rectangle	(2,2);
		\draw [fill=blue,opacity=0.3]	(3,1)	rectangle	(4,3);
		\draw [fill=blue,opacity=0.3]	(0,3)	rectangle	(1,4);
		\draw [fill=red!50!yellow,opacity=0.3]	(2,3)	rectangle	(3,4);
		
		\node at		(1.5,1.5)	{$A,C$};
		\node at		(2.5,1.5)	{$A$};
		\node at		(2.5,2.5)	{$A$};
		\node at		(0.5,0.5)	{$B$};
		\node at		(3.5,1.5)	{$C$};
		\node at		(3.5,2.5)	{$C$};
		\node at		(0.5,3.5)	{$C$};
		\node at		(2.5,3.5)	{$D$};
	\end{tikzpicture}
	\end{center}
	\caption{Predicates over $X$.}
	\label{fig:switch predicates}
\end{sidefigure}
\leavevmode
In order to formalize a control specification for $\Sigma$ in \eqref{equ:switch dynamics}, the state subsets $A,B,C,D\subseteq X$ whose shape is illustrated in \REFfig{fig:switch predicates} are considered. Given the LTL formulas over these predicates 
 \begin{align*}
	\varphi_1 &\coloneqq \square\left( \left( \lnot A\wedge \bigcirc A\right) \rightarrow \left( \bigcirc\bigcirc A\wedge \bigcirc\bigcirc\bigcirc A\right) \right),~\text{and} \phantom{XXXXXXXXXXXXXXXXXXX}\\
	\varphi_2 &\coloneqq \left( \square\lozenge B\rightarrow \lozenge C\right) \wedge \left( \lozenge D\rightarrow \square\lnot C\right), 
\end{align*}
the set $\mathcal{L}(\varphi_i)\subseteq 2^{\mathbb{N}\to X}$ collects all state trajectories of $\Sigma$ that fulfill $\varphi_i$. With this, we define the \emph{almost sure winning region} of $\Sigma$ for the specification $\varphi$ as the largest (in term of set inclusion) set of states $W_{\mathsf{in}}$ for which there exists a controller $C$ s.t.\ $P_{\alpha}^C(\mathcal{L}(\varphi_i))=1$ for every state $\alpha\in W_{\mathsf{in}}$.
The synthesis task for this case study then amounts to computing controllers $C_1$ and $C_2$ which have the \emph{almost sure winning region} of $\Sigma$ w.r.t.\ $\varphi_i$ and $W_{\mathsf{in}}$ as their initial domain.

It was shown by \citet{Majumdar2021} that this synthesis problem can be approximately solved by lifting the system $\Sigma$ to a finite $\twohalf$-player game.
The almost sure winning region of the resulting controller obtained by solving the abstract $\twohalf$-player game under-approximates the almost sure winning region of $\Sigma$.
We employ our fixpoint algorithm for solving this abstract $\twohalf$-player game, which can be reduced to a fair adversarial game by following the procedure in \REFsec{sec:stochastic}.
In \REFtab{table:performance summary}, we compare both the accelerated and the non-accelerated versions of our fixpoint algorithm against the state-of-the-art algorithm for solving this problem, which is implemented in the tool called StochasticSynthesis (SS) \cite{dutreix2020abstraction}.

\begin{table*}
	\centering
	\begin{tabular}{
		|>{\centering\arraybackslash}m{1.2\ThCScm}|
		>{\centering\arraybackslash}m{2\ThCScm}|
		>{\centering\arraybackslash}m{1.6\ThCScm}|
		>{\centering\arraybackslash}m{1.8\ThCScm}|
		>{\centering\arraybackslash}m{1.6\ThCScm}|
		>{\centering\arraybackslash}m{1.4\ThCScm}|
		>{\centering\arraybackslash}m{1.4\ThCScm}|
		>{\centering\arraybackslash}m{1.4\ThCScm}|}
		\hline
		\multirow{3}{*}{\parbox{0.8\ThCScm}{\centering Spec.}} & \multirow{3}{*}{\parbox{2\ThCScm}{\centering \# vertices in\\ $\twohalf$-game abstraction}} & \multicolumn{3}{c|}{Total synthesis time} & \multicolumn{3}{c|}{Peak memory footprint}\\
		\cline{3-8}
		  & & \multirow{2}{*}{\parbox{1.5\ThCScm}{\centering \textsf{\small\fairsyn}}} & \multirow{2}{*}{\parbox{1.5\ThCScm}{\centering \textsf{\small\fairsyn \\ w/o accl.}}} & \multirow{2}{*}{\parbox{1.5\ThCScm}{\centering \textsf{\small SS}}} & \multirow{2}{*}{\parbox{1.4\ThCScm}{\centering \textsf{\small\fairsyn}}} & \multirow{2}{*}{\parbox{1.4\ThCScm}{\centering \textsf{\small\fairsyn \\ w/o accl.}}} & \multirow{2}{*}{\parbox{1.5\ThCScm}{\centering \textsf{\small SS} }} \\
		  & &  & & & & & \\
		\hline 
		\multirow{5}{*}{\parbox{0.8\ThCScm}{\centering $\varphi_1$ \\ ($1$ \\ Rabin \\ pair)}} & \num{3.8e3} & \num{0.02}\,\si{\second} & \num{0.02}\,\si{\second} & \num{8}\,\si{\second} & \num{65}\,\si{\mebi\byte} & \num{65}\,\si{\mebi\byte} & \num{125}\,\si{\mebi\byte} \\
		& \num{2.2e4} & \num{0.2}\,\si{\second} & \num{0.4}\,\si{\second} & \num{18}\,\si{\second} & \num{68}\,\si{\mebi\byte} & \num{68}\,\si{\mebi\byte} & \num{1}\,\si{\gibi\byte}  \\
		& \num{1.1e5} &  \num{1.3}\,\si{\second} & \num{3.7}\,\si{\second} & \num{9}\,\si{\minute}\,\num{18}\,\si{\second} & \num{79}\,\si{\mebi\byte} & \num{81}\,\si{\mebi\byte} & \num{80}\,\si{\gibi\byte} \\
		& \num{6.6e5}  &  \num{5.4}\,\si{\second}  & \num{16.8}\,\si{\second} & OoM & \num{128}\,\si{\mebi\byte} & \num{126}\,\si{\mebi\byte} & \num{127}\,\si{\gibi\byte} \\
		& \num{4.3e6} &  \num{35}\,\si{\second} & \num{1}\,\si{\minute}\,\num{32}\,\si{\second} & OoM & \num{479}\,\si{\mebi\byte} & \num{478}\,\si{\mebi\byte} & \num{127}\,\si{\gibi\byte} \\
		 \hline
		 \multirow{5}{*}{\parbox{0.8\ThCScm}{\centering $\varphi_2$ \\ ($2$ \\ Rabin \\ pairs)}} & \num{3.8e3} & \num{0.4}\,\si{\second} & \num{1}\,\si{\second} & \num{30}\,\si{\second} & \num{66}\,\si{\mebi\byte} & \num{66}\,\si{\mebi\byte} & \num{156}\,\si{\mebi\byte} \\
		& \num{2.2e4} &  \num{8.2}\,\si{\second} & \num{41}\,\si{\second}  & \num{55}\,\si{\second} & \num{72}\,\si{\mebi\byte} & \num{69}\,\si{\mebi\byte} & \num{1}\,\si{\gibi\byte} \\
		& \num{1.1e5} &  \num{1}\,\si{\minute}\,\num{23}\,\si{\second} & \num{12}\,\si{\minute}\,\num{38}\,\si{\second} & \num{16}\,\si{\minute}\,\num{1}\,\si{\second} & \num{108}\,\si{\mebi\byte} & \si{102}\,\si{\mebi\byte} & \num{81}\,\si{\gibi\byte} \\
		& \num{6.6e5} &  \num{5}\,\si{\minute}\,\num{27}\,\si{\second} & \num{1}\,\si{\hour}\,\num{1}\,\si{\minute} & OoM & \num{166}\,\si{\mebi\byte} & \num{237}\,\si{\mebi\byte} & \num{126}\,\si{\gibi\byte} \\
		& \num{4.3e6} & \num{41}\,\si{\minute}\,\num{7}\,\si{\second} & \num{6}\,\si{\hour}\,\num{5}\,\si{\minute} & OoM & \SI{517}{\mebi\byte} & \SI{509}{\mebi\byte} & \SI{127}{\gibi\byte} \\
		\hline
	\end{tabular}
	\caption{Performance comparison between \fairsyn and \textsf{StochasticSynthesis} (abbreviated as \textsf{SS}) \cite{dutreix2020abstraction} on a comparable implementation of the abstract fair adversarial game (uniform grid-based abstraction).
	Col.~$1$ shows the specifications considered and the respective numbers of Rabin pairs, 
	Col.~$2$ shows the size of the resulting $\twohalf$-player game graph (computed using the algorithm given in \cite{Majumdar2021}, 
	Col.~$3$, $4$, and $5$ compare the total synthesis times and Col.~$6$, $7$, and $8$ compare the peak memory footprint (as measured using the ``time'' command) for $\fairsyn$, $\fairsyn$ w/o acceleration, and \textsf{SS} respectively.
	``OoM'' stands for out-of-memory.}
	\label{table:performance summary}
\end{table*}


	\section{Conclusion}

Many practical problems in reactive synthesis give rise to two-player games on graphs with a winning
condition of the form
\[
\mathsf{Fairness~Assumption}\ \Rightarrow\ \omega\mathsf{-regular~Specification}
\]
The prevalent way to solve games with fairness assumptions is to either ``compile'' to a new $\omega$-regular
specification for the implication or to identify selected fragments for which a ``direct'' symbolic algorithm
has been devised.
The former can handle arbitrary fairness assumptions (e.g., general Streett conditions) but yields an algorithm
of high complexity (e.g., adding the number of Streett conditions in the exponent).
The latter, exemplified by the GR(1) fragment, can only handle weak fairness (conjunctions of B\"uchi conditions).
Our observation is that many practical fairness assumptions fall into the category of strong transition liveness, and for this class, one can construct a symbolic algorithm with a slight additional
penalty that is independent of the size (number of live edges) of the liveness assumption.
As a byproduct, our algorithm improves a previous symbolic algorithm for stochastic Rabin games.
We experimentally demonstrate that a symbolic implementation of our algorithm based on BDDs can scale to large instances
derived from deterministic and stochastic synthesis problems.

        \paragraph{Acknowledgements.}
 		We thank Daniel Hausmann and Nir Piterman for valuable comments on an earlier version of this manuscript, in particular for the observation that the 
 	parity fixpoint does not allow for a ``direct transformation''. 
 	We also thank the anonymous reviewers for their constructive comments.


 \printbibliography
 \newpage
\appendix
\section{Example-Computation of the Rabin Fixpoint}\label{app:ExpFP}
\begin{figure}[b]
\begin{center}
 \includegraphics[width=0.5\linewidth]{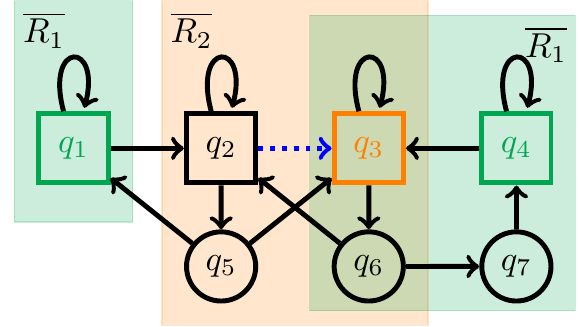}
\end{center}
\caption{Example of a fair adversarial Rabin game with two pairs $\tuple{G_1,R_1}=\tuple{\set{q_1,q_4},\set{q_2,q_5}}$
($G_1$ and $\overline{R_{1}}$ are indicated in green)
and $\tuple{G_2,R_2}=\tuple{\set{q_3},\set{q_1,q_4,q_7}}$ ($G_2$ and $\overline{R_{2}}$ are indicated in orange), 
and one live edge $\El=\set{(q_2,q_3)}$ (dashed blue).}
\label{fig:example}
\end{figure}

Consider the game graph depicted in \REFfig{fig:example}, where circles and squares denote \p{0} and \p{1} vertices, respectively. We are given a Rabin condition with two pairs $\FR = \set{\tuple{G_1,R_1},\allowbreak\tuple{G_2,R_2}}$ s.t.\
 \begin{align*}
\overline{R_{1}}=\set{q_1,q_3,q_4,q_6,q_7}   &&G_{1}=\set{q_1,q_4} &&
\overline{R_{2}}=\set{q_2,q_3,q_5,q_6}    &&G_{2}=\set{q_3}
\end{align*}
which are indicated in green and orange, respectively, in \REFfig{fig:example}. The only live edge in the game graph is indicated in dashed blue from $q_2$ to $q_3$. 
We assert that \p{0} wins from every vertex.
However, in the absence of the live edge, she wins only from $\set{q_3, q_4, q_5, q_6, q_7}$.
(This is because \p{1} can force the game to stay forever in $q_2$ from the remaining states.)

We first flatten the algorithm in \eqref{equ:Rabin_all} for two Rabin pairs. This yields the following algorithm: 
\begin{subequations}\label{equ:example}
 \begin{align}
 \nu Y_0.~\mu X_0&.~\label{equ:example:out}\\
 \big\{\nu Y_1.&\mu X_1.~ \nu Y_2.\mu X_2.~\label{equ:example:up}\\
&\apre(Y_{0},X_{0})\notag\\
\qquad\cup~
&\left(\overline{R}_{1}\cap
 \left[ 
 \left( G_{1}\cap \cpre(Y_{1})\right)
 \cup \left(\apre(Y_{1},X_{1})\right)
 \right]\right)\notag\\
 \qquad\cup~
 &\left(\overline{R}_{1}\cap  \overline{R}_{2}\cap
 \left[ 
 \left( G_{2}\cap \cpre(Y_{2})\right)
 \cup \left(\apre(Y_{2},X_{2})\right)
 \right]\right)\notag\\
 \cup~ \nu Y'_2.&\mu X'_2.~ \nu Y'_1.\mu X'_1.~\label{equ:example:down}\\
&\apre(Y_{0},X_{0})\notag\\
\qquad\cup~
&\left(\overline{R}_{2}\cap
 \left[ 
 \left( G_{2}\cap \cpre(Y'_{2})\right)
 \cup \left(\apre(Y'_{2},X'_{2})\right)
 \right]\right)\notag\\
 \qquad\cup~
 &\left(\overline{R}_{1}\cap  \overline{R}_{2}\cap
 \left[ 
 \left( G_{1}\cap \cpre(Y'_{1})\right)
 \cup \left(\apre(Y'_{1},X'_{1})\right)
 \right]\right)\big\}\notag
\end{align}
\end{subequations}

We first consider the upper part of \eqref{equ:example}, i.e., the permutation sequence $\delta=012$ (labeled by \eqref{equ:example:up}). We first recall that the computation is initialized with $Y_i^0=V$ and $X_i^0=\emptyset$ and we see from the structure of the game graph that $\cpre(V)=V$. Further, we see from the definition of $\apre$ that $\apre(\cdot,\emptyset)=\emptyset$. So, we have 
\begin{align*}
 X_2^1&=(\overline{R}_{1}\cap G_1)\cup (\overline{R}_{1}\cap  \overline{R}_{2}\cap G_2)
 =\set{q_1,q_4}\cup\set{q_3}=\set{q_1,q_3,q_4}.
\end{align*}
As $q_6$ is the only other state in $\overline{R}_{1}\cap  \overline{R}_{2}$ and $q_6$ does not have an edge to $\set{q_1,q_3,q_4}$ the iteration over $X_2$ terminates and we get $Y_2^1=\set{q_1,q_3,q_4}$. As $q_3\not\in\cpre(Y_2^1)$ the last line of the upper part of \eqref{equ:example} becomes the empty set and we terminate with $Y_2^*=X_2^*=(\overline{R}_{1}\cap G_1)=\set{q_1,q_4}$. This gives $X_1^1=\set{q_1,q_4}$ and resets $Y_2$ and $X_2$ to $V$ and $\emptyset$, respectively. Therefore, we now get
\begin{align*}
 X_2^1&=(\overline{R}_{1}\cap G_1)\cup (\overline{R}_{1}\cap \apre(V,X_1^1))\cup (\overline{R}_{1}\cap  \overline{R}_{2}\cap G_2)
 =\set{q_1,q_4}\cup\set{q_7}\cup\set{q_3}.
\end{align*}
Now, as $q_7\in X_2^1$, also $q_6$ is added before $X_2$ terminates. This now gives $Y_2^1=\set{q_1,q_3,q_4,q_6,q_7}$ and hence $q_3\in\cpre(Y_2^1)$. As there are no other states in $\overline{R}_{1}\cap  \overline{R}_{2}\cap G_2$ that can be added to this set, the iteration over $X_2$ terminates and we get $Y_2^2=\set{q_1,q_3,q_4,q_6,q_7}$, which also terminates the iteration over $Y_2$, resulting in $X_1^2=\set{q_1,q_3,q_4,q_6,q_7}$.
As there are again no other states inside $\overline{R}_{1}$ that could be added, this iteration over $X_1$ terminates, giving $Y_1^1=\set{q_1,q_3,q_4,q_6,q_7}$. Now we see that $\cpre(Y_1^1)=\set{q_3,q_4,q_6,q_7}$. As the exclusion of $q_1$ from $Y_1$ does not influence the reasoning about $\set{q_3,q_4,q_6,q_7}$ the iteration terminates with $Y_1^*=\set{q_3,q_4,q_6,q_7}$.

Now we consider the lower part of \eqref{equ:example}, i.e., the permutation sequence $\delta=021$ (labeled by \eqref{equ:example:down}). Here, we get 
\begin{align*}
 {X'_1}^1=(\overline{R}_{2}\cap G_2)\cup (\overline{R}_{1}\cap  \overline{R}_{2}\cap G_1)=\set{q_3}\cup\emptyset=\set{q_3}.
\end{align*}
For the same reason as before we see again that the last line of the lower part of \eqref{equ:example} becomes the empty set and we terminate with ${Y'_1}^*={X'_1}^*=(\overline{R}_{2}\cap G_2)=\set{q_3}$. This gives ${X'_2}^1=\set{q_3}$ and resets $Y'_1$ and $X'_1$ to $V$ and $\emptyset$, respectively. With this, we now get 
\begin{align*}
 {X'_1}^2&=(\overline{R}_{2}\cap G_2)\cup \apre(Q,{X'_2}^1)\cup (\overline{R}_{1}\cap  \overline{R}_{2}\cap G_1)
 =\set{q_3}\cup\set{q_2,q_5}\cup\emptyset.
\end{align*}
Here, for the first time, the live edge from $q_2$ to $q_3$ comes into play. If this would not be a live edge, $q_2$ would not be added to $X'_1$, as in this case the environment could trap the game in $q_2$, and thereby prevent the second Rabin pair to hold. However, due to the edge from  $q_2$ to $q_3$ being live, we know that the environment will always eventually transition from $q_2$ to $q_3$. With this, now also $q_6$ is added to $X'_1$, finally leading to a termination of the iteration over $X'_2$ with $\set{q_2,q_3,q_5,q_6}$ and hence ${Y'_2}^1=\set{q_2,q_3,q_5,q_6}$. As $q_3\in\cpre({Y'_2}^1)$ the iteration over $Y'_2$ terminates with ${Y'_2}^*=\set{q_2,q_3,q_5,q_6}$. 

With both the upper and the lower part of \eqref{equ:example} terminated, we can now take the union of $Y_1^*=\set{q_3,q_4,q_6,q_7}$ and ${Y'_2}^*=\set{q_2,q_3,q_5,q_6}$ to get $X_0^1=\set{q_2\hdots q_7}$ (reaching the part of the formula labeled with \eqref{equ:example:out}). After this update of $X_0$ all inner fixpoint variables (in \eqref{equ:example:up} and \eqref{equ:example:down}) are reset, and the upper and lower expressions in \eqref{equ:example} are re-evaluated. As $\apre(Q,X_0^1)=\set{q_2\hdots q_7}$, we see that every iteration over $X_i$ in \eqref{equ:example:up} and \eqref{equ:example:down} is essentially initialized with a set containing $\set{q_2\hdots q_7}$. This implies that $q_1$ will actually remain within $Y_1$, leading to $Y_1^*=V$, and with this $X_0^2=V$. As this implies $Y_0^1=V=Y_0^0$, the computation terminates with $Z^*=V$.

Despite all states being winning, we see that \p{0} has to play appropriately to enforce winning. 
Intuitively, from state $q_5$ she must go to $q_3$ and from $q_6$ she has to consistently either 
(i) always go to $q_2$ or 
(ii) always go to $q_7$. 
If she picks option (i), the play is won by satisfying the second Rabin pair, i.e., always eventually visiting $q_3$ while remaining within $\overline{R_2}$. 
If she picks option (ii), it is up to the environment whether the game is won by satisfying the first or the second Rabin pair. 
Intuitively, if the environment plays such that either (a) the game eventually remains in $q_4$ or (b) the edges $(q_4,q_3)$ and $(q_3,q_6)$ are taken infinitely often, 
the game fulfills the first Rabin condition. 
If, however, (c), the environment decides to trap the game in $q_3$, the game is won by satisfying the second Rabin pair. 
This influence of the environment on the selection of the satisfied Rabin pair intuitively requires the evaluation of all possible permutation sequences 
in the evaluation of the fixpoint algorithm. 
We will see later that for Rabin pairs which are ordered by inclusion (corresponding to the special case of a Rabin-chain condition), no permutation is required.

We comment that the strategy construction outlined in \REFthm{thm:strategy} provided in \REFapp{appD} chooses to enforce a transition from $q_6$ to $q_7$ (see Example~\ref{example:appendix:strategy} in \REFapp{appD} for a detailed discussion).

\section{Detailed Proofs}

\subsection{General Lemmas}
We first introduce some useful general lemmas.

\begin{lemma}\label{lem:AsubsetC}
 If $Y\supseteq X$ then
 $\cpre(Y)\cup\apre(Y,X)=\cpre(Y)$.
\end{lemma}
\begin{proof}
The claim follows from the following derivation
 \begin{align*}
 \cpre(Y)\cup\apre(Y,X)
 &=\cpre(Y)\cup\cpre(X)\cup\left(\elpre(X)\cap\eapre_1(Y)\right)\\
 &=\cpre(Y)\cup\left(\elpre(X)\cap\eapre_1(Y)\right)\\
 &=\left(\cpre(Y)\cup\elpre(X)\right)\cap\left(\cpre(Y)\cup\eapre_1(Y)\right)\\
 &=\left(\cpre(Y)\cup\elpre(X)\right)\cap\cpre(Y)\\
 &=\cpre(Y)
\end{align*}
where the second line follows from $\cpre(X)\subseteq\cpre(Y)$ (as $X\subseteq Y$) and the fourth line follows as 
$\cpre(Y)=\epre_0(Y)\cup\eapre_1(Y)\supseteq\eapre_1(Y)$.
\end{proof}

\begin{lemma}\label{lem:Ared}
 If $Y\subseteq X$ then
 $\apre(Y,X)=\cpre(X)$.
\end{lemma}
\begin{proof}
The claim follows from the following derivation
 \begin{align*}
 \apre(Y,X)
 &=\cpre(X)\cup\left(\elpre(X)\cap\eapre_1(Y)\right)\\
 &=\left(\cpre(X)\cup\elpre(X)\right)\cap\left(\cpre(X)\cup\eapre_1(Y)\right)\\
 &=\left(\cpre(X)\cup\elpre(X)\right)\cap\cpre(X)\\
 &=\cpre(X)
\end{align*}
where  the fourth line follows as 
$\cpre(X)=\epre_0(X)\cup\eapre_1(X)\supseteq\eapre_1(Y)$ as $Y\subseteq X$.
\end{proof}

\begin{lemma}\label{lem:contain}
Let $f(X,Y)$ and $g(X,Y)$ be two functions which are monotone in both $X\subseteq V$ and $Y\subseteq V$. Further, let
  \begin{align*}
  Z_a:=&\nu Y_a.~\mu X_a.~ \nu Y_b.~\mu X_b.~f(X_a,Y_a)\cup g(X_b,Y_b)\\
  Z_b:=&\nu Y_a.~\mu X_a.~ \nu Y_b.~\mu X_b.~g(X_a,Y_a)\cup f(X_b,Y_b)\\
  Z_c:=&\nu Y_c.~\mu X_c.~f(X_c,Y_c)
 \end{align*}
Then it holds that 
\begin{compactenum}[(i)]
 \item $Z_c\subseteq Z_a$ and
 \item $Z_c\subseteq Z_b$.
\end{compactenum}
If, in addition, $g(X,Y)\subseteq f(X,Y)$ for all $X,Y\subseteq V$, then it holds that
\begin{compactenum}[(i)]
 \item[(iii)] $Z_a=Z_c$ and
 \item[(iv)] $Z_b=Z_c$.
 \end{compactenum}
\end{lemma}

\begin{proof}
 We prove all claims separately:\\
\begin{inparaitem}[$\blacktriangleright$]
 \item \textit{(i)} \enquote{$Z_c\subseteq Z_a$} :
  First, consider a stage of the fixpoint evaluation where $Y_a$ and $X_a$ have their initialization value $Y_a^0:=V$ and $X_a^{00}:=\emptyset$ (here, the notation $X_a^{lk}$ refers to the value of $X_a$ computed in the $k$'th iteration over $X_a$ using the value for $Y_a$ computed in the $l$'th iteration over~$Y_a$). Then we see that $X_a^{01}=Y_b^{00*}$ where $Y_b^{00*}=f(\emptyset,V)\cup g(Y_b^{00*},Y_b^{00*})$. We therefore see that $X_a^{01}\supseteq X_c^{01}=f(\emptyset,V)$. With this, it follows from the monotonicity of $f$ and $g$ that $Y_a^{01}=X_a^{0*}\supseteq X_c^{0*}=Y_c^{1}$. 
  With this, we see that $X_a^{m1}\supseteq X_c^{m1}$ for all $m>0$ and therefore $Z_a=Y_a^{*}\supseteq Y_c^{*}=Z_c$. \\
\item \textit{(ii)} \enquote{$Z_c\subseteq Z_b$} :
Consider arbitrary values $Y_a^m$ and $X_a^{mn}$ and assume that $Y_b$ and $X_b$ have their initialization value, i.e., $Y_b^{mn0}:=V$ and $X_b^{mn00}:=\emptyset$. Then we have 
\begin{equation*}
 X_b^{mn01}=g(X_a^{mn},Y_a^m)\cup f(\emptyset,V)\supseteq X_c^{01}.
\end{equation*}
Using the same reasoning as in the previous part, we see that this implies $Y_b^{mn*}\supseteq Y_c^{*}=Z_c$. As this holds for any $m$ and $n$ it also holds when the fixed-point over $Y_a$ and $X_a$ is obtained, i.e., when we have $Z_a=Y_a^*=Y_b^{***}$, which proves the statement.\\
\item \textit{(iv)} \enquote{$Z_c\supseteq Z_b$} :
 First, observe that for the initialization values $Y_a^0=Y_b^{000}=V$ and $X_a^{00}=X_b^{0000}=\emptyset$ we have $g(\emptyset,V)\subseteq f(\emptyset,V)$. We therefore have 
 \begin{equation*}
 Y_b^{00*}=X_b^{00**}=f(X_b^{00**},Y_b^{00*})= Z_c
\end{equation*}
Now it remains to show, that the outer fixpoint cannot add any additional states. First, observe that $X_a^{01}= Y_b^{00*}$ and 
 \begin{equation*}
 X_b^{0100}=g(X_a^{01},V)\cup f(\emptyset,V)\subseteq f(X_a^{01},V)\cup f(\emptyset,V)=f(X_a^{01},V)
\end{equation*}
Now it follows from the famous acceleration result of \citet{long1994improved} that warm-starting the inner fixpoint computation with $X_a^{01}$ yields the same inner fixpoint. With this, we see that $X_a^{0n}=Z_c$ for all $n$, implying $Y_a^0=X_a^{0*}=Z_c$. As $Z_b=Y_a^*\subseteq Y_a^0$, this proves the claim.\\
\item \textit{(iii)} \enquote{$Z_c\supseteq Z_a$} :
As $g(X,Y)\subseteq f(X,Y)$ for all $X,Y\subseteq V$ it follows from the monotonicity of $g$ and~$f$ that 
\begin{equation*}
 Z_a\subseteq\nu Y_a.~\mu X_a.~ \nu Y_b.~\mu X_b.~f(X_a,Y_a)\cup f(X_b,Y_b)
\end{equation*}
with this, it follows from (iv) that $Z_a\subseteq Z_c$, what proves the claim. 
\end{inparaitem}
\end{proof}


\subsection{Additional Proofs for \REFsec{sec:RabinGames}}

\subsubsection{Proof of \REFthm{thm:Reachability}}\label{app:prop:Reachability}
\begin{theorem*}[\REFthm{thm:Reachability} restated for convenience]
  Let $\Gl =\tup{\game, \El}$ be a game graph with live edges and $\tuple{T,Q}$ be a safe reachability winning condition. Further, let
\begin{equation}\label{p:equ:Reach:FP}
 Z^*\coloneqq\nu Y.~\mu X.~T \cup (Q\cap\apre(Y,X)).
\end{equation}
Then $Z^*$ is equivalent to the winning region of $\p{0}$ in the fair adversarial game over $\Gl$ for the winning condition $\psi$ in~\eqref{equ:Reach:psi}. Moreover, the fixpoint algorithm runs in $O(n^2)$ symbolic steps, and a memoryless winning strategy for $\p{0}$ can be extracted from it.
\end{theorem*}

We denote by $Y^m$ the $m$-th iteration over the fixpoint variable $Y$ in \eqref{p:equ:Reach:FP}, where $Y^0=V$. Further, we denote by $X^{mi}$ the set computed in the $i$-th iteration over the fixpoint variable $X$ in \eqref{p:equ:Reach:FP} during the computation of  $Y^m$ where  $X^{m0}=\emptyset$.
Then it follows form \eqref{p:equ:Reach:FP} that 
\begin{align*}
X^{m1}&=X^{m0}\cup T\cup (Q\cap \apre(Y^{m-1},X^{m0}))
=\emptyset\cup T\cup (Q\cap \apre(Y^m,\emptyset))=T,\\
X^{m2}&=X^{m1}\cup T\cup (Q\cap \apre(Y^{m-1},X^{m1}))
=T\cup (Q\cap \apre(Y^{m-1},X^{m1}))\supseteq X^{m1},
\end{align*}
and therefore, in general,
\begin{align*}
X^{mi+1}=T\cup (Q\cap \apre(Y^{m-1},X^{mi}))\supseteq X^{mi}.
\end{align*}
With this, the fixed-point over $X$ corresponds to the set $X^{m*}=\bigcup_{i>0} X^{mi}=X^{mi^\uparrow}$, where $i^\uparrow$ is the iteration where the fixed-point over $X^{mi}$ is attained. 

Now consider the computation of $Y$. Here we have $Y^0=V$ and $Y^m=Y^{m-1}\cap X^{m*}\subseteq Y^{m-1}$ where equality holds when a fixed-point is reached. Hence, in particular we have $Y^*=X^{**}=Z^*$. For simplicity we denote $X^{*i}$ by $X^i$.

\smallskip
\noindent\textbf{Strategy construction.}
In order to construct a winning strategy for $\p{0}$ from \eqref{p:equ:Reach:FP}, we construct a ranking over $V$ by choosing 
\begin{align}
\label{equ:Reach:ranking}
 \rank{v}=i~\Leftrightarrow~v\in X^i\setminus X^{i-1}\quad\text{and}\quad
 \rank{v}=\infty~\Leftrightarrow v\notin Z^*.
\end{align}
As $X^0=\emptyset$, $X^1=T$ (from above) and $Z^*=\bigcup_{i>0} X^i$, it follows that $\rank{v}=1$ iff $v\in T$ and $1<\rank{v}<\infty$ iff $v\in Z^*\setminus T$. Using this ranking we define a $\p{0}$ strategy $\rho_0:V_0\rightarrow V$ s.t.\
\begin{equation}\label{equ:Reach:rho}
\rho_0(v)=\min_{(v,w)\in E} \rank{w}.
\end{equation}
We next show that this player $0$ strategy is actually winning w.r.t.\ $\psi$ (in \eqref{equ:Reach:psi}) in every fair adversarial play over $\Gl$.

\smallskip
\noindent\textbf{Soundness.}
To prove soundness, we need to show $Z^*\subseteq \mathcal{W}$. That is, we need to show that for all $v\in Z^*$ there exists a strategy for player $0$ s.t.\ the goal set $T$ is eventually reached along all live compliant plays $\pi$ starting at $v$ while staying in $Q$. We choose $\rho_0$ in \eqref{equ:Reach:rho} and show that the claim holds. 

First, it follows from the definition of $\apre$ that for a vertex $v\in Z^*$ exactly one of the following cases holds:\\
\begin{inparaenum}[(a)]
\item $v\in T$ and hence $\rank{v}=1$,\\
\item $v\in (V_0\cap Z^*)\setminus T$, i.e., $1<\rank{v}<\infty$ and $v\in Q$ and there exists a $v'\in E(v)$ with $\rank{v'}<\rank{v}$,\\ 
\item $v\in ((V_1\setminus \Vl) \cap Z^*))\setminus T$, i.e., $1<\rank{v}<\infty$ and $v\in Q$ and for all $v'\in E(v)$ it holds that $\rank{v'}<\rank{v}$, or\\
\item[($\ell$)] $v\in (\Vl\cap Z^*)\setminus T$, i.e., $1<\rank{v}<\infty$ and $v\in Q$ and there exists a $v'\in \El(v)$  with $\rank{v'}<\rank{v}$ \emph{and} $E(v)\subseteq Z^*$.\\
\end{inparaenum}
We see that $\rho_0(v)$ chooses one existentially quantified edge in (b) vertices. In all other cases player $1$ chooses the successor. 

Further, we see that any play $\pi$ which starts in $\pi(0)=v\in Z^*$ and obeys $\rho_0$ has the property that $\pi(k)\in Z^*\setminus T$ implies $\pi(k)\in Q$ and $\pi(k+1)\in Z^*$ for all $k\geq 0$. This, in turn, means that for any such state $v=\pi(k)\in Z^*\setminus T$ as well as for its successor $\pi(k+1)$ a rank is defined, i.e., $\pi(k)\in X^{i}$ for some $0<i<\infty$ and exactly one of the cases (b)-($\ell$) applies. 
 We call a vertex for which case ($\alpha$) applies, an ($\alpha$) vertex.

Now observe that the above reasoning implies that whenever an (a) vertex is hit along a play $\pi$ the claim holds. We therefore need to show that any play starting in $v\in Z^*$ eventually reaches an (a) vertex. First, consider a play in which no ($\ell$) vertex occurs. Then constantly hitting (b) and (c) vertices always reduces the rank of visited states (as we assume that $\pi$ obeys~$\rho_0$ in \eqref{equ:Reach:rho}). As the maximal rank is finite, we see that we must eventually hit a state with rank $1$, which is an (a) state. 

Note that the same argument holds when only a finite number of ($\ell$) vertices is visited along $\pi$. In this case we know that from some time onward no more ($\ell$) vertex occurs. As the last~($\ell$) vertex has a finite rank, there can only be a finite sequence of (b) and (c) vertices afterwards until finally an (a) vertex is reached.

We are therefore left with showing that on every path with an infinite number of ($\ell$) vertices, eventually an (a) vertex will be reached. We prove this claim by contradiction. I.e., we show that there cannot exist a path with infinitely many ($\ell$) vertices and no (a) vertex.

We first show that infinitely many ($\ell$) vertices and no (a) vertices in $\pi$ imply that vertices with rank $2$ can only occur finitely often along~$\pi$.\\
\begin{inparaitem}[$\blacktriangleright$]
 \item Recall that the construction of $\rho_0$ ensures that whenever we visit a state $v\in V_0\cap Z^*$ with $\rank{v}=2$ we will surely visit a state with rank~$1$ afterwards, implying the occurrence of a vertex labeled (a). As no (a) labeled vertices are assumed to occur along $\pi$, no (b) vertices with $\rank{v}=2$ occur along $\pi$.\\
\item Now assume that $v\in V_1\cap Z^*$ with $\rank{v}=2$. If $v$ is a (c) vertex all successor states will have rank $1$. With the same reasoning as before, this cannot occur.\\
\item Now assume that $v\in V_1\cap Z^*$ with $\rank{v}=2$ is labeled with ($\ell$). In this case there surely exists a successor $v'$ of $v$ s.t. $(v,v')\in E^\ell$ and $\rank{v'}=1$. But there \emph{might} also exist another successor $v''$ of $v$ (i.e., ($v''\in E(v)$) s.t. $\rank{v''}>1$. If there does not exists such a successor $v''$, all successors have rank $1$ and we again cannot visit $v$.\\
\item Now assume that $v\in V_1\cap Z^*$ with $\rank{v}=2$, labeled with ($\ell$) and there exists a successor $v''\in E(v)$ s.t. $\rank{v''}>1$. Now let us assume that such a state $v$ is visited infinitely often along $\pi$. As $\pi$ is a fair adversarial play over $G$ we know that visiting $v$ infinitely often along $\pi$ implies that $v'$ with $(v,v')\in E^\ell$ and $\rank{v'}=1$ (which surely exists by the definition of $\apre$) will also be visited infinitely often along~$\pi$. This is again a contradiction to the above hypothesis and implies that such $v$'s can only be visited finitely often.\\
\item As $V$ is a finite set, the set of states with rank $2$ is finite. Hence, the occurrence of infinitely many states with rank $2$ along $\pi$ implies that one of the above cases must occur infinitely often, which gives a contradiction to the above hypothesis.
\end{inparaitem}
Using the same arguments, we can inductively show that states with any fixed rank can only occur finitely often if states with rank~$1$ (i.e., (a)-labeled vertices) never occur. As the maximal rank is finite (due to the finiteness of $V$) this contradicts the assumption that $\pi$ is an infinite play. 

We therefore conclude that along any \emph{infinite} fair adversarial play $\pi$ with infinitely many vertices labeled by ($\ell$) we will eventually see a vertex labeled by (a). 

\smallskip
\noindent\textbf{Completeness.}
We now show that the fixpoint in \eqref{p:equ:Reach:FP} is complete, i.e., that every state in $\Zo^*:=V\setminus Z^*$ is loosing for $\p{0}$. In particular, we show that from every vertex $v\in\Zo^*$ \p{1} has a memoryless strategy $\rho_1$ s.t.\ all fair adversarial plays compliant with $\rho_1$ satisfy
\begin{align}\label{equ:Reach:psio}
 \overline{\psi}:=\lnot\psi=\lnot(Q\mathcal{U}T)=\Box\lnot T \vee \lnot T\mathcal{U}\lnot Q
\end{align}
and are hence loosing for $\p{0}$.

In order to prove the latter claim we fist compute $\Zo^*:=V\setminus Z^*$ by negating the fixpoint formula in \eqref{p:equ:Reach:FP}. For this, we define $\Xo^*:=V\setminus X$, $\Yo^*:=V\setminus Y$ and
use the negation rule of the $\mu$-calculus, i.e., $\lnot(\mu X.f(X))=\nu\Xo.V\setminus f(X)$ along with common De-Morgan laws. This results in the following derivation.
\begin{align*}
 \Zo^*&=\mu\Yo.~\nu\Xo.~\To \cap (\Qo\cup V\setminus\apre(Y,X))
\end{align*}
where
\begin{align*}
 &V\setminus\apre(Y,X)\\
 &=V\setminus \left[\cpre(X)\cup\left(\elpre(X)\cap\eapre_1(Y)\right)\right]\\
 &=\left[V\setminus\cpre(X)\right]\cap \left[V\setminus\left(\elpre(X)\cap\eapre_1(Y)\right)\right]\\
 &= \left[\epre_1(\Xo)\cup\eapre_0(\Xo)\right]\cap \left[V_0\cup (V_1\setminus\Vl)\cup\left(\Vl\setminus\left(\elpre(X)\cap\eapre_\ell(Y)\right)\right)\right]\\
 &= \left[\epre_1(\Xo)\cup\eapre_0(\Xo)\right]\cap \left[V_0\cup (V_1\setminus\Vl)\cup\left(\alpre(\Xo)\cup\epre_\ell(\Yo)\right)\right]\\
  &=\eapre_0(\Xo)\cup\epre_{1\setminus \ell}(\Xo)\cup\left[\epre_1(\Xo)\cap\left(\alpre(\Xo)\cup\epre_\ell(\Yo)\right)\right]\\
 &=\eapre_0(\Xo)\cup\epre_{1\setminus \ell}(\Xo)\cup\left[\epre_\ell(\Xo)\cap\left(\alpre(\Xo)\cup\epre_\ell(\Yo)\right)\right]\\
  &=\eapre_0(\Xo)\cup\epre_{1\setminus \ell}(\Xo)\cup\alpre(\Xo)\cup\epre_{\ell}(\Yo).
\end{align*}
The last line in the above derivation follows from the observation that $\alpre(\Xo)\subseteq\epre_l(\Xo)$ and $\Yo\subseteq\Xo$ for all iterations of the fixpoint computation.
The additionally introduced pre-operators are defined in close analogy to \eqref{equ:pre} and \eqref{equ:ell_pre} as follows:
\begin{align*}
\epre_1(S)&\coloneqq \SetComp{v\in V_1}{E(v)\cap S\neq \emptyset },\\
\eapre_0(S)&\coloneqq \SetComp{v\in V_0}{E(v)\subseteq S},\\
 \epre_{1\setminus \ell}(S)&\coloneqq \SetComp{v\in V_1\setminus \Vl}{E(v)\cap S\neq \emptyset },\\
\epre_{\ell}(S)&\coloneqq\SetComp{v\in \Vl}{E(v)\cap S\neq \emptyset },\\
\eapre_{\ell}(S)&\coloneqq \SetComp{v\in \Vl}{E(v)\subseteq S},\\
\alpre(S)&\coloneqq \SetComp{v\in \Vl}{\El(v)\subseteq S}.
\end{align*}

With this, we can conclude that 
\begin{align}\label{p:equ:Reach:FPinv}
 \Zo^*&=\mu\Yo.~\nu\Xo.~\To \cap \left(\Qo\cup\eapre_0(\Xo)\cup\epre_{1\setminus \ell}(\Xo)\cup\alpre(\Xo)\cup\epre_l(\Yo)\right).
\end{align}
where $\To=V\setminus T$ and $\Qo=V\setminus Q$.

Now denote by $\Yo^m$ the $m$-th iteration over the fixpoint variable $\Yo$ in \eqref{p:equ:Reach:FPinv}, where $\Yo^0=\emptyset$. Further, we denote by $\Xo^{mi}$ the set computed in the $i$-th iteration over the fixpoint variable $\Xo$ in~\eqref{p:equ:Reach:FPinv} during the computation of  $\Yo^m$ where  $\Xo^{m0}=V$. 
After termination of the inner fixed-point over $\Xo^{mi}$ we have by construction that $\Yo^m=\Xo^{m*}$ and therefore
\begin{align}\label{p:equ:Reach:FPinv:Ym}
 \Yo^m=\To \cap \left(\Qo\cup\eapre_0(\Yo^m)\cup\epre_{1\setminus \ell}(\Yo^m)\cup\alpre(\Yo^m)\cup\elpre(\Yo^{m-1})\right).
\end{align}

Similar to the soundness proof, we define a ranking over $V$ induced by the iterations of the smallest fixed-point, which now is $\Yo$:
\begin{align*}
 \ranko{v}=m\leftrightarrow v\in\Yo^m\setminus\Yo^{m-1}\quad\text{and}\quad \ranko{v}=\infty~\leftrightarrow~v\notin\Zo^*.
\end{align*}
This ranking can now be used to define a memoryless \p{1} strategy $\rho_1:V_1\rightarrow V$ s.t.\
\begin{equation}\label{equ:Reach:rho:inv}
\rho_1(v)=\min_{(v,w)\in E} \ranko{w}.
\end{equation}

Towards proving that $\rho_1$ is winning for $\overline{\psi}$ in \eqref{equ:Reach:psio} we first observe that for every vertex $v\in \Zo^*$ exactly one of the following cases holds:\\
\begin{inparaenum}[(a)]
\item $v\in (V_0\cap \Zo^*\cap \To)$, i.e., $\ranko{v}<\infty$ and $v\in \Qo$ or for all $v'\in E(v)$ it holds that $\ranko{v'}\leq\ranko{v}$,\\ 
\item $v\in ((V_1\setminus \Vl) \cap \Zo^*\cap \To))$, i.e., $\ranko{v}<\infty$ and $v\in \Qo$ or there exists $v'\in E(v)$ s.t.\ $\ranko{v'}\leq\ranko{v}$, or\\
\item[($\ell_\forall$)] $v\in (\Vl\cap Z^*\cap\To)$ and $\ranko{v}<\infty$ and $v\in \Qo$ or for all $v'\in \El(v)$ holds that $\ranko{v'}\leq\ranko{v}$\\
\item[($\ell_\exists$)] $v\in (\Vl\cap Z^*\cap\To)$ and $\ranko{v}>1$ (and  $\ranko{v}<\infty$), and ($\ell_\forall$) does not hold, but there exists a $v'\in E(v)$ s.t.\ $\ranko{v'}<\ranko{v}$.
\end{inparaenum}

Using this observation, we now show that every fair adversarial play $\pi$ compliant with $\rho_1$ satisfies $\overline{\psi}$ in \eqref{equ:Reach:psio}, that is, either stays in $\To$ forever, or eventually visits $\Qo$ before visiting $T$.

First, observe that for every node $v\in\Zo^*$ one of the cases (a),(b),($\ell_\forall$), or ($\ell_\exists$) holds. If $v$ is an~(a) vertex, we see that either $v
\in \Qo$ or for all choices of $\p{0}$ (i.e., for any \p{0} strategy), the play remains in $\Zo^*\subseteq\To$. 
Further, it is obvious that $\rho_1$ ensures, that whenever a (b) vertex is seen, the play remains in $\Zo^*\subseteq\To$ if we do not already have $v
\in \Qo$. The same is true for ($\ell_\forall$) vertices.

Now consider a fair adversarial play $\pi$ that is compliant with $\rho_1$ and $\pi(0)\in\Zo^*\subseteq\To$. Then it follows from the above intuition that for all visits to (a),(b),($\ell_\forall$) we have two cases: 
(i) Either $\overline{\psi}$ is immediately true on $\pi$ by visiting $\Qo$ (and having been in $\Zo^*\subseteq\To$ in all previous time steps). In this case the suffix of $\pi$ is irrelevant, because $\p{0}$ has already lost (by visiting $\Qo$ without seeing $T$). Or
(ii) the play remains in $\Zo^*\subseteq\To$.
Now observe that this is also true for infinite visits to (a),(b),($\ell_\forall$) vertices. As $\pi$ is fair adversarial, visiting a ($\ell_\forall$) vertex infinitely often, implies that all live edges are taking infinitely often, which all ensure that the play remains in $\Zo^*\subseteq\To$ or is immediately lost by visiting $\Qo$. 
Therefore, the only interesting case occurs if $\pi$ visits ($\ell_\exists$) vertices. If such a vertex is visited finitely often, $\rho_1$ ensures that the play stays in $\Zo^*\subseteq\To$. However, if they are visited infinitely often, a live edge that leaves $\Zo^*$ will also be taken infinitely often. Hence, in order to ensure that $\pi$ is loosing for $\p{0}$, we need to show that $\rho_1$ enforces that ($\ell_\exists$) vertices are only visited finitely often. 

To see this, let $v$ be an ($\ell_\exists$) vertex and observe that $\ranko{v}$ is finite and larger than $1$. At the first visit of $\pi$ to $v$, $\rho_1$ decreases the rank as it chooses by definition one of the existentially quantified successors $v'\in \El(v)$ with $\ranko{v'}<\ranko{v}$. Now observe that for all other cases (a),(b),($\ell_\forall$) either $\Qo$ is visited and the play is immediately loosing for $\p{0}$ or the play is kept in $\Zo^*\subseteq\To$ and the strategy $\rho_1$ never increases the rank. As every vertex has a unique rank, $\rho_1$ ensures that every ($\ell_\exists$) vertex is visited at most once along every compliant fair adversarial play that remains in $\Zo^*\subseteq\To$. This proves the claim.

 \subsubsection{Proof of \REFthm{thm:SingleRabin}}\label{app:prop:SingleRabin}
\begin{theorem*}[\REFthm{thm:SingleRabin} restated for convenience]
 Let $\Gl =\tup{\game, \El}$ be a game graph with live edges and $Q,G\subseteq V$ be two state sets over $\game$. Further, let
\begin{equation}\label{p:equ:SR:FP}
 Z^*\coloneqq\nu Y.~\mu X.~Q\cap
 \left[ 
 \left( G\cap \cpre(Y)\right)
 \cup \left(\apre(Y,X)\right)
 \right].
\end{equation}
Then $Z^*$ is equivalent to the winning region of $\p{0}$ in the fair adversarial game over $\Gl$ for the winning condition $\psi$ in \eqref{equ:SR:psi}. 
Moreover, the fixpoint algorithm runs in $O(n^2)$ symbolic steps, and a memoryless winning strategy for $\p{0}$ can be extracted from it.
\end{theorem*} 
 
 In order to simplify the proof of \REFprop{app:prop:SingleRabin}, we first prove the following lemma.
 
 \begin{lemma}\label{lem:FPsimpleSingleRabin}
  Let $Q,G\subseteq V$ and   
  \begin{subequations}\label{equ:Bcomp}
   \begin{align}
  \textstyle Z^*\coloneqq&\nu Y.\mu X.
Q\cap
 \left[ 
 \left( G\cap \cpre(Y)\right)
 \cup \apre(Y,X)
 \right]\label{equ:Bcomp:a}\\
  \textstyle \Zt^*\coloneqq&\nu \Yt.\nu Y.\mu X.Q\cap
 \left[ 
 \left( G\cap \cpre(\Yt)\right)
 \cup \apre(Y,X)
 \right].\label{equ:Bcomp:b}
\end{align}
\end{subequations}
Then $Z^*=\Zt^*$.
 \end{lemma}
 
 \begin{proof}
 To prove the claim we consider a third version of the fixpoint algorithm, namely
 \begin{equation*}
    \textstyle \Zc^*\coloneqq\nu \Yt.\nu Y.\mu X.Q\cap
 \left[ 
 \left( G\cap \cpre(\Yt)\right)
 \cup
  \left( G\cap \cpre(Y)\right)
 \cup \apre(Y,X)
 \right].
 \end{equation*}

 Then it immediately follows from the monotonicity of all involved functions that $\Zt^*\subseteq \Zc^*$. It further follows from \REFlem{lem:contain} (iv) that $Z^*=\Zc^*$. It therefore remains to show that $\Zc^*\subseteq \Zt^*$ to prove the claim. We actually show $\Zc^*\subseteq \Zt^*$.

 Let $\Yt^0=Y^{00}=V$. Then it immediately follows that the computation of $X^{00*}$ returns the same set for both fixed-points. It further follows that $Y^{0n}\subseteq \Yt^0$, which implies $( G\cap \cpre(Y^{0n}))\subseteq  ( G\cap \cpre(\Yt^0))$ and therefore the set $\Yt^1$ coincides for both fixed-points. Now recall from \cite{long1994improved} that warm-starting the inner fixpoint computation with the largest fixed-point retained from previous values of outer fixpoint variables, does not change the resulting fixed-point. With this, we can use $Y^{10}=\Yt^{1}$ and observe that this implies that the computation of $\Yt^2$ becomes again identical for both fixed-points. Re-applying this argument until termination shows, that indeed $\Zc^*\subseteq \Zt^*$.
\end{proof}

With \REFlem{lem:FPsimpleSingleRabin} in place, we can use \eqref{equ:Bcomp:b} instead of \eqref{p:equ:SR:FP} to prove \REFthm{thm:SingleRabin}. 
Further, let us define $Z^*(\tuple{T,Q})$ to be the set of states computed by the fixpoint algorithm in \eqref{equ:Reach:FP}. Then we know that upon termination we have
 \begin{equation}\label{equ:proof:SR:Ytstar}
  \Zt^*=\Yt^*=Z^*(\tuple{Q\cap G\cap \cpre(\Yt^*),Q}).
 \end{equation}
Now we will use \eqref{equ:proof:SR:Ytstar} to prove soundness and completeness of \REFthm{thm:SingleRabin}.

\smallskip
\noindent\textbf{Soundness}
Let us now define $T:=Q\cap G\cap \cpre(\Yt^*))$.
Pick any state $v\in\Zt^*$ and the strategy~$\rho_0$ defined as in \eqref{equ:Reach:rho} over the sets $X^i$ computed in the last iteration over $X$ when computing $Z^*(\tuple{T,Q})$. Further, let $\pi$ be an arbitrary fair adversarial play starting in $v$ and being compliant with $\rho_0$. Then we need to show that $\pi$ fulfills $\psi$ in \eqref{equ:SR:psi}. 

Using \eqref{equ:proof:SR:Ytstar} and the fact that $v\in\Zt^*$ we know from \REFthm{thm:Reachability} that $\pi$ fulfills $Q\mathcal{U}T$. That is, there exists a $k\in\mathbb{N}$ s.t.\ $\pi(i)\in Q$ for all $i<k$ and $\pi(k)\in T=Q\cap G\cap \cpre(\Yt^*))$. With this we know that (a) $\pi(k)\in Q$, (b) $\pi(k)\in G$ and (c) $v\in \cpre(\Yt^*)$. Now we have two cases: (c.1) If $\pi(k)\in V^1$, then it follows from the definition of $\cpre$ that $E(\pi(k))\subseteq \Yt^*$. As $\Yt^*=\Zt^*$, we know $\pi(k+1)\in \Zt^*$. (c.2) If $\pi(k)\in V^0$ we know that $\rank{\pi(k)}=min_{v'\in E(\pi(k))}\rank{v'}$. Now recall that $\Zt^*=\Yt^*=Y^*=\bigcup_{i>0} X^i$. Hence, any state with rank $0<n<\infty$ is contained in $\Zt^*$ and hence, we have $\pi(k+1)\in \Zt^*$. With this, we can successively re-apply \REFthm{thm:Reachability} to $\pi(k+1)$. This shows that $G$ is visited infinitely often along $\pi$ while $\pi$ always remains within $Q$.

\smallskip
\noindent\textbf{Completeness}
Let $\mathcal{W}\subseteq V$ be the set of states from which $\p{0}$ has a winning strategy w.r.t.~$\psi$ in \eqref{equ:SR:psi}. In order to prove completeness, we need to show that $\mathcal{W}\subseteq Z^*$.

Recall, that for all states $v\in \mathcal{W}$ there exists a strategy $\rho_0$ s.t.\ all compliant fair adversarial plays $\pi$ fulfill $\psi$. Now consider the weaker LTL formula $\widetilde{\psi}:=Q\mathcal{U}(Q\cap G)$ and let $\widetilde{\mathcal{W}}$ be the winning state set for $\widetilde{\psi}$. 
Then we know by construction that $\widetilde{\psi}$ holds for $\pi(0)$ and for every $\pi(k)\subseteq Q\cap G$ while $\pi$ always remains in $Q$. We can therefore strengthen $\widetilde{\psi}$ to $\widetilde{\psi}:=Q\mathcal{U} (Q\cap G\cap \cpre(\widetilde{\mathcal{W}}))$ and see that still $\psi\rightarrow\widetilde{\psi}$ and therefore $\mathcal{W}\subseteq\widetilde{\mathcal{W}}$. 

Now observe that it follows from \REFthm{thm:Reachability} that $\widetilde{\mathcal{W}}= Z^*(\tuple{Q\cap G\cap \cpre(\widetilde{\mathcal{W}}),Q})$. It further follows from the monotonicity of the $\mu$-calculus formula that $\Zt^*$ is the \emph{largest} set of states s.t.\ equality holds in \eqref{equ:proof:SR:Ytstar}. We therefore have to conclude that $\widetilde{\mathcal{W}}\subseteq \Zt^*$. As we have shown that $\mathcal{W}\subseteq\widetilde{\mathcal{W}}$, the claim is proved.


\subsection{Proof of Theorem~\ref{thm:FPsoundcomplete}}\label{appD}

\begin{theorem*}[\REFthm{thm:FPsoundcomplete} restated for convenience]
Let $\Gl =\tup{\game, \El}$ be a game graph with live edges and 
$\FR$ be a Rabin condition over $\game$ with index set $P=[1;k]$. 
Further, let
 \begin{subequations}\allowdisplaybreaks
   \begin{align*}
 \textstyle Z^*\coloneqq\nu Y_{p_0}.\mu X_{p_0}.&\textstyle \bigcup_{p_1\in P} \nu Y_{p_1}.\mu X_{p_1}.\\
 &\textstyle \bigcup_{p_2\in P\setminus\set{p_1}} \nu Y_{p_2}.\mu X_{p_2}.\notag\\
&\qquad\vdots\notag\\
 &\textstyle \bigcup_{p_k\in P\setminus\set{p_1,\hdots, p_{k-1}}}\nu Y_{p_k}.\mu X_{p_k}.
\left[\bigcup_{j=0}^k \mathcal{C}_{p_j}\right],\notag
\end{align*}
where
 \begin{align*}
 \mathcal{C}_{p_j}\coloneqq
\bigcap_{i=0}^{j} \overline{R}_{p_i}\cap
 \left[ 
 \left( G_{p_j}\cap \cpre(Y_{p_j})\right)
 \cup \left(\apre(Y_{p_j},X_{p_j})\right)
 \right],
\end{align*}
 \end{subequations}
with $p_0=0$, $G_{p_0}\coloneqq\emptyset$ and $R_{p_0}\coloneqq\emptyset$. 
Then $Z^*$ is equivalent to the winning region $\WlR$ of $\p{0}$ in the fair adversarial game over $\Gl$ for the winning condition $\vR$ in \eqref{equ:vR}. Moreover, the fixpoint algorithm runs in $O(n^{k+2}k!)$ symbolic steps, and a memoryless winning strategy for $\p{0}$ can be extracted from it.
\end{theorem*}

This section contains the proof of \REFthm{thm:FPsoundcomplete} which is inspired by the proof of \citet{PitermanPnueli_RabinStreett_2006} for \enquote{normal} Rabin games. We first give a construction of a ranking induced by the fixpoint algorithm in \eqref{equ:Rabin_all} in \REFsec{sec:theory:strategy}, and use this ranking to define a memoryless $\p{0}$ strategy. As part of the soundness proof for \REFthm{thm:FPsoundcomplete} in \REFsec{sec:theory:sound},  we then show that this extracted strategy is indeed a winning strategy of $\p{0}$ in the \emph{fair adversarial game} over~$\Gl$ w.r.t.\ $\vR$. Further, we show in \REFsec{sec:theory:complete} that the fixpoint algorithm in \eqref{equ:Rabin_all} is also complete,  that is $\WlR\subseteq Z^*$. Intuitively, completeness shows that if $Z^*$ is empty, there indeed exists no live-sufficient winning strategy (with arbitrary memory) for the given fair adversarial Rabin game. Additional lemmas and proofs can be found in \REFapp{appD:additional}. 
The time complexity of the algorithm is proven separately in \REFapp{app:acceleration}.

\subsubsection{Strategy Extraction}\label{sec:theory:strategy}
Our strategy extraction is adapted from the ranking of \citet[Section~3.1]{PitermanPnueli_RabinStreett_2006}.
Recall, that we consider the set of Rabin pairs $\FR=\set{\tuple{G_1,R_1},\hdots,\tuple{G_k,R_k}}$ with index set $P=\set{1,\hdots,k}$ and the artificial Rabin pair $\tuple{G_0,R_0}$ s.t.\ $G_0=R_0=\emptyset$. 
A \emph{permutation} of the index set $P$ is an one-to-one and onto function from $P$ to $P$; as usual, we write $p_1\ldots p_k$ to denote the permutation
mapping $i$ to $p_i$, for $i =1,\ldots, k$.  We define $\Pi(P)$ to be the set of all permutations over $P$. 
The \emph{configuration domain} of the Rabin condition $\FR$ is defined as 
\begin{align}
\label{equ:RabinConfDomain}
D(\FR):=&\left\{p_0 i_0 p_1 i_1 \hdots p_k i_k~\mid~i_j\in[0;n], ~p_0=0,~ p_1\hdots p_k\in\Pi(P)\right\}\cup \set{\infty}
\end{align}
where $n<\infty$ is a natural number which is larger then the maximal number of iterations needed in any instance of the fixpoint computation in \eqref{equ:Rabin_all} which is known to be finite. 
If $\FR$ is clear from the context, we write $D$ instead of $D(\FR)$. 

\smallskip
\noindent\textbf{Intuition:}
We first explain the intuition behind the chosen ranking. For this we consider the definition of ranks for states $v\in Z^*$ in an iterative fashion. First, consider the last iteration over~$X_{p_0}$ converging to the fixed-point $Z^*=Y_{p_0}^*=\bigcup_{i_0>0} X_{p_0}^{i_0}$ where $X_{p_0}^{0}:=\emptyset$. By flattening \eqref{equ:Rabin_all} we see that for all $i_0>0$ we have 
\begin{subequations}\label{equ:Xdeltaip}
 \begin{align}\label{equ:Xdeltaip:0}
 X_{p_0}^{i_0}&=\apre(Y_{p_0}^*,X_{p_0}^{i_0-1})\cup\mathcal{A}_{p_0 i_0}
\end{align}
where $\mathcal{A}_{p_0 i_0}$ collects all remaining terms of the fixpoint algorithm in \eqref{equ:Rabin_all} and will be specified later. For now, we want to assign a \enquote{minimal rank} to all states added to $Z^*$ via the first term in \eqref{equ:Xdeltaip:0}. Let us assume that the right \enquote{minimal rank} for these states is 
\begin{equation*}
 d=p_0 i_0 p_1 0\hdots p_k 0\quad\text{with}\quad p_1<p_2<\hdots<p_k~\text{and}~i_0>0.
\end{equation*}
We assign this rank to $v$ iff  $v\in \apre(Y_{p_0}^*,X_{p_0}^{i_0-1})\setminus X_{p_0}^{i_0-1}$, i.e., if $v$ is not already added to the fixed-point in a previous iteration. The intuition behind this rank choice is that we want to remember  that we have added $v$ to $Z^*$ in the $i_0$'s computation over $X_{p_0}$, which sets the counter for $p_0$ in $d$ to $i_0$. We keep all other counters at $0$ because there is no actual contribution of terms involving variables $X_{p_i}$ for $p_i\in P$ for the \enquote{adding} of $v$. 

Now recall that 
\begin{equation*}
   X_{p_0}^{i_0}= \bigcup_{p_1\in P} Y^{*}_{p_1}
  =\bigcup_{p_1\in P} \bigcup_{i_1>0} X^{i_1}_{p_1}.
\end{equation*}
Further, we know that 
\begin{equation}\label{equ:rank_guess:0}
 \apre(Y_{p_0}^*,X_{p_0}^{i_0-1})\subseteq X^{i_1}_{p_1} \quad\text{for all}\quad p_1\in P~\text{and}~i_1>0.
\end{equation}
Hence, any state added to the fixed-point via $X_{p_0}^{i_0}$ (which is not contained in $X_{p_0}^{i_0-1}$) is either added via $\apre(Y_{p_0}^*,X_{p_0}^{i_0})$ or via any other remaining term within $X^{i_1}_{p_1}$ for at least one $p_1$ and $i_1>0$. So let us explore the ranking in the latter case.

For this, let us proceed by going over all $X^{i_1}_{p_1}$ in increasing order over $P$, i.e, we start with selecting $p_1=1$. Further, we remember that we compute the next iteration over $X_{p_1}$ (i.e., $X^{i_1}_{p_1}$ given $X^{i_1-1}_{p_1}$) as part of computing the set $X_{p_0}^{i_0}$. I.e., we remember the \emph{computation-prefix} $\delta=p_0 i_0$ in the computation of $X^{i_1}_{p_1}$. To make $\delta$ explicit, we denote $X^{i_1}_{p_1}$ by $X^{i_1}_{\delta p_1}$. Now, we again consider the last iteration over $X_{\delta p_1}$ converging to the fixed-point $Y^*_{\delta p_1}$ (for the currently considered computation-prefix $\delta$). Then we have 
\begin{align*}
 X^{i_1}_{\delta p_1}=
 &\underbrace{\apre(Y_{p_0}^*,X_{p_0}^{i_0-1})}_{=:S_\delta}
 \cup
 \underbrace{
\overline{R}_{p_1}\cap \left[\left(G_{p_1}\cap \cpre(Y_{\delta p_1}^*)\right)\cup\apre(Y_{\delta p_1}^*,X_{\delta p_1}^{i_1-1})\right]
}_{=:\mathcal{C}_{\delta p_1 i_1}}
\cup \mathcal{A}_{\delta p_1 i_1}.
\end{align*}

We now want to assign the \enquote{minimal rank} to all states that are added to the fixed-point via $\mathcal{C}_{\delta p_1 i_1}$. The immediate choice of this rank is 
\begin{align}
\label{equ:rank_guess:1}
 &d=p_0 i_0 p_1 i_1 p_2 0 \hdots p_k 0 =\delta p_1 i_1 p_2 0 \hdots p_k 0
 \quad\text{with}~ p_2<\hdots<p_k\quad\text{and}~i_0,i_1>0.
\end{align}
(Note that we do not necessarily have $p_1<p_2$!) 

We only want to assign this rank to states that are actually added to the fixed-point via~$\mathcal{C}_{\delta p_1 i_1}$, i.e., do not already have a rank assigned. First, all states $v\in S_\delta$ already have an assigned rank (as discussed before). Second, for $i_1>1$ all states in $\mathcal{C}_{\delta p_1 i_1-1}$ have already an assigned rank.  But, third, also all states that have been added by considering a different $X_{\tilde{p}_1}$ with $\tilde{p}_1\in P$ being smaller then the currently considered $p_1$ also have an already assigned rank. 

Now consider the ranking choices suggested in \eqref{equ:rank_guess:0} and \eqref{equ:rank_guess:1}. Then we see that all already assigned ranks are \emph{smaller} (in terms of the lexicographic order over $D$) than the one in \eqref{equ:rank_guess:1}. To see this, first consider a state $v\in S_\delta$. Either, $v\in X_{p_0}^{i_0-1}$ in which case its $0$'th counter is smaller then $i_0$ (i.e., $i_0-1<i_0$) or $v$ has been added via $S_\delta$, in which case the $0$'th counter is equivalent but the first counter is $0$ and therefore smaller then $i_1$ in \eqref{equ:rank_guess:1} (as, $i_1>0$). 
Now consider a state $v\in X_{\tilde{p}_1}$ with $\tilde{p}_1<p_1$. In this case we see that $0$'th counter is equivalent but the first permutation index is smaller (as 
$\tilde{p}_1<p_1$). 

We can therefore avoid specifying exactly in which set $v$ should \emph{not} be contained to be a newly added state. We can simply collect all possible rank assignments for every state and then, post-process this set to select the smallest rank in this set. Let us now generalize this idea to all possible configuration prefixes.

\end{subequations}

\begin{proposition}\label{prop:flattening}
 Let $\delta=p_0i_0\hdots p_{j-1} i_{j-1}$ be a \emph{configuration prefix}, $p_j\in P\setminus\set{p_1,\hdots,p_{j-1}}$ the next permutation index and $i_{j}>0$ a counter for $p_j$. Then the flattening of \eqref{equ:Rabin_all} for this configuration prefix is given by
\begin{subequations}\label{equ:It_pj}
\begin{align}\label{equ:Xdeltai}
 X_{\delta p_j}^{i_j}=
 &
 \underbrace{
 S_{\delta}\cup
\mathcal{C}_{\delta p_j i_j}
 }_{S_{\delta p_j i_j}}
 \cup\mathcal{A}_{\delta p_j i_j}
\end{align}
where 
 \begin{align}
&Q_{p_0\hdots p_a}:=
\bigcap_{b=0}^{a} \overline{R}_{p_b},\label{equ:Qdelta}\\ 
&\mathcal{C}_{\delta p_a i_a}:=\left(Q_{\delta p_a}
 \cap G_{p_a}\cap \cpre(Y_{\delta p_a}^*)\right)\cup \left(Q_{\delta p_a}
 \cap \apre(Y_{\delta p_a}^*,X_{\delta p_a}^{i_a-1})\right),\\
 &S_{p_0i_0\hdots p_{a}i_{a}}:=\bigcup_{b=0}^{a} \mathcal{C}_{p_0i_0\hdots p_b i_b},\label{equ:Sdelta}\\
&\mathcal{A}_{\delta p_j i_j}:=\bigcup_{p_{j+1}\in P\setminus\set{p_1,\hdots,p_{j}}} \bigcup_{i_{j+1}>0} 
\left(X^{i_{j+1}}_{\delta p_j i_j p_{j+1}}
\setminus S_{\delta p_j i_j}\label{equ:Adeltai}
\right)
\end{align}
\end{subequations}
\end{proposition}
As this flattening follows directly from the structure of the fixpoint algorithm in \eqref{equ:Rabin_all} and the definition of $\mathcal{C}_{p_j}$ in \eqref{equ:Rabin_all_Cpj}, the proof is omitted.

Using the flattening of \eqref{equ:Rabin_all} in \eqref{equ:It_pj} we can define a ranking function induced by \eqref{equ:Rabin_all} as follows.

\begin{definition}\label{def:ranking}
Given the premises of \REFprop{prop:flattening}, we define $\underline{\gamma}:=p_{j+1}0p_{j+2}0\hdots p_k 0$ with $p_{j+1}<p_{j+2}<\hdots <p_k$ to be the minimal configuration post-fix. Then we define the rank-set $R:V\rightarrow 2^D$ s.t.\
 \begin{inparaenum}[(i)]
  \item $\infty \in R(v)$ for all $v\in V$, and
  \item $\delta p_j i_j \underline{\gamma}\in R(v)$ iff $v\in S_{\delta p_j i_j}$.
 \end{inparaenum}
 The ranking function $\rank{}:V\rightarrow D$ is defined s.t. $\rank{}: v \mapsto \min\set{R(v)}$.
\end{definition}

Based on the ranking in \REFdef{def:ranking} we define a memory-less player $0$ strategy $\rho_0$, s.t.\ $\rho_0(v)$ forces progress to a state reachable from $v$ which has minimal rank compared to all other successors of $v$. We prove \REFthm{thm:strategy} in \REFsec{sec:theory:sound}.

\begin{theorem}\label{thm:strategy}
Given the premises of \REFprop{prop:flattening}, the memoryless player $0$ strategy $\rho_0:V^0\cap Z^*\rightarrow V^1$ s.t.\ 
 \begin{align}\label{equ:strategy}
 \rho_0(v):=\min_{(v,w)\in E} (\rank{w}),
\end{align}
is a winning strategy for player $0$ in the \emph{fair adversarial game} over $\Gl$ w.r.t.\ $\vR$.
\end{theorem}

\begin{example}\label{example:appendix:strategy}
 Consider the Rabin game depicted in \REFfig{fig:example} and discussed in \REFapp{app:ExpFP}. 
 Here, the strategy construction outlined in \REFthm{thm:strategy} enforces a transition from $q_6$ to $q_7$ and a transition from $q_5$ to $q_3$. This is observed by noting that $\rank{q_2}=002012$ and $\rank{q_7}=001121$ where $\rank{q_7}<\rank{q_2}$. In addition, $\rank{q_1}=011021$ and $\rank{q_3}=001121$, where $\rank{q_3}<\rank{q_1}$. 
\end{example}

\subsubsection{Soundness}\label{sec:theory:sound}
We now show why the fixpoint algorithm in \eqref{equ:Rabin_all} is \emph{sound}, i.e., why $Z^*\subseteq \WlR$ in 
\REFthm{thm:FPsoundcomplete} holds. In addition, we also show that \REFthm{thm:strategy} holds.  

We prove soundness by an induction over the nesting of fixed-points in \eqref{equ:Rabin_all} from inside to outside. In particular, we iteratively consider instances of the flattening in \eqref{equ:It_pj}, starting with $j=k$ as the base case, and doing an induction from \enquote{$j+1$} to \enquote{$j$}.
To this end, we consider a local winning condition which refers to the current configuration-prefix $\delta=p_0i_0\hdots p_{j-1}i_{j-1}$ in~\eqref{equ:It_pj}, namely
\begin{align}\label{equ:psi_delta}
\psi_{\delta p_j}:=
\left(
\begin{array}{rl}
&Q_{\delta p_j}\mathcal{U} S_{\delta}\\
\vee &\Box Q_{\delta p_j}\wedge \Box\Diamond G_{p_j}\\
\vee &\Box Q_{\delta p_j} \wedge \left(\bigvee_{i\in P\setminus\set{p_0,\hdots,p_{j}}}\left(\Diamond\Box\overline{R}_{i}\wedge \Box\Diamond G_{i}\right)\right) 
\end{array}
\right).
\end{align}
Further, we denote by $\mathcal{W}_{\delta p_j}$ the set of states for which player $0$ wins the \emph{fair adversarial game} over $\Gl$ w.r.t.\ $\psi_{\delta p_j}$ in \eqref{equ:psi_delta}. 

By recalling that for $p_j=p_0=0$ we have $Q_{p_0}=V$,  $S_{\varepsilon}=\emptyset$ and $G_{p_0}=\emptyset$, we see that for $j=0$ the condition in \eqref{equ:psi_delta} simplifies to 
\begin{align*}
\psi_{p_0}=\bigvee_{i\in P}\left(\Diamond\Box\overline{R}_{i}\wedge \Box\Diamond G_{i}\right).
\end{align*}
This implies that $\psi_{p_0}$ is equivalent to $\vR$ in \eqref{equ:vR}. Given this observation, the proof of soundness in \REFthm{thm:FPsoundcomplete} proceeds by inductively showing that 
\begin{equation}\label{equ:claim_pj}
X_{\delta p_j}^{i_j}\subseteq  \mathcal{W}_{\delta p_j}
\end{equation}
for any configuration prefix $\delta$, next permutation index $p_j$ and counter $i_j>0$. Thereby, we ultimately also prove this claim for $p_j=p_0=0$ where $\delta$ is the empty string and $Y_{p_0}^*=\bigcup_{i_0>0}X_{p_0}^{i_0}$ coincides with $Z^*$ in \eqref{equ:Rabin_all}, which proves the statement. 
	
With this insight the proof of \REFthm{thm:strategy} as well as the soundness part of \REFthm{thm:FPsoundcomplete} reduce to the following proposition.

\begin{proposition}\label{prop:soundness_j}
 For all $j\in [0,k]$, computation-prefixes $\delta=p_0i_0\hdots p_{j-1}i_{j-1}$, next permutation index $p_j\in P\setminus \set{p_0,\hdots,p_{j-1}}$, counter $i_j>0$ and state $v\in X^{i_j}_{\delta p_j}$ the strategy $\rho_0$ in \eqref{equ:strategy} wins the fair adversarial game over $\Gl$ w.r.t.\ $\psi_{\delta p_j}$ in \eqref{equ:psi_delta}.
\end{proposition}

To see why \REFprop{prop:soundness_j} holds, we consider the computation of $X_{\delta p_j}^{i_j+1}$ in \eqref{equ:Xdeltai} and observe that the states in $X_{\delta p_j}^{i_j+1}$ can be clustered based on their rank induced via \REFdef{def:ranking} as follows (see \REFsec{appD:case_split_sound} for a full proof).

 \begin{proposition}\label{prop:case_split_sound}
 Given the premisses of \REFprop{prop:soundness_j}, let \begin{align*}
 \underline{\gamma}&=p_{j+1}0p_{j+2}0\hdots p_k 0
 &&\text{with}\quad p_{j+1}<p_{j+2}<\hdots <p_k, \quad\text{and}\\
 \overline{\gamma}&=p_{j+1}np_{j+2}n\hdots p_k n
  &&\text{with}\quad p_{k}<p_{k-1}<\hdots <p_{j+1}
\end{align*}
 be the \emph{minimal} and \emph{maximal} post-fix, respectively.
 Then, for all $v\in X_{\delta p_j}^{i}$ exactly one of the following cases holds:
 \begin{compactenum}[(a)]
  \item $v\in S_\delta$ and $\rank{v}\leq \delta p_j 0\underline{\gamma}$,
  \item $v\in  Q_{\delta p_j} \cap G_{p_j}\cap\cpre(Y^*_{\delta p_j})$ and $\rank{v}=\delta p_j 1\underline{\gamma}$,
  \item $v\in  Q_{\delta p_j}\cap \apre(Y^*_{\delta p_j},X^{i_j-1}_{\delta p_j})$ and $\rank{v}=\delta p_j i_j\underline{\gamma}$ s.t.\ $i_j>1$, or
  \item $v\in \mathcal{A}_{\delta p_j i_j}$ and there exists $\underline{\gamma}<\gamma'\leq \overline{\gamma}$ s.t.\ $\rank{v}= \delta p_j i_j\gamma'$. 
 \end{compactenum}
 \end{proposition}
Using \REFprop{prop:case_split_sound} we prove \REFprop{prop:soundness_j} by an induction over $j$. 

\begin{proof}[Proof of \REFprop{prop:soundness_j}]
\smallskip
\noindent\textbf{Base case:}
First, for $j=k$ the last line of \eqref{equ:psi_delta} disappears. Then the proof reduces to \REFthm{thm:Reachability} and \REFthm{thm:SingleRabin} in the following way. First, we fix all fixpoint variables $Y^*_{p_0\hdots p_l}$ and $X^{i_l}_{p_0\hdots p_l}$ for $l<j$ as well as $Y^*_{\delta p_j}$. With this, we see that 
$T:=S_\delta\cup (Q_{\delta p_j} \cap G_{p_j}\cap\cpre(Y^*_{\delta p_j}))$ becomes a fixed set of states and \eqref{equ:Xdeltai} reduces to
\begin{equation*}
 X^{i_j}_{\delta p_j}=T\cup (Q_{\delta p_j} \cap \apre(Y^*_{\delta p_j},X^{i_j-1}_{\delta p_j}))
\end{equation*}
where we know that $X^{i_j}_{\delta p_j}\subseteq Y^*_{\delta p_j}$. Further, it follows form \REFprop{prop:case_split_sound} that for all $X^{i_j}_{\delta p_j}$ the ranking only differs by the $i_j$ count. Hence, we can replace $\rho_0$ in \eqref{equ:strategy} by the simpler strategy $\rho_0$ in \eqref{equ:Reach:rho} that only considers the $i_j$ count as the rank of states in $Y^*_{\delta p_j}=\bigcup_{i_j>0}X^{i_j}_{\delta p_j}$. With this it follows from \REFthm{thm:Reachability} that for any fair adversarial play $\pi$ compliant with $\rho_0$ in \eqref{equ:strategy} and starting in~$X^{i_j}_{\delta p_j}$ for some $i_j\geq 0$ it holds that $Q_{\delta p_j}\mathcal{U} T$. This implies that whenever such a play $\pi$ eventually reaches a state in $S_\delta\subseteq T$ the first line of \eqref{equ:psi_delta} holds. 

Now assume that $\pi$ does not reach a state in $S_\delta\subseteq T$. Then it reaches a state in $Q_{\delta p_j} \cap G_{p_j}\cap\cpre(Y^*_{\delta p_j})$ and therefore has a successor state $v'\in Y^*_{\delta p_j}=\bigcup_{i_j>0}X^{i_j}_{\delta p_j}$. Hence, $v'\in X^{i_j}_{\delta p_j}$ for some $i_j\geq 0$. By repeatedly applying this argument we see that $\pi$ either eventually reaches  a state in $S_\delta\subseteq T$ or it remains infinitely in $\mathcal{C}_{\delta p_j \cdot}$. In the latter case, it follows from \REFthm{thm:SingleRabin} that the second line of \eqref{equ:psi_delta} holds.

\smallskip
\noindent\textbf{Induction step:}
 For the induction step (from \enquote{$j+1$} to \enquote{$j$}) we first analyze the assumption. I.e., we know that for the longer computation prefix $\delta'=\delta p_j i_j$ and any next permutation index $p_{j+1}$ we have that $Y^*_{\delta' p_{j+1}}\subseteq \mathcal{W}_{\delta' p_{j+1}}$ for all $p_{j+1}\in P\setminus\set{p_1,\hdots,p_{j}}$. Now recall that \eqref{equ:Adeltai} implies 
 \begin{equation*}
  \mathcal{A}_{\delta p_j i_j}=\textstyle\bigcup_{p_{j+1}\in P\setminus\set{p_1,\hdots,p_{j}}} Y^*_{\delta' p_{j+1}}
\setminus S_{\delta p_j i_j}
 \end{equation*}
 and therefore, we know that for all $v\in \mathcal{A}_{\delta p_j i_j}$ there exists a $p_{j+1}$ s.t.\ $v\in\mathcal{W}_{\delta' p_{j+1}}$. That is, any fair adversarial play starting in $v$ that is compliant with $\rho_0$ in \eqref{equ:strategy} fulfills \eqref{equ:psi_delta}.

  Therefore, whenever a fair adversarial play $\pi$ starting in $X^{i_j}_{\delta p_j}$ visits a vertex $v\in \mathcal{A}_{\delta p_j i_j}$ (i.e., case (d) holds), we know that $\pi$ could possibly come back to a state 
  $v\in S_{\delta'p_{j+1}}= S_\delta\cup\mathcal{C}_{\delta p_j i_j}$
  (via the first line of $\psi_{\delta' p_{j+1}}$). 
  
  In this case, \REFprop{prop:case_split_sound} ensures that the $i_j$ count of the rank of states always stays constant while the play stays in $\mathcal{A}_{\delta p_j i_j}$. Therefore, one can ignore these finite sequences of (d) vertices in $\pi$ while applying the ranking arguments of \REFthm{thm:Reachability} and \REFthm{thm:SingleRabin}. I.e., we can conclude that in this case either the first or the second line of \eqref{equ:psi_delta} holds for $\pi$.
It remains to show that $\pi$ fulfills the last line of \eqref{equ:psi_delta} if $\pi$ eventually stays within $\mathcal{A}_{\delta p_j i_j}$ forever. First, observe that this is only possible if $S_\delta$ is not visited along $\pi$. Hence, we know that $Q_{\delta p_j}$ holds along $\pi$ until $\mathcal{A}_{\delta p_j i_j}$ is entered and never left. Further, as $\mathcal{A}_{\delta p_j i_j}$ is assumed to be never left after some time $k>0$, we know that from that time onward there exists no $p_{j+1}$ s.t.\ $S_{\delta'p_{j+1}}$ is visited again by $\pi$.  
This implies that for all vertices $\pi(k')$ with $k'>k$ the last two lines of $\psi_{\delta' p_{j+1}}$ (denoted $\psi'_{\delta' p_{j+1}}$) must be true for at lease one $p_{j+1}$. Hence, $\pi$ fulfills the property
\begin{subequations}\label{equ:Psi:outline}
\begin{align}\label{equ:Psi}
  \Psi_{\delta p_j}:=&\Box Q_{\delta p_j} \wedge\Diamond\underbrace{\left(\textstyle\bigvee_{p_{j+1}\in P\setminus\set{p_1,\hdots,p_{j}}}  \psi'_{\delta' p_{j+1}}\right)}_{\Psi'_{\delta p_j}}
\end{align}
With this, it remains to show that $\Psi_{\delta p_j}$ implies that the last line of \eqref{equ:psi_delta} is true for $\pi$. In particular, we can show that both statements are equivalent, i.e.,
  \begin{align}\label{equ:Psi:outline:b}
   \Psi_{\delta p_j}=&\Box Q_{\delta p_j}\wedge\bigvee_{p_{j+1}\in P\setminus\set{p_1,\hdots,p_{j}}} 
\left(\Diamond\Box\overline{R}_{p_{j+1}}\wedge \Box\Diamond G_{p_{j+1}}
 \right)
  \end{align}
Equation \eqref{equ:Psi:outline} is proved in \REFsec{appD:proofPsi}. This conclues the proof.
 \end{subequations} 
\end{proof}

\subsubsection{Completeness}\label{sec:theory:complete}
We now show why the fixpoint algorithm in \eqref{equ:Rabin_all} is complete,  i.e., why $\WlR\subseteq Z^*$ in \REFthm{thm:FPsoundcomplete} holds.  

We also prove completeness by an induction over the nesting of the fixpoints in \eqref{equ:Rabin_all} from inside to outside. In particular, we iteratively consider the fixed-points $Y^*_{\delta p_j}$ and show that $Y^*_{\delta p_j}\subseteq \mathcal{W}_{\delta p_j}$. As $\psi_{\delta p_j}$ simplifies to $\vR$ in \eqref{equ:vR} for $p_j=p_0=0$, we ultimately show that $\WlR\subseteq Z^*$ in \REFthm{thm:FPsoundcomplete}. 
With this insight the proof of the completeness part of \REFthm{thm:FPsoundcomplete} reduces to the following proposition. 

\begin{proposition}\label{prop:completeness_j}
 For all $j\in [0,k]$, computation-prefixes $\delta=p_0i_0\hdots p_{j-1}i_{j-1}$ and next permutation index $p_j\in P\setminus \set{p_0,\hdots,p_{j-1}}$ it holds that $\mathcal{W}_{\delta p_j}\subseteq Y_{\delta p_j}^*$. 
\end{proposition}

\begin{proof}
The proof  proceeds by a nested induction over $j$ starting with $j=k$.

\smallskip
\noindent\textbf{Base case:}
Recall that for $j=k$ the last line of \eqref{equ:psi_delta} disappears. Hence, for any state $v\in\mathcal{W}_{\delta p_j}$ either the first or the second line of \eqref{equ:psi_delta}  holds. 
Then the proof reduces to \REFthm{thm:Reachability} and \REFthm{thm:SingleRabin} in the following way. 

First, we fix all fixpoint variables $Y^*_{p_0\hdots p_l}$ and $X^{i_l}_{p_0\hdots p_l}$ for $l<j$ as well as $Y^*_{\delta p_j}$. With this, we see that 
$T:=S_\delta\cup (Q_{\delta p_j} \cap G_{p_j}\cap\cpre(Y^*_{\delta p_j}))$ becomes a fixed set of states and \eqref{equ:Xdeltai} reduces to
\begin{equation*}
 Y^*_{\delta p_j}=Z^*(\tuple{T,Q_{\delta p_j}})
\end{equation*}
where $Z^*(\tuple{T,Q})$ is the set of states computed by the fixpoint algorithm in \eqref{equ:Reach:FP}. 

Then it follows from \REFthm{thm:Reachability} that any state $v\in V$ for which there exists a fair adversarial play $\pi$ that is winning for the winning condition $Q_{\delta p_j}\mathcal{U} T$ is contained in $Y^*_{\delta p_j}$. If, indeed the first line of \eqref{equ:psi_delta} holds for $\pi$, this ensures that the claim holds.

Now assume that $Q_{\delta p_j}\mathcal{U} T$ holds for $\pi$ but $S_\delta$ is never reached. Hence, $Q_{\delta p_j}\mathcal{U} (Q_{\delta p_j} \cap G_{p_j}\cap\cpre(Y^*_{\delta p_j}))$ holds for $\pi$. With this, it follows form \REFthm{thm:SingleRabin} that any state $v\in V$ for which there exists a fair adversarial play $\pi$ for which the second line of \eqref{equ:psi_delta} holds is contained in $Y^*_{\delta p_j}$, proving the claim in this case.

\smallskip
\noindent\textbf{Induction Step:}
For the induction from \enquote{$j+1$} to \enquote{$j$} we first analyze the assumption. I.e., we know that for the longer computation prefix $\delta'=\delta p_j$ and any next permutation index $p_{j+1}$ we have that $\mathcal{W}_{\delta' p_{j+1}}\subseteq Y^*_{\delta'p_{j+1}}$.
Further, observe that $\Psi'_{\delta p_j}\subseteq \bigcup_{p_{j+1}\in P\setminus\set{p_1,\hdots,p_{j}}}\mathcal{W}_{\delta' p_{j+1}}\setminus S_{\delta p_j i_j}$ by construction. 
We therefore have 
\begin{equation*}
 \Psi'_{\delta p_j}\subseteq \bigcup_{p_{j+1}\in P\setminus\set{p_1,\hdots,p_{j}}}Y^*_{\delta'p_{j+1}}\setminus S_{\delta p_j i_j}=\mathcal{A}_{\delta p_j i_j}.
\end{equation*}
With this observation, we see that any fair adversarial play $\pi$ which fulfills the last line of \eqref{equ:psi_delta} also fulfills the weaker condition $Q_{\delta p_j}\mathcal{U} \mathcal{A}_{\delta p_j i_j}$. Therefore, the claim follows from the same reasoning as in the base case by re-defining $T$ to 
$T:=S_\delta\cup (Q_{\delta p_j} \cap G_{p_j}\cap\cpre(Y^*_{\delta p_j})) \cup\mathcal{A}_{\delta p_j i_j}$. 
\end{proof}

\subsubsection{Additional Lemmas and Proofs}\label{appD:additional}
In this section we provide additional lemmas and proofs to support the proof of 
\REFthm{thm:FPsoundcomplete} and \REFthm{thm:strategy}.

\subsubsection{Proof of \REFprop{prop:case_split_sound}}\label{appD:case_split_sound}

\begin{lemma}\label{lem:ranking_app}
Given the premisses of \REFprop{prop:case_split_sound}, it holds for all $v\in X_{\delta p_j}^{i_j}$ that
 \begin{compactenum}[(i)]
 \item $v\in S_\delta$ iff $\rank{v}\leq \delta p_j 0\underline{\gamma}$
 \item $v\in X_{\delta p_j}^{i_j}$ iff $\rank{v}\leq \delta p_j i_j \overline{\gamma}$
\item $v\in Y_{\delta p_j}^*$ iff $\rank{v}\leq \delta p_j n\overline{\gamma}$ 
 \item $v\in \mathcal{A}_{\delta p_j i_j}$ iff there exists $\underline{\gamma}<\gamma'\leq\overline{\gamma}$ s.t.\ $\rank{v}= \delta p_j i_j\gamma'$
\end{compactenum}
\end{lemma}

\begin{proof}[Proof of \REFlem{lem:ranking_app}] We prove all claims separately.\\
\begin{inparaenum}[(i)]
\item It immediately follows from \REFdef{def:ranking} (i) that $\delta p_j 0\underline{\gamma}\in R(v)$ iff $v\in S_\delta$. If it is the minimal element in $R(v)$ then $\rank{v}=\delta p_j 0\underline{\gamma}$, if not, there exists a smaller element in $R(v)$, and then $\rank{v}<\delta p_j 0\underline{\gamma}$ from the definition of $\rank{}$.\\
\item First, observe, that for $j=k$ it follows from \eqref{equ:Xdeltai} that $X_{\delta p_k}^{i_k}=S_\delta p_k i_k$ and therefore from (i) that $v\in X_{\delta p_k}^{i_k}$ iff $\rank{v}\leq \delta p_k i_k$. Now we do an induction, assuming that for any $p_{j+1}\in P\setminus\set{p_0,\hdots,p_j}$ and $0<i_{j+1}\leq n$ it holds that $v\in X_{\delta p_{j+1}}^{i_{j+1}}$ iff $\rank{v}\leq \delta' p_{j+1} i_{j+1} \overline{\gamma'}$ (where $\delta'$ goes up to index $j$ and $\gamma'$ starts only at index $j+2$. Now recall that 
\begin{equation*}
   X_{\delta p_j}^{i_j}= \bigcup_{p_{j+1}\in P\setminus\set{p_0,\hdots,p_j}} Y^{*}_{\delta p_{j+1}}
  =\bigcup_{p_{j+1}\in P\setminus\set{p_0,\hdots,p_j}} \bigcup_{i_{j+1}>0} X^{i_{j+1}}_{\delta p_j i_j p_{j+1}}.
\end{equation*}
Hence, $v\in X_{\delta p_j}^{i_j}$ iff there exists $p_{j+1}\in P\setminus\set{p_0,\hdots,p_j}$ and $0<i_{j+1}\leq n$ s.t.\ $v\in X^{i_{j+1}}_{\delta p_j i_j p_{j+1}}$. Now we know that for any choice of $p_{j+1}$ and $i_{j+1}$ we have $\rank{v}\leq \delta' p_j i_j p_{j+1} i_{j+1} \overline{\gamma'}$. Now the worst case, in terms of the lexicographic ordering over $D$ is that $p_{j+1}=\max(P\setminus\set{p_0,\hdots,p_j})$ and $i_{j+1}=n$. Hence, we know that $\rank{v}\leq \delta p_j i_j  \overline{\gamma}$. \\
\item As $Y^{*}_{\delta p_{j}}=\bigcup_{i_{j}>0} X^{i_{j}}_{\delta p_j}$ it follows that there exists  $0<i_{j}\leq n$ s.t.\ $v\in X^{i_{j}}_{\delta p_j}$ and (from (ii)) therefore $\rank{v}\leq \delta p_j i_j \overline{\gamma}$. Again, the worst case is $i_j=n$, giving $\rank{v}\leq \delta p_j n\overline{\gamma}$.\\
\item It follows from \eqref{equ:Xdeltai} that $v\in \mathcal{A}_{\delta p_j i_j}$ iff $v\in X_{\delta p_j}^{i_j}\setminus S_{\delta p_j i_j}$. Hence, it follows from (i) and (ii) that $\rank{v}> \delta p_j 0\underline{\gamma}$ and $\rank{v}\leq \delta p_j i_j \overline{\gamma}$ which is true iff there exists $\underline{\gamma}<\gamma'\leq\overline{\gamma}$ s.t.\ $\rank{v}= \delta p_j i_j\gamma'$, which proves the statement.
\end{inparaenum}
\end{proof}

Given these properties of the ranking function, we are ready to prove the suggested case split in \REFprop{prop:case_split_sound}.

\begin{proof}[Proof of \REFprop{prop:case_split_sound}]
We call a vertex $v\in V$ that fulfills cases ($\alpha$) in either \REFlem{lem:ranking_app} or \REFprop{prop:case_split_sound} an ($\alpha$)-vertex.
First, observe that cases (i) and (iv) in \REFlem{lem:ranking_app} coincide with cases (a) and (d), respectively, in \REFprop{prop:case_split_sound}. Further, recall that $X^1_{\delta p_j}=\emptyset$. Therefore, $X^1_{\delta p_j}$ only contains (a)-,(b)- and (d)-vertices, as $\apre(\cdot,\emptyset)=\emptyset$. Now we know from (ii) that for any $v\in X^1_{\delta p_j}$ we have $\rank{v}\leq \delta p_j 1 \overline{\gamma}$. Now excluding the rankings for (a)- and (d)-vertices we obtain that (b)-vertices must have rank $\rank{v}\leq \delta p_j 1 \underline{\gamma}$. Similarly, for every $i_j>1$ we know that $X^{i_j}_{\delta p_j}$ contains (a)-, (b)-, (c)- and (d)- vertices. Now excluding (a)-, (b)- and (d)- vertices yields $\rank{v}\leq \delta p_j i_j \underline{\gamma}$ for all (c)-vertices.
\end{proof}

\subsubsection{Proof of \eqref{equ:Psi:outline}}\label{appD:proofPsi}

Given the notation in \REFsec{sec:theory:sound} we prove that the equality in \eqref{equ:Psi:outline} holds.

First recall that  
\begin{align}
& \Psi'_{\delta' p_{j+1}}:= \left(
\begin{array}{rl}
&\Box Q_{\delta' p_{j+1}} \wedge \Box\Diamond G_{p_{j+1}}\\
\vee &\Box  Q_{\delta' p_{j+1}}  \wedge \left(\bigvee_{i\in \Pt{j+1}}\left(\Diamond\Box\overline{R}_{i}\wedge \Box\Diamond G_{i}\right)\right) 
\end{array}
\right),
\label{equ:psiprime_jm1}
\end{align}
where $\Pt{j+1}:=P\setminus\set{p_1,\hdots,p_{j+1}}$.

For the insertion of \eqref{equ:psiprime_jm1} into \eqref{equ:Psi} we have the following observations. First, observe that $\Diamond(B\vee C)=\Diamond B\vee \Diamond C$, i.e., we can distribute the eventuality operator preceding $\Psi'_{\delta' p_{j+1}}$ over both lines. Second, we can re-order the preceding disjunction over $p_{j+1}$ in \eqref{equ:Psi} and the disjunction between the two lines of \eqref{equ:psiprime_jm1}. 
This yields to the following condition
\begin{align}
 \Psi_{\delta p_j}=&\textstyle\Box Q_{\delta p_j}\wedge
\left(
\bigvee_{p_{j+1}\in \Pt{j}} (\Diamond \lambda_1)
\vee\bigvee_{p_{j+1}\in \Pt{j}} (\Diamond \lambda_2)
\right)\nonumber \\
=& \underbrace{\textstyle\left( \Box Q_{\delta p_j} \wedge \bigvee_{p_{j+1}\in \Pt{j}} (\Diamond \lambda_1) \right)}_{=:\Psi_1}
\vee \underbrace{\textstyle\left( \Box Q_{\delta p_j} \wedge \bigvee_{p_{j+1}\in \Pt{j}} (\Diamond \lambda_2)\right)}_{=:\Psi_2}, \label{equ:derive:a}
\end{align}
where $\lambda_i$ denotes the $i$-th line of the conjunction in \eqref{equ:psiprime_jm1}.

Now let us investigate the terms $\Psi_1$ and $\Psi_2$ in \eqref{equ:derive:a} separately. For $\Psi_1$, observe that  $\Diamond\Box\Diamond A=\Box\Diamond A$ and $\Diamond(\Box A\wedge \Box B)=\Diamond\Box A\wedge \Diamond\Box B$. Further we have $Q_{\delta' p_{j+1}}=Q_{\delta p_j}\wedge \overline{R}_{j+1}\subseteq Q_{\delta p_j}$ and hence
\begin{align*}
 \Psi_1=&\Box Q_{\delta p_j}\wedge\bigvee_{p_{j+1}\in \Pt{j}} 
\left(\Diamond\Box(Q_{\delta p_j}\wedge\overline{R}_{p_{j+1}})\wedge \Box\Diamond G_{p_{j+1}}
 \right)
\end{align*}
By using the equality $\Diamond\Box(A\wedge B)=\Diamond\Box A\wedge \Diamond\Box B$ and the fact that $Q_{\delta p_j}$ is independent of the choice of $p_{j+1}$ we get
\begin{align}
 \Psi_1=&\Box Q_{\delta p_j}\wedge\Diamond\Box Q_{\delta p_j}\wedge\bigvee_{p_{j+1}\in \Pt{j}} 
\left(\Diamond\Box\overline{R}_{p_{j+1}}\wedge \Box\Diamond G_{p_{j+1}}
 \right)\notag\\
 =&\Box Q_{\delta p_j}\wedge\bigvee_{p_{j+1}\in \Pt{j}} 
\left(\Diamond\Box\overline{R}_{p_{j+1}}\wedge \Box\Diamond G_{p_{j+1}}
 \right).\label{equ:derive_line1}
\end{align}

To analyze $\Psi_2$ in \eqref{equ:derive:a}, recall that the eventuality operator $\Diamond$ distributes over disjunctions. 
We can therefore move the inner disjunction over $i$ outside and get
\begin{align*}
 \Psi_2=&\Box Q_{\delta p_j}\wedge
 \bigvee_{p_{j+1}\in \Pt{j}} \left(\bigvee_{i\in \Pt{j+1}}\left[\Diamond\left(\Box Q_{\delta' p_{j+1}} \wedge \left(\Diamond\Box\overline{R}_{i}\wedge \Box\Diamond G_{i}\right)\right) \right]\right)
\end{align*}
Now observe that $\left(\Diamond\Box\overline{R}_{i}\wedge \Box\Diamond G_{i}\right)=\Diamond\left(\Box\overline{R}_{i}\wedge \Box\Diamond G_{i}\right)$ and 
$\Diamond(\Box A\wedge \Diamond B)=\Diamond\Box A\wedge \Diamond B$. Additionally using $Q_{\delta' p_{j+1}}=Q_{\delta p_j}\wedge \overline{R}_{p_{j+1}}\subseteq Q_{\delta p_j}$  we get 

\begin{align*}
 \Psi_2=&\Box Q_{\delta p_j}\wedge
 \bigvee_{p_{j+1}\in \Pt{j}} \left(\bigvee_{i\in \Pt{j+1}}\left[\Diamond\Box (Q_{\delta p_j}\wedge \overline{R}_{p_{j+1}})\wedge \left(\Diamond\Box\overline{R}_{i}\wedge \Box\Diamond G_{i}\right)\right]\right)
\end{align*}

Now we can do the same trick as in the simplification of $\Psi$ (see \eqref{equ:derive_line1}) to remove the $Q_{\delta p_j}$ term inside the disjunction and get
\begin{align}\label{equ:derive:b}
 \Psi_2=&\Box Q_{\delta p_j}\wedge
 \bigvee_{p_{j+1}\in \Pt{j}} \left(\bigvee_{i\in \Pt{j+1}}\left[\Diamond\Box \overline{R}_{p_{j+1}}\wedge \left(\Diamond\Box\overline{R}_{i}\wedge \Box\Diamond G_{i}\right)\right]\right)
\end{align}

To see how we can simplify \eqref{equ:derive:b}, let us assume that the set $\Pt{j}$ contains three elements, e.g., $\set{a,b,c}$. Then we can expand \eqref{equ:derive:b} to
\begin{align*}
 &\Diamond\Box \overline{R}_{a}\wedge \left(\Diamond\Box\overline{R}_{b}\wedge \Box\Diamond G_{b}\right)\\
  \vee~&\Diamond\Box \overline{R}_{a}\wedge \left(\Diamond\Box\overline{R}_{c}\wedge \Box\Diamond G_{c}\right)\\
  \vee~&\Diamond\Box \overline{R}_{b}\wedge \left(\Diamond\Box\overline{R}_{a}\wedge \Box\Diamond G_{a}\right)\\
  \vee~&\Diamond\Box \overline{R}_{b}\wedge \left(\Diamond\Box\overline{R}_{c}\wedge \Box\Diamond G_{c}\right)\\ 
 \vee~&\Diamond\Box \overline{R}_{c}\wedge \left(\Diamond\Box\overline{R}_{b}\wedge \Box\Diamond G_{b}\right)\\
  \vee~&\Diamond\Box \overline{R}_{c}\wedge \left(\Diamond\Box\overline{R}_{a}\wedge \Box\Diamond G_{a}\right)
\end{align*}

Now, we can re-order terms and get
\begin{align*}
 &\left(\Diamond\Box\overline{R}_{b}\wedge \Box\Diamond G_{b}\right) \wedge \left(\Diamond\Box\overline{R}_{a}\vee \Diamond\Box\overline{R}_{c}\right)\\
  \vee~&\left(\Diamond\Box\overline{R}_{c}\wedge \Box\Diamond G_{c}\right) \wedge \left(\Diamond\Box\overline{R}_{a}\vee \Diamond\Box\overline{R}_{b}\right)\\
  \vee~&\left(\Diamond\Box\overline{R}_{a}\wedge \Box\Diamond G_{a}\right) \wedge \left(\Diamond\Box\overline{R}_{b}\vee \Diamond\Box\overline{R}_{c}\right)
\end{align*}

Generalizing this observation, we get the following formula equivalent to \eqref{equ:derive:b}
\begin{align}\label{equ:derive_line2}
 \Psi_2=\Box Q_{\delta p_j}\wedge\bigvee_{p_{j+1}\in \Pt{j}}\left(\left(\Diamond\Box\overline{R}_{p_{j+1}}\wedge \Box\Diamond G_{p_{j+1}})\right)
\wedge \bigvee_{j\in \Pt{j+1}}\Diamond\Box \overline{R}_{j}\right)
\end{align}

Now recall that $A\wedge B \Rightarrow A$ for any choice of $A$ and $B$. With this one can verify that $\Psi_2\Rightarrow\Psi_1$ as the term after the disjuction over $p_{j+1}$ in \eqref{equ:derive_line2} implies the term after the disjuction over $p_{j+1}$ in \eqref{equ:derive_line1}. 
Hence, 
the set of states which fulfill $\Psi_1$ in \eqref{equ:derive_line1} is always larger then the set of states which fulfill $\Psi_2$ \eqref{equ:derive_line2}). As both terms are connected by a conjunction in \eqref{equ:derive:a}, we can ignore $\Psi_2$ in \eqref{equ:derive:a} and obtain
\begin{align}\label{equ:propertyequ}
 \Psi_{\delta p_j}=\Psi_1=\Box Q_{\delta p_j}\wedge\bigvee_{p_{j+1}\in \Pt{j}} 
\left(\Diamond\Box\overline{R}_{p_{j+1}}\wedge \Box\Diamond G_{p_{j+1}}
 \right).
\end{align}
This concludes the proof of \eqref{equ:Psi:outline} as \eqref{equ:propertyequ} coincides with \eqref{equ:Psi:outline:b}.


\subsection{Additional Proofs for \REFsec{sec:SimpleRabinGames}}

\subsubsection{Fair Adversarial Rabin Chain Games}\label{app:RabinChain}

\begin{theorem*}[\REFthm{thm:RabinC_all} restated for convenience]
Let $\Gl =\tup{\game, \El}$ be a game graph with live edges and 
$\FR$ be a Rabin condition over $\game$ with $k$ pairs for which the chain condition \eqref{equ:RCprop_1} holds. 
Further, let 
   \begin{subequations}\label{p:equ:RabinC_all}
   \begin{align}\label{p:equ:Rabin_all_FP}
 \textstyle Z^*\coloneqq\nu Y_{0}.~\mu X_{0}.~\nu Y_{k}.~\mu X_{k}.~\nu Y_{k-1}.~\hdots\mu X_{1}.~\bigcup_{j=0}^k \Ct_{j},
\end{align}
 \begin{align*}
\text{where}~~\Ct_{j}\coloneqq
\overline{R}_j\cap
 \left[ 
 \left( G_{j}\cap \cpre(Y_{j})\right)
 \cup \apre(Y_{j},X_{j})
 \right]
\end{align*}
with $G_{p_0}\coloneqq\emptyset$ and $R_{p_0}\coloneqq\emptyset$. 

Then $Z^*$ is equivalent to the winning region $\WlR$ of $\p{0}$ in the fair adversarial game over~$\Gl$ for the winning condition $\vR$ in \eqref{equ:vR}. Moreover, the fixpoint algorithm runs in $O(n^{k+2})$ symbolic steps, and a memoryless winning strategy for $\p{0}$ can be extracted from it.
 \end{subequations}
 \end{theorem*}

In this section we prove \REFthm{thm:RabinC_all}. That is, we prove that for Rabin chain conditions, the fixpoint computing $Z^*$ in \eqref{equ:Rabin_all} simplifies to the one in \eqref{p:equ:RabinC_all}. This is formalized in the next proposition.
\begin{proposition}\label{prop:RCfixpoint}
Given the premisses of \REFthm{thm:RabinC_all} let $Z^*$ be the fixed-point of the $\mu$-calculus expression of \eqref{equ:Rabin_all} and $\Zt^*$ the fixed-point of \eqref{p:equ:RabinC_all}. Then $Z^*=\Zt^*$.
 \end{proposition}
If \REFprop{prop:RCfixpoint} holds, we immediately see that \REFthm{thm:RabinC_all} directly follows from \REFthm{thm:FPsoundcomplete}. It therefore remains to prove \REFprop{prop:RCfixpoint}.

Similar to the soundness and completeness proof for \REFthm{thm:FPsoundcomplete} we prove \REFprop{prop:RCfixpoint} by an induction over the nesting of fixpoints in \eqref{equ:Rabin_all} form inside to outside. Here, however we do not need to explicitly refer to counters $i_j$ as in \REFprop{thm:RabinC_all}. Hence, we can look at permutation prefixes instead of configuration prefixes. We have the following proposition.

 \begin{proposition}\label{prop:RC_inductive}
Let $P$ be the index set of the Rabin chain condition $\FR$ in \REFthm{thm:RabinC_all}. Further, for any $j\in[0;k]$ let $\delta:=p_0 p_1\hdots p_{j-1}$ be a \emph{permutation prefix},  $\Pt{\delta}:=P\setminus \set{p_0,\hdots, p_{j-1}}$ the reduced index set and $q_0:=p_j\in \Pt{\delta}$ the current permutation index. Further, define\footnote{Observe that $\delta p_j=p_0\hdots p_{j-1} p_j$ is itself a permutation prefix.}
\begin{subequations}\label{equ:Zstarj:a}
 \begin{align}
 \textstyle Z^*_{\delta p_j}\coloneqq&\nu Y_{q_0}.~\mu X_{q_0}.\notag\\
 &\textstyle \qquad\bigcup_{q_1\in \Pt{\delta p_j}} \nu Y_{q_1}.~\mu X_{q_1}.~\notag\\
&\qquad\vdots\notag\\
 &\qquad\qquad\textstyle \bigcup_{q_n\in \Pt{\delta p_j}\setminus\set{q_1,\hdots, q_{n-1}}}\nu Y_{q_n}.~\mu X_{q_n}.~
S_\delta\cup\left[\bigcup_{\ell=0}^n \mathcal{C}_{\delta q_\ell}\right]\label{equ:Zstarj:a:FP}
\end{align}
where $n\coloneqq k-j$,
\begin{align}\label{equ:Zstarj:a:C}
  \mathcal{C}_{\delta q_j}\coloneqq
Q_\delta\cap\bigcap_{i=0}^{\ell} \overline{R}_{q_i}
\cap
 \left[ 
 \left( G_{q_\ell}\cap \cpre(Y_{q_\ell})\right)
 \cup \left(\apre(Y_{q_\ell},X_{q_\ell})\right)
 \right],
\end{align}
\end{subequations}
$Q_\delta\coloneqq\bigcap_{i=0}^{j} \overline{R}_{p_i}$ and $S_{p_0\hdots p_{j-1}}\coloneqq\bigcup_{b=0}^{j-1}\mathcal{C}_{p_0\hdots p_b}$.

Then it holds that 
\begin{subequations}\label{equ:Zstarj:b}
\begin{align}
 Z^*_{\delta p_j}=&\nu Y_{r_0}.~\mu X_{r_0}.~\nu Y_{r_1}.~\mu X_{r_1}.~\hdots\nu Y_{r_n}.~\mu X_{r_n}.
 ~\textstyle S_\delta\cup\left[\bigcup_{\ell=0}^n \Ct_{\delta r_\ell}\right],
 \label{equ:Zstarj:b:FP}
\end{align}
where
\begin{align}\label{equ:Zstarj:b:C}
\Ct_{\delta r_\ell}:=
 Q_{\delta p_j}\cap\overline{R}_{r_\ell}\cap\left[ 
 \left( G_{r_\ell}\cap \cpre(Y_{r_\ell})\right)
 \cup \left(\apre(Y_{r_\ell},X_{r_\ell})\right)
 \right]
\end{align}
\end{subequations}
with $r_i\in \Pt{\delta p_j}$ for all $i\in[1;n]$ such that $r_1>r_2>\hdots >r_n$ and $r_0=q_0=p_j$.
\end{proposition}
 
It should be noted that \REFprop{prop:RC_inductive} needs to hold for any choice of $j$ and $\delta$. 
Further, we have slightly abused notation by not specifying the values of the fixpoint parameters used within $S_\delta$. 
This is, however, not relevant for the proof of \REFprop{prop:RC_inductive} and we should interpret $S_\delta$ as a term computed by an arbitrary choice of the involved fixpoint parameters. 

Now, it should be obvious that for the choice $j=0$ we get $\delta=\varepsilon$ and $S_\delta=\emptyset$. Further, we see that in this case, we have $\Pt{\delta p_0}=P$ which implies that $Z^*_{p_0}$ in \eqref{equ:Zstarj:a} coincides with $Z^*$ in \eqref{equ:Rabin_all}. Further, as $\Pt{\delta p_0}=P$ we must have $r_1=k$, $r_2=k-1$, $\hdots$, $r_k=1$ and $r_0=p_0=0$ to fulfill the requirements on $r$. Further $Q_{p_0}=\overline{R_0}=Q$. Therefore $Z^*_{p_0}$ in \eqref{equ:Zstarj:b} coincides with $Z^*$ in \eqref{p:equ:RabinC_all} in this case. Hence, proving \REFprop{prop:RC_inductive} for any $j$ (including $j=0$), immediately proves \REFprop{prop:RCfixpoint}.  
 
In the remainder of this section we prove \REFprop{prop:RC_inductive} by an induction over $j$, starting with $j=k$ as the base case. Now observe that for $j=k$ we have $\Pt{\delta p_j}=\emptyset$ and hence both \eqref{equ:Zstarj:a} and \eqref{equ:Zstarj:b} reduce to a two-nested fixpoint over the variables $Y_{q_0},~X_{q_0}$ and $Y_{r_0},~X_{r_0}$, respectively, where $r_0=q_0=p_k$ by definition. Further, we see that $\mathcal{C}_{\delta q_0}=\Ct_{\delta r_0}$ by definition, which immediately proves the claim of \REFprop{prop:RC_inductive} for the base case.

In the remainder of this section we prove the induction step from \enquote{$j$} to \enquote{$j-1$} in a series of definitions and lemmas. 
 
 \begin{definition}
 Let $\tilde{P}\subseteq \mathbb{N}$ be a set of $n$ indices and  $\beta=q_1\hdots q_n$ with $q_i\in \tilde{P}$ and $q_i\neq q_j$ for all $j\neq i$ a full permutation sequence of the elements from $\tilde{P}$. For $1\leq j\leq l\leq n$ we call $\beta_{jl}=q_jq_{j+1}\hdots q_j$ a \emph{maximal decreasing sub-sequence} of $\beta$ if (i) $q_j<q_{j+1}<\hdots<q_l$, (ii) $q_{j-1}>q_j$ or $j=1$, and (iii) $q_l>q_{l+1}$ or $l=n$. 
\end{definition}

We see that, by definition, the first maximally decreasing sub-sequences of a permutation sequence $\beta$ starts with $q_1$. Intuitively, decreasing sub-sequences allow to immediately utilize the properties in \eqref{equ:RCprop_1} to simplify the fixpoint expression.

\begin{lemma}\label{lem:RC_a}
 Let $\delta$, $\Pt{\delta}$ and $q_0=p_j$ as in \REFprop{prop:RC_inductive}, $\beta=q_1\hdots q_n$ a full permutation sequence of $\Pt{\delta p_j}$ and $\beta_{jl}=q_jq_{j+1}\hdots q_j$ a \emph{maximal decreasing sub-sequence} of $\beta$.
 Then 
  \begin{align}
    &\textstyle \nu Y_{q_j}.~\mu X_{q_j}.~\hdots\nu Y_{q_l}.~\mu X_{q_l}.~ \bigcup_{i=j}^l \mathcal{C}_{\delta q_i}
    =  \nu Y_{q_j}.~\mu X_{q_j}.~\mathcal{C}_{\delta q_j}
   \end{align}
\end{lemma}

\begin{proof}
Let $\alpha:=q_0\hdots q_{j-1}$ and observe that 
\begin{align*}
 \mathcal{C}_{\delta q_j}&=Q_{\delta\alpha} \cap \left[\left( \overline{R}_{j}\cap G_{q_j}\cap \cpre(Y_{q_j})\right)\cup \left(\overline{R}_{j}\cap \apre(Y_{q_j},X_{q_j})\right)\right]\\
  \mathcal{C}_{\delta q_{j+1}}&=Q_{\delta\alpha} \cap \left[\left(\overline{R}_{j}\cap \overline{R}_{j+1}\cap G_{q_{j+1}}\cap \cpre(Y_{q_j})\right)
  \cup \left(\overline{R}_{j}\cap \overline{R}_{j+1}\cap \apre(Y_{q_j},X_{q_j})\right)\right]\\
  &=Q_{\delta\alpha} \cap \left[\left(\overline{R}_{j}\cap G_{q_{j+1}}\cap \cpre(Y_{q_j})\right)
  \cup \left(\overline{R}_{j}\cap \apre(Y_{q_j},X_{q_j})\right)\right],
\end{align*}
where the simplification of $\mathcal{C}_{\delta q_{j+1}}$ follows from $\overline{R}_{j}\subseteq \overline{R}_{j+1}$ (see \eqref{equ:RCprop_1}).
So $\mathcal{C}_{\delta q_j}$ and $\mathcal{C}_{\delta q_{j+1}}$ really only differ by the $G_{q_j}$ (resp. $G_{q_{j+1}}$) term in the first term of the disjunct. As $G_{q_j}\supseteq G_{q_{j+1}}$ (see \eqref{equ:RCprop_1}) and all terms in the first part of the disjunct are intersected, we see that $\mathcal{C}_{\delta q_j}\supseteq\mathcal{C}_{\delta q_{j+1}}$. With this it follows from case (iii) in \REFlem{lem:contain} that 
\begin{align*}
 &\textstyle \nu Y_{q_j}.~\mu X_{q_j}.\nu Y_{q_{j+1}}.~\mu X_{q_{j+1}}.~ \mathcal{C}_{\delta q_{j}}\cup\mathcal{C}_{\delta q_{j+1}}
 =  \nu Y_{q_j}.~\mu X_{q_j}.~\mathcal{C}_{\delta q_j}.
\end{align*}
Applying this argument to all $i\in[j;l]$ proves the claim.
\end{proof}

\begin{definition}
 We say that a permutation sequence $\beta$ has \emph{chain index} $m$ if it contains $m$ \emph{maximal decreasing} sub-sequences. For $\beta=q_1\hdots q_n$ with chain index $m$ we define its reduction~$\beta_\downarrow$ as $\beta_\downarrow:=r_1...r_m$ such that $r_m=q_j$ if $\beta_{jl}$ is the $m$'th maximally decreasing sub-sequence of $\beta$. 
\end{definition}
  
  \begin{lemma}\label{lem:RC_b}
    Let $\delta$, $\Pt{\delta}$ and $q_0=p_j$ as in \REFprop{prop:RC_inductive}, $\beta=q_1\hdots q_n$ a full permutation sequence of $\Pt{\delta p_j}$ with chain index $m$ and $\beta_\downarrow:=r_1...r_m$. Then 
   \begin{align}
   &\nu Y_{q_0}.~\mu X_{q_0}.~\nu Y_{q_1}.~\mu X_{q_1}.~\hdots\nu Y_{q_n}.~\mu X_{q_n}\bigcup_{j=0}^n \mathcal{C}_{\delta q_j}\notag\\
    &\qquad=\nu Y_{r_0}.~\mu X_{r_0}.~\nu Y_{r_1}.~\mu X_{r_1}.~\hdots\nu Y_{r_m}.~\mu X_{r_m} \bigcup_{l=0}^m \mathcal{C}_{\delta q_l}
   \end{align}
   where $q_0=r_0=p_j$.
  \end{lemma}

  \begin{proof}
First, observe that by construction we always have $r_1=q_1$. Hence, $Q_{\delta\alpha}$ in the proof of \REFlem{lem:RC_a} reduces to $Q_{\delta q_1}$ in this case. Further, consider $r_2=q_j$  and observe that in this case $Q_{\delta\alpha}=Q_\delta\cap\bigcap_{i=0}^{j-1}\overline{R}_{q_i}=Q_{\delta q_0}\cap\overline{R}_{q_{1}}=Q_{\delta p_j}\cap\overline{R}_{r_{1}}$ as $q_1\hdots q_{j-1}$ is a maximal decreasing sub-sequence by construction. Iteratively re-applying this argument along with \REFlem{lem:RC_a} for every $l\in[1,m]$ therefore proves the claim. 
\end{proof}

Now observe that we can re-apply \REFlem{lem:RC_b} to $\beta_\downarrow$ and reduce it even more. That means, $\beta_\downarrow$ could now again have maximal decreasing sub-sequences and we therefore can reduce it to $(\beta_\downarrow)_\downarrow$. This might again be reducible and so forth. We therefore define the \emph{maximal reduced permutation sequence} $\beta_{\Downarrow}=(((\beta_\downarrow)_\downarrow)\hdots)_\downarrow=r_1\hdots r_n$ such that $r_1>r_2>\hdots r_n$, i.e. the chain index of $\beta_{\Downarrow}$ is equivalent to its length. With this, we have the following result.

\begin{lemma}\label{lem:RC_c}
Let $\delta$, $\Pt{\delta}$ and $q_0=p_j$ as in \REFprop{prop:RC_inductive}, $\beta=q_1\hdots q_n$ a full permutation sequence of $\Pt{\delta p_j}$ and $\beta_\Downarrow:=r_1...r_m$ its maximal reduced permutation sequence. Then 
   \begin{align}
   &\nu Y_{q_0}.~\mu X_{q_0}.~\nu Y_{q_1}.~\mu X_{q_1}.~\hdots\nu Y_{q_n}.~\mu X_{q_n}\bigcup_{j=0}^n \mathcal{C}_{\delta q_j}\notag\\
    &\qquad=\nu Y_{r_0}.~\mu X_{r_0}.~\nu Y_{r_1}.~\mu X_{r_1}.~\hdots\nu Y_{r_m}.~\mu X_{r_m}\bigcup_{l=0}^m \Ct_{\delta q_l}
   \end{align}
\end{lemma}
\begin{proof}
 It follows from the definition of $\beta_\Downarrow$ and repeatably applying \REFlem{lem:RC_b} that 
 \begin{align*}
  &\nu Y_{q_0}.~\mu X_{q_0}.~\nu Y_{q_1}.~\mu X_{q_1}.~\hdots\nu Y_{q_n}.~\mu X_{q_n}\bigcup_{j=0}^n \mathcal{C}_{\delta q_j}\\
    &\qquad=\nu Y_{r_0}.~\mu X_{r_0}.~\nu Y_{r_1}.~\mu X_{r_1}.~\hdots\nu Y_{r_m}.~\mu X_{r_m}\bigcup_{l=0}^m\mathcal{C}_{\delta r_l}
 \end{align*}
 Now we have by definition that $r_0=q_0$ and $r_1=q_1$ and therefore $\mathcal{C}_{\delta r_0}=\Ct_{\delta r_0}$ and $\mathcal{C}_{\delta r_1}=\Ct_{\delta r_1}$ by definition. Now recall that $r_1>r_2$, hence $\overline{R}_{r_1}\cap\overline{R}_{r_2}=\overline{R}_{r_2}$. Iteratively applying this argument gives $\mathcal{C}_{\delta r_l}=\Ct_{\delta r_l}$ for all $l\in[1,n]$, what proves the claim.
\end{proof}

Note that the only full permutation sequence of $\Pt{\delta p_j}$ with \emph{chain index} $n$ is the one where $q_1>q_2>\hdots>q_n$, giving $\beta_\downarrow=\beta_\Downarrow=\beta$. Hence, the sequence $r_1\hdots r_n$ used in \eqref{equ:Zstarj:b} is actually the \emph{maximal permutation sequence} of $\Pt{\delta p_j}$. We see that all other full permutation sequences~$\gamma$ of $\Pt{\delta p_j}$ have \emph{chain index} $m$ such that $1\leq m<n$. As the $\Ct$ terms in \eqref{equ:RabinC_all_Cpj} do not depend on the history of permutation sequences from $\Pt{\delta p_j}$, we see that any term constructed for a non-maximal permutation sequence is contained in the term constructed for the maximal permutation sequence. This is formalized in the next lemma.

\begin{lemma}\label{lem:RC_d}
Let $\delta$, $\Pt{\delta}$ and $q_0=p_j$ as in \REFprop{prop:RC_inductive} and let $\beta=r_1...r_n$ be the maximal permutation sequence of $\Pt{\delta p_j}$, that its $\beta=\beta_\Downarrow$. Further, let $\gamma\neq\beta$ be a full permutation sequence of $\Pt{\delta p_j}$ such that $\gamma_\Downarrow=s_1\hdots s_m$ with $m<n$. Then 
\begin{align}
 &\nu Y_{r_1}.~\mu X_{r_1}.~\hdots\nu Y_{r_n}.~\mu X_{r_n}\bigcup_{l=1}^n\Ct_{\delta r_l}\\
 &\quad\quad\subseteq\nu Y_{s_1}.~\mu X_{s_1}.~\hdots\nu Y_{s_m}.~\mu X_{s_m}\bigcup_{l=1}^m\Ct_{\delta s_l}
 \label{equ:proofRC:contain}
\end{align}
\end{lemma}

\begin{proof}
As $\beta$ is a full permutation sequence of $\Pt{\delta p_j}$ we know that for any $i\in[1;m]$ there exists one $j\in[1;n]$ such that $s_i=r_j$. Further, as $\Ct$ does not depend on the history of the permutation sequence $\beta$ and $\gamma$ we see that  $\Ct_{\delta s_i}=\Ct_{\delta r_j}$ in this case. As $m<n$ we see that the first line of \eqref{equ:proofRC:contain} contains the fixpoint variables and $\Ct$ terms of the second line of \eqref{equ:proofRC:contain}. We can therefore apply \REFlem{lem:contain} (i) and (ii) which immediately proves the claim. 
\end{proof}

Using this result, we are finally ready to prove the induction step of \REFprop{prop:RC_inductive}.

\begin{proof}[Proof of \REFprop{prop:RC_inductive}]
Recall that \REFprop{prop:RC_inductive} trivially holds for $j=k$ which constitutes the base case of an induction over $j$. Now let us prove the induction step. Hence, let us assume that \REFprop{prop:RC_inductive} holds for $j$.
Now consider \enquote{$j-1$}, i.e., consider the permutation prefix $\delta'=p_0\hdots p_{j-2}$ and pick any $p_{j-1}\in P_{\delta'}$. By the induction hypothesis, we know that \REFprop{prop:RC_inductive} holds for $\delta=p_0\hdots p_{j-1}$ and any choice of $p_j\in\Pt{\delta}$. That is, $Z^*_{\delta p_j}$ can be computed using \eqref{equ:Zstarj:b}. With this, the fixpoint algorithm in \eqref{equ:Zstarj:a} for $\delta'$ and $p_{j-1}$ simplifies to
 \begin{align*}
 \textstyle Z^*_{\delta'p_{j-1}}=Z^*_{\delta}=&\textstyle \nu Y_{p_{j-1}}.~\mu X_{p_{j-1}}.\bigcup_{p_j\in \Pt{\delta}} Z^*_{\delta p_j}.
\end{align*}
Here, for any choice $p_j\in \Pt{\delta}$, the term $Z^*_{\delta p_j}$ is given by \eqref{equ:Zstarj:b} where $r_0=p_j$ and $\beta_{p_j}=r_1\hdots r_n$ being the \emph{maximal permutation sequence} of $\Pt{\delta p_j}$. Now observe that for $j>0$ and any choice of $p_j$ we see that $\gamma=r_0\hdots r_n$ is actually a permutation sequence of $\Pt{\delta}$, but not necessarily the maximal one. However, observe that the maximal permutation sequence $\beta$ of $\Pt{\delta}$ (that is $\beta=\beta_{\Downarrow}$) is actually defined by $\beta=\tilde{p}_j\beta_{\tilde{p}_j}$ for $\tilde{p}_j:=\max(\Pt{\delta})$. With this, we can apply \REFlem{lem:RC_d} to see that $Z^*_{\delta p_j}\subseteq Z^*_{\delta \tilde{p}_j}$ for all $p_j\in\Pt{\delta}$. With this we obtain
 \begin{align*}
  \textstyle Z^*_{\delta'p_{j-1}}=Z^*_{\delta}=&\textstyle \nu Y_{p_{j-1}}.~\mu X_{p_{j-1}}.~ Z^*_{\delta \tilde{p}_j}.
 \end{align*}
One can now verify that this allows us to choose $r_0=p_{j-1}$, $r_1=\tilde{p}_j$ and $r_2\hdots r_{n+1}=\beta_{\tilde{p}_j}$ and have $r_1>r_2>\hdots r_{n+1}$. Hence, $Z^*_{\delta'p_{j-1}}$ can be written in the form of \eqref{equ:Zstarj:b}, which proves the statement. 
\end{proof}

\subsubsection{Fair Adversarial Parity Games}\label{app:Parity}
We now consider a \emph{parity} winning condition with a set $\FP=\set{C_1,C_2,\hdots\allowbreak C_{2k}}$, where each $C_i\subseteq V$ is the set of vertices of $\game$ with color $i$. 
 Further, $\FP$ partition's the set of vertices, i.e., $\bigcup_{i\in[1,2k]}C_i=V$ and $C_i\cap C_j=\emptyset$ for all $i,j\in[0,2k-1]$ such that $i\neq j$.

 \begin{theorem*}[\REFthm{thm:Parity_all} restated for convenience]
Let $\Gl =\tup{\game, \El}$ be a game graph with live edges and 
$\FP$ be a parity condition over $\game$ with $2k$ colors. 
Further, let 
\begin{align}\label{p:equ:Parity_live}
 \textstyle Z^*\coloneqq &\nu Y_{2k}.~\mu X_{2k-1}.\hdots\nu Y_{2}.~\mu X_{1}.\\
  &~\cup(C_{2k}\cap \cpre(Y_{2k}))\cup ((C_1\cup\hdots\cup C_{2k-1})\cap \apre(Y_{2k},X_{2k-1}))\notag\\
   &~\cup~\hdots\notag\\
   &~\cup(C_{4}\cap \cpre(Y_{4}))\cup ((C_1\cup C_2\cup C_{3})\cap \apre(Y_4,X_{3}))\notag\\
 &~\cup (C_2\cap \cpre(Y_2))\cup (C_1\cap \apre(Y_2,X_1))\notag
\end{align}
Then $Z^*$ is equivalent to the winning region $\WlR$ of $\p{0}$ in the fair adversarial game over~$\Gl$ for the winning condition $\vP$ in \eqref{equ:vP}. Moreover, the fixpoint algorithm runs in $O(n^{k+1})$ symbolic steps, and a memoryless winning strategy for $\p{0}$ can be extracted from it.
\end{theorem*}

\begin{proof}
 A parity winning condition $\FP$ with $2k$ colors corresponds to the Rabin chain winning condition 
\begin{align}\label{equ:Fcond}
 &\set{\tuple{F_{2},F_{3}},\hdots,\tuple{F_{2k},\emptyset}}\quad\text{s.t.}\quad F_{i}:=\bigcup_{j=i}^{2k} C_{j},
\end{align}
which has $k$ pairs.
Translating the Rabin chain condition induced by $\FP$ in \eqref{equ:Fcond} into a Rabin condition as in \REFthm{thm:FPsoundcomplete} we get the tuple $\FR = \set{\tuple{G_1,R_1},\hdots,\tuple{G_k,R_k}}$ such that
\begin{subequations}\label{equ:RCprop_2}
 \begin{align}
 R_i= &F_{2i+1}=\textstyle\bigcup_{j=2i+1}^{2k} C_j\label{equ:RCprop_2:a}\\
 \overline{R}_i= &\textstyle\bigcup_{j=1}^{2i}C_j\label{equ:RCprop_2:b}\\
 G_i= & F_{2i}=\textstyle\bigcup_{j=2i}^{2k} C_j\label{equ:RCprop_2:c}\\
 \overline{R}_i\cap G_i= & C_{2i}\label{equ:RCprop_2:d}
\end{align}
\end{subequations}
 
First, observe that $R_0=G_0=\emptyset$ have been artificially introduced, and result in $\Ct_0=\apre(Y_0,X_0)$. 
Further, as we have assumed that $\FP$ is such that $\bigcup_{i\in[1,2k]}C_i=V$, we can equivalently write
\begin{align*}
 \Ct_0=&\left(\bigcup_{j=1}^{2k}C_j\right) \cup \apre(Y_0,X_0)
 = ((C_1\cup\hdots\cup C_{2k})\cap \apre(Y_0,X_0))
\end{align*}
For $j>0$, by using \eqref{equ:RCprop_2} we observe that the definition of $\Ct_j$ in \eqref{equ:RabinC_all_Cpj} can be written as
 \begin{align*}
\Ct_{j}=&
 \left(C_{2j}\cap \cpre(Y_{j})\right)\cup \left(\textstyle\left(\bigcup_{l=1}^{2j}C_l\right)\cap \apre(Y_{j},X_{j})\right)\\
 =&\left(C_{2j}\cap \cpre(Y_{j})\right)
 \cup\left(C_{1}\cap \apre(Y_{j},X_{j})\right)
 \cup\hdots\cup\left(C_{2j}\cap \apre(Y_{j},X_{j})\right).
\end{align*}
With this, we obtain the following fixpoint equation
   \begin{align}\label{equ:proof:RtoParity:a}
 \textstyle Z^*:=&\nu Y_{0}.~\mu X_{0}.~\nu Y_{k}.~\mu X_{k}.\hdots\nu Y_{1}.~\mu X_{1}.\\
 &~((C_1\cup\hdots\cup C_{2k})\cap \apre(Y_0,X_0))\notag\\
  &~\cup(C_{2k}\cap \cpre(Y_{k}))\cup ((C_1\cup\hdots\cup C_{2k})\cap \apre(Y_k,X_k))\notag\\
   &~\cup~\hdots\notag\\
 &~\cup (C_2\cap \cpre(Y_1))\cup ((C_1\cup C_2)\cap \apre(Y_1,X_1))\notag
\end{align}

Now consider \REFlem{lem:contain} and let us define
\begin{align*}
 g(X_0,Y_0):=&((C_1\cup\hdots\cup C_{2k})\cap \apre(Y_0,X_0))\\
 f(X_k,Y_k):=&\left(C_{2k}\cap \cpre(Y_{k})\right)\cup((C_1\cup\hdots\cup C_{2k})\cap \apre(Y_k,X_k)).
\end{align*}
It is immediately obvious that $g(X,Y)\subseteq f(X,Y)$ for all $X$ and $Y$. We can therefore apply \REFlem{lem:contain} (iv) and observe that the computation remains unchanged if we remove the fixpoint variables $X_0$ and $Y_0$. 

Now changing subscripts of iteration variables gives the following FP equation. 
   \begin{align}
 \textstyle Z^*:=&\nu Y_{2k}.~\mu X_{2k-1}.\hdots\nu Y_{2}.~\mu X_{1}.\\
  &~\cup(C_{2k}\cap \cpre(Y_{2k}))\cup ((C_1\cup\hdots\cup C_{2k})\cap \apre(Y_{2k},X_{2k-1}))\notag\\
   &~\cup~\hdots\notag\\
 &~\cup (C_2\cap \cpre(Y_2))\cup ((C_1\cup C_2)\cap \apre(Y_2,X_1))\notag
\end{align}

Now we recall from \REFlem{lem:AsubsetC} and \REFlem{lem:Ared} that for all $j$ such that $k\geq j\geq 1$ 
we have 
\begin{align*}
 (C_{2j}\cap \cpre(Y_j))\cup (C_{2j} \cap \apre(Y_j,X_j))=(C_{2j}\cap \cpre(Y_j)). 
\end{align*}

This yields
\begin{align}\label{equ:proof:RtoParity:b}
 \textstyle Z^*:=&\nu Y_{2k}.~\mu X_{2k-1}.\hdots\nu Y_{2}.~\mu X_{1}.\\
  &~\cup(C_{2k}\cap \cpre(Y_{2k}))\cup ((C_1\cup\hdots\cup C_{2k-1})\cap \apre(Y_{2k},X_{2k-1}))\notag\\
   &~\cup~\hdots\notag\\
 &~\cup (C_2\cap \cpre(Y_2))\cup (C_1\cap \apre(Y_2,X_1))\notag
\end{align}

\end{proof}

\begin{remark}
 For the reduction of \enquote{normal} Rabin chain games to parity games we would need to further simplify \eqref{equ:proof:RtoParity:b} for the special case where all  $\apre(Y,X)$  are substituted by with $\cpre(Y)$. In this case, however, we observe that in any valid iteration it always holds that $X_{i+1}\subseteq Y_i$ for all even $i$ and $X_{j+2}\subseteq X_j$ for all odd $j$. We can therefore remove all terms for particular colors that have already appeared in inner fixpoint computations. Doing this yields the normal fixpoint for parity games presented in \eqref{equ:Parity_normal}. For fair-adversarial parity games, this simplification is not possible due to the dependence of $\apre$ on both $Y$ and $X$. 
\end{remark}

\subsubsection{Fair Adversarial Generalized Co-Büchi Games}\label{app:CoBuechi}

\begin{theorem*}[\REFthm{thm:GCB} restated for convenience]
Let $\Gl =\tup{\game, \El}$ be a game graph with live edges and 
$\mathcal{A}$ be a generalized Co-B\"{u}chi winning condition $\game$ with $r$ pairs. 
Further, let 
\begin{align}\label{p:equ:GCB}
  \textstyle Z^*\coloneqq&\nu Y_{0}.~\mu X_{0}.~\bigcup_{a\in [1;r]} \nu Y_{a}.~
\apre(Y_{0},X_{0})\cup (A_a\cap \cpre(Y_a)).
\end{align}
Then $Z^*$ is equivalent to the winning region $\WlR$ of $\p{0}$ in the fair adversarial game over~$\Gl$ for the winning condition $\varphi$ in \eqref{equ:vGCB}. Moreover, the fixpoint algorithm runs in $O(rn^2)$ symbolic steps, and a memoryless winning strategy for $\p{0}$ can be extracted from it.
\end{theorem*}

In this section we prove \REFthm{thm:GCB}. That is, we prove that for generalized Co-Büchi conditions, the fixpoint computing $Z^*$ in \eqref{equ:Rabin_all} simplifies to the one in \eqref{p:equ:GCB}. This is formalized in the next proposition.

\begin{proposition}\label{prop:Co-Buechi}
  Let $\FR = \set{\tuple{G_1,R_1},\hdots,\tuple{G_k,R_k}}$ be a Rabin condition such that \eqref{equ:Rtilde} holds. Further let $Z^*$ be the fixed-point of the $\mu$-calculus formula \eqref{equ:Rabin_all} and $\Zt^*$ the fixed-point of \eqref{p:equ:GCB}. Then $Z^*=\Zt^*$.
 \end{proposition}
 
\begin{proof}
 Now consider the flattening of \eqref{equ:Rabin_all} in \eqref{equ:It_pj} for $\widetilde{\FR}$. Then we see that for all $j>0$ we have 
\begin{align*}
 \mathcal{C}_{\delta p_j i_j}&:=\left(Q_{\delta p_j}
 \cap \cpre(Y_{\delta p_j}^*)\right)
 \cup \left(Q_{\delta p_j}
 \cap \apre(Y_{\delta p_j}^*,X_{\delta p_j}^{i_j-1})\right)\\
 &=Q_{\delta p_j}
 \cap \left(\cpre(Y_{\delta p_j}^*)\cup\apre(Y_{\delta p_j}^*,X_{\delta p_j}^{i_j-1})\right)
\end{align*}
and we always have $X_{\delta p_j}^{i_j-1}\subseteq Y_{\delta p_j}^*$. With this, it follows from \REFlem{lem:AsubsetC} that 
\begin{align}\label{equ:proof:C_gcb}
  \mathcal{C}_{\delta p_j i_j}&=Q_{\delta p_j}
 \cap \cpre(Y_{\delta p_j}^*)
\end{align}
for all $\delta$, $p_j$ and $i_j$ with $j>0$.

Now observe that for $\delta'=\delta p_j i_j$ and all $p_{j+1}\in P\setminus\set{p_0,\hdots,p_j}$ we have 
\begin{equation*}
 Q_{\delta' p_{j+1}}= Q_{\delta p_{j}}\cap \overline{R}_{p_{j+1}}\subseteq Q_{\delta p_{j}}.
\end{equation*}
It further follows from the structure of the fixpoint in \eqref{equ:Rabin_all} that
\begin{align*}
 Y^*_{\delta p_j}&=\bigcup_{i_j>0} X_{\delta p_j}^{i_j
 }=\bigcup_{i_j>0}\bigcup_{p_{j+1}\in P\setminus{p_0,\hdots,p_j}}Y^*_{\delta' p_{j+1}}
\end{align*}
and therefore 
\begin{equation*}
 Y^*_{\delta' p_{j+1}}\subseteq Y^*_{\delta p_j}.
\end{equation*}
With this we get
\begin{equation*}
 \mathcal{C}_{\delta' p_{j+1} i_{j+1}}\subseteq\mathcal{C}_{\delta p_{j} i_j}
\end{equation*}
for all $\delta$, $p_j$ and $i_j$ with $j>0$.
Then it follows from \REFlem{lem:contain} (iii) that for every permutation sequence $\delta=p_0p_1\hdots p_k$ the union over all $\mathcal{C}'s$ terms simplifies to two terms, one for $j=0$ and one for $j=1$. Using this insight, we see that for the particular Rabin condition $\widetilde{\FR}$ the fixpoint algorithm in \eqref{equ:Rabin_all} simplifies to 
\begin{align}\label{equ:GR1_simple_1}
 \nu Y_{0}.~\mu X_{0}.~\bigcup_{p_1\in P} \nu Y_{p_1}.~ \mu X_{p_1}.~
\mathcal{C}_{p_0}\cup \mathcal{C}_{p_1}.
\end{align}
Now recalling that $\mathcal{C}_{p_1}$ simplifies to $A_a\cap \cpre(Y_a)$ for $a=p_1$ (see \eqref{equ:proof:C_gcb}) if \eqref{equ:Rtilde} holds, and that $\mathcal{C}_{p_0}=\apre(Y_0,X_0)$ as $R_0=Q_0=\emptyset$, we see that \eqref{equ:GR1_simple_1} coincides with \eqref{p:equ:GCB}.
\end{proof}


\subsection{Additional Proofs for \REFsec{sec:GenRabinGames}}\label{app:GenRabinGames}

\subsubsection{Proof of \REFthm{thm:sGB}}\label{app:proof:sGB}

\begin{theorem*}[\REFthm{thm:sGB} restated for convenience]
Let $\Gl =\tup{\game, \El}$ be a game graph with live edges and $\tuple{\Fc,Q}$ with $\mathcal{F}=\set{\ps{1}F,\hdots,\ps{s}F}$ a safe generalized Büchi winning condition.
Further, let 
\begin{align}\label{p:equ:sGB:FP}
  \textstyle Z^*\coloneqq&\nu Y.\bigcap_{b\in [1;s]} \mu \ps{b}X.~
  Q\cap 
  \left[
(\ps{b}F\cap  \cpre(Y))\cup \apre(Y,\ps{b}X)
\right].
\end{align}
Then $Z^*$ is equivalent to the winning region $\WlR$ of $\p{0}$ in the fair adversarial game over $\Gl$ for the winning condition $\varphi$ in \eqref{equ:sGB:psi}. 
Moreover, the fixpoint algorithm runs in $O(sn^2)$ symbolic steps, and a finite-memory winning strategy for $\p{0}$ can be extracted from it.
\end{theorem*}

Our goal is to prove \REFthm{thm:sGB} by a reduction to \REFthm{thm:SingleRabin} and \REFthm{thm:Reachability}. We therefore first show that a similar construction of an extended fixpoint $\Zt$ as in \eqref{equ:Bcomp} within the proof of \REFthm{thm:SingleRabin} also works for the generalized case. This is formalized in the following proposition.

\begin{proposition}\label{prop:ZeqZt}
 Given the premises of \REFthm{thm:sGB}, let
 \begin{subequations}\label{equ:GBcomp}
   \begin{align}\label{equ:GBcomp:a}
  \textstyle Z^*\coloneqq&\nu Y.~\bigcap_{b\in [1;s]}\mu \ps{b}X.~
 Q\cap 
  \left[(\ps{b}F\cap  \cpre(Y))\cup \apre(Y,\ps{b}X)\right]
\end{align}
and 
  \begin{align}\label{equ:GBcomp:b}
  \textstyle \Zt^*\coloneqq&\nu \Yt.~\bigcap_{b\in [1;s]} \nu \ps{b}\Yt.~\mu \ps{b}\Xt.~
 Q\cap 
  \left[(\ps{b}F\cap  \cpre(\Yt))\cup \apre(\ps{b}\Yt,\ps{b}\Xt)\right].
\end{align}
\end{subequations}
Then $\Zt^*=Z^*$.
 \end{proposition}

However, as in \eqref{equ:GBcomp} a conjunction is used to update $Y$, the proof is not as straight forward as for \eqref{equ:Bcomp}. 
We first show for both equations \eqref{equ:GBcomp:a} and \eqref{equ:GBcomp:b} that, upon termination, we have $Y^*=\ps{b}X^*$ for all $b\in [1;s]$. Both claims are formalized in \REFlem{lem:ZeqZt:a} and \REFlem{lem:ZeqZt:b}, respectively.

 \begin{lemma}\label{lem:ZeqZt:a}
 Given the premises of \REFprop{prop:ZeqZt}, let $\ps{b}X^i$ be the set computed in the $i$-th iteration over the fixpoint variable $\ps{b}X$ in \eqref{equ:GBcomp:a} during the last iteration over $Y$, i.e., $Y=Z^*$ already. Further, we define $\ps{b}X^0=\emptyset$ and $\ps{b}X^*:=\bigcup_{i>0} \ps{b}X^i$. 
  Then it holds that $Z^*=\ps{b}X^*$ for all $b\in [1;s]$.
 \end{lemma}

 \begin{proof}
  We fix $Y=Z^*$ and $b\subseteq[1;s]$ and observe from \eqref{equ:GBcomp:a} that 
  \begin{align*}
   \ps{b}X^{0}=(\ps{b}F\cap  \cpre(Z^*))
  \end{align*}
  and therefore
   \begin{align*}
   \ps{b}X^{1}&=\ps{b}X^{0}\cup(\ps{b}F\cap  \cpre(Z^*))\cup \apre(Z^*,\ps{b}X^{0})\\
   &=(\ps{b}F\cap  \cpre(Z^*))\cup\apre(Z^*,\ps{b}X^{0})\supseteq\ps{b}X^{0}
  \end{align*}
  With this, we have in general that 
  \begin{align*}
  \ps{b}X^{i+1}= &\ps{b}X^{i} \cup (\ps{b}F\cap  \cpre(Z^*))\cup \apre(Z^*,\ps{b}X^{i})\\
  =&(\ps{b}F\cap  \cpre(Z^*))\cup \apre(Z^*,\ps{b}X^{i})
  \end{align*}
  which implies $\ps{b}X^{i+1}\supseteq\ps{b}X^{i}$. Hence, $\ps{b}X^*:=\bigcup_{i\in[0,i_{max}]} \ps{b}X^i\allowbreak=\ps{b}X^{i_{max}}$, and therefore, in particular
  \begin{align}\label{equ:proof:lem:ZeqZt:a:1}
  \ps{b}X^*= (\ps{b}F\cap  \cpre(Z^*))\cup \apre(Z^*,\ps{b}X^*).
  \end{align}
  By recalling that $Z^*=\bigcap_{b}\ps{b}X^*$ we see that $Z^*\subseteq\ps{b}X^*$. 
  
  For the inverse direction, we use the observation $Z^*\subseteq\ps{b}X^*$ together with \REFlem{lem:Ared} to see that $\apre(Z^*,\ps{b}X^*)=\cpre(\ps{b}X^*)$. With this $(\ps{b}F\cap  \cpre(Z^*))\subseteq\cpre(Z^*)\subseteq\cpre(\ps{b}X^*)=\apre(Z^*,\ps{b}X^*)$ and hence \eqref{equ:proof:lem:ZeqZt:a:1} reduces to
   \begin{align*}
  \ps{b}X^*=\cpre(\ps{b}X^*)\supseteq\cpre(Z^*).
  \end{align*}
  As the last equality holds for all $b\subseteq[1;s]$ we see that
  \begin{align}\label{equ:proof:lem:ZeqZt:a:2}
   Z^*=\bigcap_{b}\ps{b}X^*=\bigcap_{b}\cpre(\ps{b}X^*)\supseteq\cpre(Z^*).
  \end{align}
  
  We can now use \eqref{equ:proof:lem:ZeqZt:a:2} to proof that $Z^*\supseteq\ps{b}X^*$ also holds. To show this, we pick a vertex $v\in\ps{b}X^*$ and prove that $v\in Z^*$. To that end, observe that either (i) $v\in (\ps{b}F\cap  \cpre(Z^*))\subseteq\cpre(Z^*)\subseteq Z^*$ which immediately proves the statement, or (ii)  $v\in\apre(Z^*,\ps{b}X^*)$. If (ii) holds we again have two cases. Either (a) $v\in \cpre(\ps{b}X^*)$ which implies that there exists a finite sequence $\cpre(\cpre(\hdots\cpre(\ps{b}X^1)\hdots))$ where $\ps{b}X^1=\ps{b}F\cap  \cpre(Z^*)\subseteq\cpre(Z^*)\subseteq Z^*$ and therefore $v\in \cpre(\cpre(\hdots\cpre(Z^*)\hdots))\subseteq Z^*$. Finally we could have (b) that $v\in \epre_l(\ps{b}X^*)\cap\eapre_1(Z^*)\subseteq\eapre_1(Z^*)\subseteq\cpre(Z^*)\subseteq Z^*$, which again proves the statement.
 \end{proof}

 \begin{lemma}\label{lem:ZeqZt:b}
 Given the premises of \REFprop{prop:ZeqZt}, let $\ps{b}Y^i$ be the set computed in the $i$-th iteration over the fixpoint variable $\ps{b}Y$ in \eqref{equ:GBcomp:b} during the last iteration over $Y$, i.e., $Y=\Zt^*$ already. Further, we define $\ps{b}Y^0=V$ and $\ps{b}Y^*:=\bigcap_{i>0} \ps{b}Y^i$. 
  Then it holds that $\Zt^*=\ps{b}Y^*$ for all $b\in [1;s]$.
  \end{lemma}
  
  \begin{proof}
  Recall that $\Zt^*=\bigcap_{b}\ps{b}Y^*$ from the structure of the fixpoint algorithm in \eqref{equ:GBcomp:b}. To prove $\Zt^*=\ps{b}Y^*$ for all $b\in [1;s]$ it therefore suffices to show that $\ps{b}Y^*=\ps{b'}Y^*$ for any two $b,b'\in[1;s]$ s.t.\ $b\neq b'$.
  
  Towards this goal, recall from \REFthm{thm:Reachability} that $\ps{b}Y^*$ is exactly the set of states from which player $0$ can win a fair adversarial reachability game with target $\ps{b}T:=\ps{b}F\cap  \cpre(\Zt^*)$. However, every state $v\in\ps{b}T$ allows player $0$ to force the game to a state $v'\in \Zt^*=\bigcap_{b'}\ps{b'}Y^*$. Therefore, by definition player $0$ has a strategy to reach a state $v'\in\ps{b'}Y^*$ from any state $v\in\ps{b}Y^*$ for any $b'\in[1;s]$ s.t.\ $b\neq b'$. As, however $\ps{b'}Y^*$ is defined as the winning region of player $0$ w.r.t.\ the goal set $\ps{b'}T:=\ps{b'}F\cap  \cpre(\Zt^*)$, we know that there actually exists a player $0$ strategy to drive the game from any $v\in\ps{b}Y^*$ to $\ps{b'}T$, and therefore, by definition $\ps{b}Y^*\subseteq\ps{b'}Y^*$. As this inclusion holds mutually for all $b,b'\in[1;s]$ s.t.\ $b\neq b'$ we have that $\ps{b}Y^*=\ps{b'}Y^*$. With this, it immediately follows that $\Zt^*=\ps{b}Y^*$ for all $b\in [1;s]$.
\end{proof}

With \REFlem{lem:ZeqZt:a} and \REFlem{lem:ZeqZt:b} in place, it remains to show that the retained fixpoints are indeed equivalent, which is achieved by the following lemma.

\begin{lemma}
 Given the premises of \REFprop{prop:ZeqZt} it holds that
 \begin{compactitem}[(i)]
  \item $Z^*\not\subset \Zt^*$, and
  \item $\Zt^*\not\subset Z^*$
 \end{compactitem}

\end{lemma}

\begin{proof}
 We show both claims by contradiction.\\
 \begin{inparaitem}[$\blacktriangleright$]
 \item \textit{(i)} Assume $Z^*\subset \Zt^*$. As $Y^0=V$ and $Z^*=Y^k$ for some $k>0$ this implies that there exists an $i>0$ s.t. $Y^i\supseteq \Zt^*\supset Y^{i+1}$.
As $Y^{i+1}= \bigcap_{b} \ps{b}X^{i*}$, this implies the existence of a $b\in[1;s]$  s.t.\ $\Zt^*\supset \ps{b}X^{i*}$, where
 \begin{equation*}
  \ps{b}X^{i*}=\mu \ps{b}X. Q\cap \left[(\ps{b}F\cap  \cpre(Y^i))\cup \apre(Y^i,\ps{b}X)\right]
 \end{equation*}
 On the other hand,
  \begin{equation*}
   \Zt^*=\ps{b}\Yt^{**}=\ps{b}X^{***}=\mu \ps{b}X. Q\cap \left[(\ps{b}F\cap  \cpre(\Zt^*))\cup \apre(\Zt^*,\ps{b}X)\right]
  \end{equation*}
  As $Y^i\supseteq \Zt^*$ it follows from monotonicity of all involved functions that $\ps{b}X^{i*}\supseteq \ps{b}X^{***}$ which yields a contradiction.\\
 \item \textit{(ii)} 
 Now we assume $\Zt^*\subset Z^*$. As $\Yt^0=V$ and $\Zt^*=\Yt^k$ for some $k>0$ this implies that there exists an $i>0$ s.t. $\Yt^i\supseteq Z^*\supset \Yt^{i+1}$. 

As $Y^{i+1}= \bigcap_{b} \ps{b}Y^{i*}$, this implies the existence of $b\in[1;s]$  s.t.\ $Z^*\supset  \ps{b}Y^{i*}$. We recall that 
 \begin{equation*}
  \ps{b}Y^{i*}=\nu \ps{b}Y. \mu \ps{b}X. Q\cap \left[(\ps{b}F\cap  \cpre(\Yt^i))\cup \apre(\ps{b}Y^i,\ps{b}X)\right]
 \end{equation*}
  Now observe that $\ps{b}Y^{i0}=V\supseteq Z^*$. Hence, for $Z^*\supset  \ps{b}Y^{i*}$ to be true there must exists a $j$ s.t.\ $Y^{ij}\supseteq Z^*\supset Y^{ij+1}$, where 
   \begin{equation*}
 \ps{b}Y^{ij+1}= \ps{b}X^{ij*}=\mu \ps{b}X. Q\cap \left[(\ps{b}F\cap  \cpre(\Yt^i))\cup \apre(\ps{b}Y^{ij},\ps{b}X)\right].
 \end{equation*}
 Now it is however easy to see that it follows from monotonicity again that we have $Y^{ij}\supseteq \Zt^*$
 whenever $Y^{ij}\supseteq Z^*$, which yields the intended contradiction.
  \end{inparaitem}
\end{proof}

Using \REFprop{prop:ZeqZt} we know that \eqref{equ:GBcomp:a} and \eqref{equ:GBcomp:b} compute the same set. Hence, we can use \eqref{equ:GBcomp:b} instead of \eqref{p:equ:sGB:FP} to prove \REFthm{thm:sGB}. This allows us to simply reduce the proof of \REFthm{thm:sGB} to \REFthm{thm:SingleRabin} and \REFthm{thm:Reachability} as formalized below.
 
 \begin{proof}[Proof of \REFthm{thm:sGB}] 
\noindent\textbf{Soundness \& Completeness:}
Let us define $Z^*(\tuple{T,Q})$ to be the set of states computed by the fixpoint algorithm in \eqref{equ:Reach:FP}. Then it follows from \eqref{equ:GBcomp:b} that 
   \begin{equation*}
   \Zt^*=\nu Y.~\bigcap_{b\in [1;s]} Z^*(\tuple{Q\cap \ps{b}F \cap \cpre(Y),Q}).
 \end{equation*}
 In particular, it follows from \REFlem{lem:ZeqZt:b} that 
 \begin{equation*}
  \Zt^*=Z^*(\tuple{Q\cap \ps{b}F \cap \cpre(\Zt^*),Q})~\forall b\in [1;s].
 \end{equation*}
Now let us define $\ps{b}\mathcal{W}$ to be the fair adversarial winning state set for 
 \begin{equation*}
  \ps{b}\psi=\Box Q\wedge \Box\Diamond\ps{b}F.
 \end{equation*}
 With this, it follows from \REFthm{thm:SingleRabin} that $\Zt^*=\ps{b}\mathcal{W}$ for all $b\in [1;s]$. Therefore, we obviously have $\bigcap_{b\in[1;s]}\ps{b}\mathcal{W}=\Zt^*$. Now let $\mathcal{W}$ be the fair adversarial winning set w.r.t.\ 
  \begin{equation*}
  \psi=\Box Q\wedge \bigwedge_{b\in[1;s]}\Box\Diamond(\ps{b}F).
 \end{equation*}
 (compare \eqref{equ:vgR}). Then we always have $\mathcal{W}\subseteq\bigcap_{b\in[1;s]}\ps{b}\mathcal{W}$ which immediately implies $\mathcal{W}\subseteq \Zt^*$. However, as $\ps{a}\mathcal{W}=\ps{b}\mathcal{W}$ for all $a,b\in[1;s]$, we know that $\psi$ holds for all $v\in \Zt^*$, hence $Z^*\subseteq\mathcal{W}$.

\smallskip
\noindent\textbf{Strategy construction:}
We can define a rank function for every $b$ as in \eqref{equ:Reach:rho} within the proof of \REFthm{thm:Reachability} (see \REFapp{app:prop:Reachability}), i.e., 
 \begin{equation}
  \ps{b}\rank{v}=i\quad\text{iff}\quad v\in\ps{b}X^{i}\setminus\ps{b}X^{i-1}.
 \end{equation} 
 Then, we have a different strategy, $\ps{b}\rho_0$, which is defined via \eqref{equ:Reach:rho} (see \REFapp{app:prop:Reachability}) using the corresponding $\ps{b}\rank{}$ function. With this, we define a new strategy $\rho$ which circles through all possible goal sets in a pre-defined order. That is 
 \begin{align}\label{equ:strategy:gR}
  \rho_0(v,b)=
  \begin{cases}
   \ps{b}\rho_0(v) & v\notin \ps{b}F\\
   \ps{b^+}\rho_0(v) & v\in \ps{b}F
  \end{cases}
 \end{align}
where $b^+=b+1$ if $b<s$ and $b^+=1$ if $b=s$. 

The strategy in \eqref{equ:strategy:gR} is obviously winning for $\psi$ in \eqref{equ:vgR} as every $\ps{b}\rho_0$ is a winning strategy for~$\ps{b}\psi$ (from \REFthm{thm:SingleRabin}) and upon reaching $\ps{b}F$ we know that the respective state $v$ is also contained in $\cpre(\Zt^*)$ where $\Zt^*=\ps{b^+}Y^*$. Now it follows from the definition of $\cpre$ that $\cpre(\ps{b^+}Y^*)\subseteq \ps{b^+}Y^*$, hence, allowing to apply $\ps{b^+}\rho_0$ upon reaching $\ps{b}F$.
\end{proof}

\subsubsection{Proof for \REFthm{thm:GenRabin}}\label{app:GenRabinproof}

\begin{theorem*}[\REFthm{thm:GenRabin} restated for convenience]
Let $\Gl =\tup{\game, \El}$ be a game graph with live edges and 
$\FgR$ be a generalized Rabin condition over $\game$ with index set $P=[1;k]$. 
Further, let
 \begin{subequations}\allowdisplaybreaks
 \label{p:equ:GenRabin_all}
   \begin{align}
    Z^*:=&\nu Y_{0}.~\mu X_{0}.~\notag\\
         &\bigcup_{p_1\in P}  \nu Y_{p_1}.~\bigcap_{l_1\in[1;m_{p_1}]}\mu \ps{l_1}X_{p_1}.~\\
&\qquad\ddots\notag\\
 &\qquad\bigcup_{p_k\in P\setminus\set{p_1,\hdots, p_{k-1}}}\hspace{-0.3cm}\nu Y_{p_k}.~\bigcap_{l_k\in[1;m_{p_k}]}\hspace{-0.3cm}\mu \ps{l_k}X_{p_k}.~
 \bigcup_{j=0}^k \ps{l_j}\mathcal{C}_{p_j},\notag
\end{align}
where
 \begin{align*}
 \ps{l_j}\mathcal{C}_{p_j} :=
\left(\bigcap_{i=0}^{j} \overline{R}_{p_i}\right)\cap
 \left[ 
 \left( \ps{l_j}G_{p_j}\cap \cpre(Y_{p_j})\right)
 \cup \apre(Y_{p_j},\ps{l_j}X_{p_j})
 \right]
\end{align*}
 \end{subequations}
with $p_0=0$, $G_{p_0}\coloneqq\set{\emptyset}$ and $R_{p_0}\coloneqq\emptyset$. 
Then $Z^*$ is equivalent to the winning region $\WlR$ of $\p{0}$ in the fair adversarial game over $\Gl$ for the winning condition $\vR$ in \eqref{equ:vgR}. Moreover, the fixpoint algorithm runs in $O(n^{k+2}k! m_1 \ldots m_k)$ symbolic steps, and a finite-memory winning strategy for $\p{0}$ can be extracted from it.
\end{theorem*}	

We show how the proof of \REFthm{thm:FPsoundcomplete} in \REFapp{appD} needs to be adapted in order to prove the generalized version of \REFthm{thm:FPsoundcomplete}, namely \REFthm{thm:GenRabin}, instead.

\smallskip
\noindent\textbf{Strategy Construction:}
Similar to the finite-memory strategy constructed for generalized Büchi games in \REFapp{app:proof:sGB}, the strategy for 
generalized Rabin games needs to remember the index of all the goal sets currently \enquote{chased} for each permutation index up to $p_j$. To formalize this, we define the set of full goal chain sequences for a given generalized Rabin specification $\FgR$ by 
\begin{align}
 \Phi(\FgR):=\SetComp{\ell_0\ell_1\hdots \ell_k}{\ell_0=1,~\ell_j\in[0;m_j]}.
\end{align}
If $\FgR$ is clear from the context we simply write $\Phi$. Given a  goal chain prefix $\phi:=\ell_0\ell_1\hdots \ell_{j-1}$ we can now construct a ranking for each such prefix, using the flattening of \eqref{p:equ:GenRabin_all} instead of \eqref{equ:Rabin_all}. This yields the following proposition which follows from \REFprop{prop:flattening} by simply annotating all terms with the goal chain prefix $\phi$.

\begin{proposition}\label{prop:GenR:flattening}
 Let $\delta=p_0i_0\hdots p_{j-1} i_{j-1}$ be a \emph{configuration prefix}, $\phi:=\ell_0\ell_1\hdots \ell_{j-1}$ a \emph{goal chain prefix}, $p_j\in P\setminus\set{p_1,\hdots,p_{j-1}}$ the next permutation index, $\ell_j\in[1;m_{p_j}]$ the next goal set and $i_{j}>0$ a counter for $p_j$. Then the flattening of \eqref{p:equ:GenRabin_all} for this configuration and goal prefix is given by
\begin{subequations}\label{equ:GenR:It_pj}
\begin{align}\label{equ:GenR:Xdeltai}
 \ps{\phi \ell_j}X_{\delt p_j}^{i_j}=
 &
 \underbrace{
 \ps{\phi}S_{\delta}\cup
\ps{\ell_j}\mathcal{C}_{\delta p_j i_j}
 }_{\ps{\phi \ell_j}S_{\delta p_j i_j}}
 \cup\ps{\phi \ell_j}\mathcal{A}_{\delta p_j i_j}
\end{align}
where 
 \begin{align}
&Q_{p_0\hdots p_a}:=
\bigcap_{b=0}^{a} \overline{R}_{p_b},\label{equ:GenR:Qdelta}\\ 
&\ps{\ell_a}\mathcal{C}_{\delta p_a i_a}:=\left(Q_{\delta p_a}
 \cap \ps{\ell_a}G_{p_a}\cap \cpre(Y_{\delta p_a}^*)\right)
 \cup \left(Q_{\delta p_a}
 \cap \apre(Y_{\delta p_a}^*,\ps{\ell_a}X_{\delta p_a}^{i_a-1})\right)\notag\\
 &\ps{\ell_0\hdots \ell_a}S_{p_0i_0\hdots p_{a}i_{a}}:=\bigcup_{b=0}^{a} \ps{\ell_b}\mathcal{C}_{p_0i_0\hdots p_b i_b},\label{equ:GenR:Sdelta}\\
&\ps{\phi \ell_i}A_{\delta p_j i_j}:=\textstyle\bigcup_{p_{j+1}\in P\setminus\set{p_1,\hdots,p_{j}}} \left(\textstyle\bigcap_{\ell_{j+1}\in[1;m_{p_{j+1}}]} \left(\textstyle\bigcup_{i_{j+1}>0} 
\left(\ps{\phi \ell_j \ell_{j+1}}X^{i_{j+1}}_{\delta p_j i_j p_{j+1}}
\setminus \ps{\phi \ell_i}S_{\delta p_j i_j}
\right)\right)\right).
\label{equ:GenR:Adeltai}
\end{align}
\end{subequations}
\end{proposition}

Again we see that this flattening follows directly from the structure of the fixpoint algorithm in \eqref{p:equ:GenRabin_all} and the definition of $\ps{l_j}\mathcal{C}_{p_j}$ in \eqref{equ:GenRabin_all_Cpj}. Using the flattening of \eqref{p:equ:GenRabin_all} in \eqref{equ:GenR:It_pj} we can define a ranking function for each goal chain prefix $\phi$ identical to  \REFdef{def:ranking}. That is, given the premises of \REFprop{prop:GenR:flattening}, we define $\ps{\phi \ell_j}R:V\rightarrow 2^{\Dt}$ s.t.\
 \begin{inparaenum}[(i)]
  \item $\infty \in \ps{\phi \ell_j}R(v)$ for all $v\in V$, and
  \item $\delta p_j i_j \underline{\gamma}\in \ps{\phi \ell_j}R(v)$ iff $v\in \ps{\phi \ell_j}S_{\delta p_j i_j}$.
 \end{inparaenum}
The ranking function $\ps{\phi}\rank{}:V\rightarrow D$ is then again defined as in \REFdef{def:ranking} s.t. $\ps{\phi}\rank{}: v \mapsto \min\set{\ps{\phi}R(v)}$. Similarly, we can define a memoryless winning strategy for every fixed goal sequence $\phi$ as in \eqref{equ:strategy}. That is, 
 \begin{align}
 \ps{\phi}\rho_0(v):=\min_{(v,w)\in E} (\ps{\phi}\rank{w}).
\end{align}
Now, similar to the proof of \REFthm{thm:sGB} (see \REFsec{thm:sGB}) we can \enquote{stack} these memoryless winning strategies to define a new strategy with finite memory which circles through all possible goal sets in a pre-defined order. That is
 \begin{align}\label{equ:strategy:GenRabin}
  \rho_0(v,\phi \ell_j):=
  \begin{cases}
   \ps{\phi \ell_j}\rho_0(v) & v\notin \ps{\ell_j}F\\
   \ps{\phi \ell_j^+}\rho_0(v) & v\in \ps{\ell_j}F
  \end{cases}
 \end{align}
where $\ell_j^+:=\ell_j+1$ if $\ell_j<m_{p_j}$ and $\ell_j^+:=1$ if $\ell_j=m_{p_j}$. 

Using this goal chain dependent ranking function, the proof of soundness and completeness of \eqref{p:equ:GenRabin_all} along with the proof that $\rho_0$ in \eqref{equ:strategy:GenRabin} is indeed a winning strategy for player $0$ in the fair adversarial generalized Rabin game, follows exactly the same lines as the proof in \REFapp{appD}. That is, 
 we iteratively consider instances of the flattening in \eqref{equ:GenR:It_pj}, starting with $j=k$ as the base case, and doing an induction from \enquote{$j+1$} to \enquote{$j$}.
To this end, we consider a \emph{generalized} local winning condition which refers not only to the current configuration-prefix $\delta=p_0i_0\hdots p_{j-1}i_{j-1}$ but also to the current goal chain prefix $\phi:=\ell_0\hdots \ell_{j-1}$. Hence, \eqref{equ:psi_delta} gets modified to 
\begin{align}
\ps{\phi}\psi_{\delta p_j}:=
\left(
\begin{array}{rl}
&Q_{\delta p_j}\mathcal{U} \ps{\phi}S_{\delta}\\
\vee &\Box Q_{\delta p_j}\wedge \bigwedge_{\ell_j\in[1;m_{p_j}]}\Box\Diamond \ps{\ell_j}G_{p_j}\\
\vee &\Box Q_{\delta p_j} \wedge \left(\displaystyle\bigvee_{i\in \Pt{j}}\left(\Diamond\Box\overline{R}_{i}\wedge \bigwedge_{b\in[1;m_{i}]}\Box\Diamond \ps{b}G_{i}\right)\right) 
\end{array}
\right)
\end{align}
where $\Pt{j}=P\setminus\set{p_0,\hdots,p_{j}}$.
With this, it becomes obvious that the proof of soundness, completeness and the winning strategy for \REFthm{thm:GenRabin} follows exactly the same reasoning as in \REFapp{appD} while additionally using \REFthm{thm:sGB} to reason about the conjunction over goal sets. 

The only remaining part to be shown concerns the last line of $\ps{\phi}\psi_{\delta p_j}$. For this, we recall from \REFapp{sec:theory:sound} that the induction step from \enquote{$j+1$} to \enquote{$j$} relies on the fact that 
\begin{align}\label{equ:GR:Psi}
  \ps{\phi \ell_j}\Psi_{\delta p_j}:=&\Box Q_{\delta p_j} \wedge\Diamond\left(\textstyle\bigvee_{p_{j+1}\in P\setminus\set{p_1,\hdots,p_{j}}}  \ps{\phi'}\psi'_{\delta' p_{j+1}}\right)
\end{align}
is indeed equivalent to the last line of $\ps{\phi}\psi_{\delta p_j}$, where $\ps{\phi'}\psi'_{\delta' p_{j+1}}$ denotes the last two lines of $\ps{\phi'}\psi_{\delta' p_{j+1}}$ with $\phi':=\phi \ell_j$ and $\delta':=\delta p_j$.

For (non-generalized) Rabin games this equivalence is proved in \REFapp{appD:proofPsi}. It can be seen by inspection within this proof, that using a conjunction over goal sets instead of a single goal set within the second and third line of $\ps{\phi}\psi_{\delta p_j}$ does not change any step in the derivation. Therefore, the same derivation can be used in the generalized case and is therefore omitted. This concludes the proof of \REFthm{thm:GenRabin}.

\subsubsection{Proof of \REFthm{thm:GR1}}\label{app:GRoneproof}

\begin{theorem*}[\REFthm{thm:GR1} restated for convenience]
Let $\Gl =\tup{\game, \El}$ be a game graph with live edges and $(\mathcal{A},\Fc)$ a GR(1) winning condition.
Further, let 
\begin{align*}
Z^*=&\nu Y_k.~\bigcap_{b\in[1;s]}\mu \ps{b}X_k.~\bigcup_{a\in [1;r]} \nu Y_{a}.~
(F_b\cap\cpre(Y_k))\cup \apre(Y_k,\ps{b}X_k)\cup (\overline{A}_a\cap \cpre(Y_a)).\notag
\end{align*}
Then $Z^*$ is equivalent to the winning region $\WlR$ of $\p{0}$ in the fair adversarial game over $\Gl$ for the winning condition $\varphi$ in \eqref{equ:vGRo}. Moreover, the fixpoint algorithm runs in $O(n^2r s)$ symbolic steps, and a finite-memory winning strategy for $\p{0}$ can be extracted from it.
\end{theorem*}

Within this section we proof \REFthm{thm:GR1}. That is, we prove that for GR(1) winning conditions, the fixpoint computing $Z^*$ in \eqref{p:equ:GenRabin_all} simplifies to the one in \eqref{equ:GR1:FP}. This is formalized in the next proposition.

\begin{proposition}\label{prop:GR1proof}
  Let $\FgR$ be a generalized Rabin condition with $k$ pairs s.t.\ \eqref{equ:GR1toRabin} holds for $r:=k-1$. Further let $Z^*$ be the fixed-point of the $\mu$-calculus formula \eqref{p:equ:GenRabin_all} and $\tilde{Z}^*$ be the fixed-point of \eqref{equ:GR1:FP}. Then $Z^*=\tilde{Z}^*$.
 \end{proposition}
 If \REFprop{prop:GR1proof} holds, we immediately see that \REFthm{thm:GR1} directly follows from \REFthm{thm:GenRabin}. It therefore remains to prove that  \REFprop{prop:GR1proof} holds.
 
 \begin{proof}
  First, consider an arbitrary permutation sequence $\delta=p_0\hdots p_k$. Then we know that there exists exactly one $j>0$ s.t.\ $p_j=k$ and all other indices come from the set $[1;r]$. We can therefore define $\gamma'=p_1\hdots p_{j+1}$ and $\gamma''=p_{j+1}\hdots p_k$ s.t.\ $p_i\in[1;r]$ for all $i\neq j$. We note that $\gamma'=\varepsilon$ if $j=1$ and $\gamma''=\varepsilon$ if $j=k$. With this we have $\delta=p_0\gamma' p_j \gamma''$.
  
  By inspecting \eqref{equ:GR1toRabin} we see that the first $r$ pairs of the generalized Rabin condition induced by the GR(1) specification actually form a Generalized Co-Büchi condition (compare \eqref{equ:Rtilde} in \REFsec{sec:SimpleRabinGames}). Hence, given a permutation sequence $\delta=p_0\gamma' p_j \gamma''$ we can use the same reasoning as in the proof of \REFthm{thm:GCB} in \REFapp{app:CoBuechi} to see that
   \begin{align}\label{equ:proof:GR1:chain}
    &\mathcal{C}_{p_1}\supseteq\hdots\supseteq\mathcal{C}_{p_{j-1}}~\text{and}~
    \mathcal{C}_{p_{j+1}}\supseteq\hdots\supseteq\mathcal{C}_{p_k}.
   \end{align}
  
  Now recall from the proof of \REFthm{thm:RabinC_all} in \REFapp{app:RabinChain} that these inclusions allow to recursively apply \REFlem{lem:contain} to delete all $\mathcal{C}$ terms which are included in either $\mathcal{C}_{p_1}$ or $\mathcal{C}_{p_{j+1}}$ along with the fixpoint variables used within these terms (compare \REFlem{lem:RC_a} where now $\gamma'$ and $\gamma''$ are interpreted as decreasing sub-sequences). Applying these simplifications to \eqref{p:equ:GenRabin_all} (in exactly the same manner as these simplifications where applied to \eqref{equ:Rabin_all} in the proof of \REFthm{thm:RabinC_all}) results in a simpler fixpoint algorithm where all permutation sequences have the form $\delta=0q_1kq_2$ with $q_1\neq q_2$ and $q_1,q_2\in[1;r]$ (here $q_1$ and $q_2$ correspond to $p_1$ and $p_{j+1}$ in \eqref{equ:proof:GR1:chain}, and $k$ corresponds to $p_j$). 
  
  Now we can inspect \eqref{equ:GR1toRabin} again to see that $R_{i}\supseteq R_{k}$ and $G_{i}\supseteq \ps{b}G_{p_j}$ for all $i\in[1;r]$ and $b\in[1;s]$. This can be understood as a \enquote{generalized Rabin chain condition} (compare \eqref{equ:RCprop_1} in \REFsec{sec:SimpleRabinGames}). Hence, we can apply \REFlem{lem:RC_a} one more time, now to the \enquote{decreasing sub-sequence} $q_1k$ within every permutation sequence. Again, utilizing this argument iteratively in \eqref{p:equ:GenRabin_all} yields a simpler fixpoint algorithm which only contains permutation sequences $\delta=0ka$ with $a\in[1;r]$. This proves that $Z^*$ is equivalent to the set 
   \begin{align*}
 &\nu Y_{0}.~\mu X_{0}.~\nu Y_k.\bigcap_{b\in [1;s]}\mu \ps{b}X_{0}.\bigcup_{a\in [1;r]} \nu Y_{a}.~ \mu X_{a}.
\quad\mathcal{C}_{p_0}\cup \ps{b}\mathcal{C}_{k}\cup \mathcal{C}_{a}.
\end{align*}
 
 Now inserting the simplifications for terms from the generalized Co-Büchi part (see \eqref{equ:proof:C_gcb} in \REFapp{app:CoBuechi}) and using $R_0=G_0=\emptyset$, we obtain 
    \begin{align*}
 &\nu Y_{0}.~\mu X_{0}.~\nu Y_k.~\bigcap_{b\in [1;s]}\mu \ps{b}X_{0}.~\bigcup_{a\in [1;r]} \nu Y_{a}.~\\
&\quad\apre(Y_{0},X_{0})
\cup (\ps{b}F\cap\cpre(Y_k))\cup \apre(Y_k,\ps{b}X_k)
\cup (\overline{A}_a\cap \cpre(Y_a)).
\end{align*}
Now we can apply \REFlem{lem:contain} (iii) again to remove the first occurrence of the $\apre$ term to obtain the same expression as in \eqref{equ:GR1:FP}. This concludes the proof.
 \end{proof}


\subsection{Additional Proofs for \REFsec{sec:stochastic}}\label{app:stoch}

\subsubsection{Preliminaries}
\smallskip\noindent
\textbf{$\onehalf$-player game:}
A special case of $\twohalf$-player game graphs is a \emph{Markov Decision Process} (MDP) or \emph{$\onehalf$-player game}, which is obtained by assuming that every $\p{0}$ vertex in $V_0 $ has only one outgoing edge.\footnote{Alternatively, we could also define $\onehalf$-player game graphs by restricting the outgoing edges from the $\p{1}$ vertices; our choice is actually tailored for the content of the rest of the section.}
Analogously to the $\twohalf$-player games, for a given $\onehalf$-player game graph $\game$, we use the notation $P_{v^0}^{\rho_1}(\game\models \varphi)$ to denote the probability of occurrence of the event $\game\models \varphi$ when the runs initiate at $v^0$ and when $\p{1}$ uses the strategy $\rho_1$.

\smallskip\noindent
\textbf{Role of end components in $\onehalf$-player game:}
Limiting behaviors in a $\onehalf$-player game can be characterized using the structure of the underlying game graph.
We summarize one key technical argument in the following.

Let $\game = \tuple{V,V_0,V_1,V_r,E}$ be a $\onehalf$-player game graph.
A set of vertices $U\subseteq V$ is called \emph{closed} if (1) for every $v\in U\cap V_r$, $E(v)\subseteq U$, and (2) for every $v\in U\cap (V_0\cup V_1)$, $E(v)\cap U\neq \emptyset$.
A closed set of vertices $U$ induces a \emph{subgame graph} $(V',V_0',V_1',V_r',E')$, denoted by $\game\downarrow U$, which is itself a $\onehalf$-player game graph and is defined as follows:
\begin{itemize}
	\item $V'=U$,
	\item $V_0'=U\cap V_0$,
	\item $V_1'=U\cap V_1$,
	\item $V_r'=U\cap V_r$, and
	\item $E' = E\cap (U\times U)$.
\end{itemize}
A set of vertices $U\subset V$ of a $\onehalf$-player game graph $\game$ is an \emph{end component} if (a) $U$ is closed, and (b) the subgame graph $\game\downarrow U$ is strongly connected.

Denote the set of all end components of $\game$ by $\mathcal E\subset 2^V$. 
The next lemma states that under every strategy $\rho_1$ (being memoryless or not) of $\p{1}$ in the $\onehalf$-player game, the set of states visited infinitely often along a play is an end component with probability one. 

\begin{lemma}\cite[Thmeorem~3.2]{de1997formal}
For every $\onehalf$-player game graph, for every vertex $v\in V$, and every $\p{1}$ strategy $\rho_1$, 
\begin{equation}
	P_v^{\rho_1}\left( \game\models \bigvee_{U\in \mathcal{E}} \left(\lozenge\square U \wedge \bigwedge_{u\in U} \square\lozenge u\right)\right) = 1.
\end{equation}
\end{lemma}

This lemma implies the following corollary, which is motivated by similar claim for Rabin winning conditions in the literature \cite{DBLP:conf/icalp/ChatterjeeAH05}.
 
\begin{corollary}
\label{cor:good end component}
For a given $\onehalf$-player game, for a given vertex $v\in V$, and for a given $\p{1}$ strategy $\rho_1$, a generalized Rabin condition $\FgR = \set{\tuple{\Gc_1,R_1},\hdots,\tuple{\Gc_k,R_k}}$ is satisfied almost surely if and only if 
for every end component $U$ reachable from $v^0$, there is a $j\in\set{1,2,\ldots,k}$ such that $U\cap R_j = \emptyset$ and for every $l\in [1;m_j]$, $U\cap \ps{l}G_j\ne \emptyset$.
\end{corollary}

\subsubsection{Proof of \REFthm{thm:Reduction}}

\begin{theorem*}[\REFthm{thm:Reduction} restated for convenience]
 Let $\game$ be a $\twohalf$-player game graph, $\FgR$ be a generalized Rabin condition, $\varphi \subseteq V^\omega$ be the corresponding LTL specification (Eq.~\eqref{equ:vgR}) over the set of vertices $V$ of $\game$, and $\Dr(\game)$ be the reduced two-player game graph.
 Let $\WlR\subseteq \widetilde{V}$ be the set of all the vertices from where $\p{0}$ wins the fair adversarial game over $\Dr(\game)$ for the winning condition $\varphi$, and $\mathcal{W}^{\mathit{a.s.}}$ be the almost sure winning set of $\p{0}$ in the game graph $\game$ for the specification $\varphi$.
 Then, $\WlR = \mathcal{W}^{\mathit{a.s.}}$.
 Moreover, a winning strategy in $\Dr(\game)$ is also a winning strategy in $\game$, and vice versa.
\end{theorem*}

	We define the fairness constraint on the random edges of $\game$ as per Eq.~\eqref{equ:vl}:
	\[
	\varphi^\ell := \wedge_{(v,v')\in E_r} \square\lozenge v \rightarrow \square\lozenge (v\wedge \bigcirc v').
	\]
	We first show that $\WlR \subseteq \mathcal{W}^{\mathit{a.s.}}$.
	Consider an arbitrary initial vertex $v^0\in \WlR$ and an arbitrary strategy $\rho_1$ of $\p{1}$ in $\game$.
	Let $\rho_0^*$ be a corresponding winning strategy for $\p{0}$ from $v^0$ for the fair adversarial game over $\Dr(\game)$ for the winning condition $\varphi$.
	By definition, $\rho_0^*$ realizes the specification $\varphi$, whenever the adversary satisfies the strong fairness condition on the live edges in $\Dr(\game)$.
	On the other hand, the live edges in $\Dr(\game)$ are exactly the random edges in $\game$.
	In other words, we already know that if we apply the \emph{same} strategy $\rho_0^*$ to $\game$, then $\inf_{\rho_1\in R_1} P_{v^0}^{\rho_0^*,\rho_1}(\game\models \varphi^\ell \rightarrow \varphi) = 1$.
	
	We first show that the random edges $E_r$ also satisfy the strong fairness condition $\varphi^\ell$ \emph{almost surely}; actually we show that the probability of violation of $\varphi^\ell$ in $\game$ is $0$.
	Consider the following:
	\begin{align*}
	 P_{v^0}^{\rho_0^*,\rho_1}\left(\game\models \lnot\varphi^\ell \right)
		 =& \ P_{v^0}^{\rho_0^*,\rho_1}\left(\game\models \lnot \bigwedge_{(v,v')\in E_r} \square\lozenge v \rightarrow \square\lozenge (v\wedge \bigcirc v') \right)\\
		=& \ P_{v^0}^{\rho_0^*,\rho_1}\left(\game\models \bigvee_{(v,v')\in E_r} \square\lozenge v \wedge \lozenge\square \lnot (v\wedge \bigcirc v') \right)\\
		\leq & \  \sum_{(v,v')\in E_r} P_{v^0}^{\rho_0^*,\rho_1}\left(\game\models \square\lozenge v \wedge \lozenge\square \lnot (v\wedge \bigcirc v') \right).
	\end{align*}
	We show that the right-hand side of the last inequality equals to $0$ by proving that for every $(v,v')\in E_r$,
	\[
	P_{v^0}^{\rho_0^*,\rho_1}\left(\game\models \square\lozenge v \wedge \lozenge\square \lnot(v\wedge \bigcirc v') \right) = 0.
	\]
	Consider any arbitrary $(v,v')\in E_r$ and assume that the probability of taking the edge $(v,v')$ from $v$ is $p_1$.
	Let $\pi$ be a play on $\game$ and $(i_0,i_1,i_2,\ldots)$  be the infinite sequence of time indices when the vertex $v$ is visited.
	For every $i_k$, the probability of \emph{not} visiting $v'$ for the next $l$ time steps $(i_{k+1}+1,\ldots,i_{k+l}+1)$ is given by $(1-p)^l$, which converges to $0$ as $l$ approaches~$\infty$.
	This proves that for every $i_k$, eventually there will be a $v'$ at $(i_k+1)$ with probability~$1$; in other words $v'$ will be visited infinitely often with probability $1$.
	Hence, it follows that $ \sum_{(v,v')\in E_r} P_{v^0}^{\rho_0^*,\rho_1}\left(\game\models \square\lozenge v \wedge \lozenge\square \lnot (v\wedge \bigcirc v') \right) = 0$, which in turn establishes that $P_{v^0}^{\rho_0^*,\rho_1}\left(\game\models \lnot\varphi^\ell \right) = 0$.
	
	Now consider the following derivation:
	\begin{align*}
		P_{v^0}^{\rho_0^*,\rho_1}(\game\models \varphi^\ell \rightarrow \varphi)
		= & \ P_{v^0}^{\rho_0^*,\rho_1}(\game\models \lnot\varphi^\ell \vee \varphi)
		\leq P_{v^0}^{\rho_0^*,\rho_1}(\game\models \lnot\varphi^\ell ) + P_{v^0}^{\rho_0^*,\rho_1}(\game\models  \varphi)\\
		= & \ 0 + P_{v^0}^{\rho_0^*,\rho_1}(\game\models  \varphi)
		= P_{v^0}^{\rho_0^*,\rho_1}(\game\models  \varphi).
	\end{align*}
	Since we know that $P_{v^0}^{\rho_0^*,\rho_1}(\game\models \varphi^\ell \rightarrow \varphi) = 1$, hence it follows that $P_{v^0}^{\rho_0^*,\rho_1}(\game\models  \varphi) = 1$.
	
	Next, we show that $\WlR \supseteq \mathcal{W}^{\mathit{a.s.}}$.
	Consider an arbitrary initial vertex $v^0\in \mathcal{W}^{\mathit{a.s.}}$. Let $\rho_0^*$ be a corresponding almost sure winning strategy for $\p{0}$ from $v^0$ in the $\twohalf$-player game $\game$ with the specification $\varphi$. 
	We show that $\p{0}$ wins the fair adversarial game over  $\Dr(\game)$ for the winning condition $\varphi$ from vertex $v^0$ using the strategy $\rho_0^*$.
	
	Let $\rho_1\in R_1$ be any arbitrary $\p{1}$ strategy in the game $\Dr(\game)$ such that the unique resultant play $\pi = (v^0,v^1,\ldots)$ due to $\rho_0^*$ and $\rho_1$ satisfies the fairness assumption.
	We use the notation $\Inf(\pi)$ to denote the set of infinitely occurring vertices along the play $\pi$, i.e., $\Inf(\pi)\coloneqq \set{w\in V \mid \forall m\in \mathbb{N}_0\;.\;\exists n>m\;.\;v^n=w} $.
	First we show that (i) the set of vertices $\Inf(\pi)$ forms an end component in $\game$, and moreover (ii) there exists a $\p{1}$ strategy $\rho_1'$ in the game~$\game$ such that $P_{v^0}^{\rho_0^*,\rho_1'}(\game \models \Inf(\pi)) > 0$.
	Claim~(i) follows by observing the following:
	\begin{itemize}
		\item For all $v\in \Inf(\pi)\cap V_r$, $V_r(v)\subseteq \Inf(\pi)$, as otherwise in $\Dr(\game)$ there would be a vertex in $\El(v)$ and outside $\Inf(\pi)$ which would be visited infinitely many times due to infinitely many visits to $v$.
		\item For every $v\in \Inf(\pi)\cap (V_0\cup V_1)$, $E(v) \neq \emptyset$, as otherwise in $\Dr(\game)$ the play $\pi$ would reach a dead-end.
		\item The subgame graph $\game\downarrow \Inf(\pi)$ is strongly connected, as otherwise in $\Dr(\game)$ there would be two vertices $u,v\in \Inf(\pi)$ so that $v$ would not be reachable from $u$, contradicting the assumption that both $u$ and $v$ are visited infinitely often by $\pi$.
	\end{itemize}
	Claim~(ii) follows by defining a strategy $\rho_1'\equiv \rho_1$ on $\game$.
	Now observe that for every edge $(v,v')$ chosen by $\p{1}$ from a vertex $v\in \dom(\El)$ in $\Dr(\game)$, there exists a corresponding positive probability edge $(v,v')$ in $\game$.
	Since $\Inf(\pi)$ is entered by $\pi$ after finite time steps, hence the Claim~(ii) follows.
	
	Now, from Cor.~\ref{cor:good end component} it follows that there is a $j\in\set{1,2,\ldots,k}$ such that $\Inf(\pi)\cap R_j = \emptyset$ and for every $l\in \set{1,\ldots,m_j}$, $\Inf(\pi)\cap \ps{l}G_j\ne \emptyset$.
	Thus the play $\pi$ satisfies the generalized Rabin condition $\FgR$.
	Since this holds for any arbitrary $\p{1}$ strategy, hence $\WlR \supseteq \mathcal{W}^{\mathit{a.s.}}$ and $\rho^*$ is the corresponding winning strategy for $\p{0}$.


\section{The Accelerated Fixpoint Algorithm}
\label{app:acceleration}

Consider the fixpoint algorithm in \eqref{equ:Rabin_all}. In the correctness proof of \REFthm{thm:FPsoundcomplete} discussed in \REFapp{appD}, we have been remembering so called configuration prefixes $\delta=p_0i_0\hdots p_{j-1}i_{j-1}$ for some $j\leq k$ for every fixpoint variable $X$ (see Eq.~\eqref{equ:RabinConfDomain}). We denoted by $X_{\delta p_j}^{i_j}$ the set of states computed in the $i_j$'th iteration of the fixpoint computation over $X_{p_j}$ after the fixpoint over $Y_{p_j}$ has already terminated within the $i_{j-1}$th iteration over $X_{p_{j-1}}$ after the fixed-point over $Y_{p_{j-1}}$ has terminated in the $i_{j-2}$th iteration over $X_{p_{j-2}}$ and so forth. 

In order to describe the accelerated implementation of \eqref{equ:Rabin_all}, we do not assume that the fixpoints over $Y$-variables have already terminated, but additionally remember their counters~$m$. This leads to configuration prefixes $\delta=p_0m_0i_0\hdots p_{j-1}m_{j-1}i_{j-1}$ and lets us define that $X_{\delta p_j}^{m_ji_j}$ is the set of states computed in the $i_j$th iteration of the fixpoint computation over $X_{p_j}$ during the $m_j$th iteration over $Y_{p_j}$, computing the set $Y_{\delta p_j}^{m_j}$ and so forth.

Given two configuration prefixes $\delta=p_0m_0i_0\hdots p_{j-1}m_{j-1}i_{j-1}$ and \\
$\delta'=p'_0m'_0i'_0\hdots p'_{j-1}m'_{j-1}i'_{j-1}$ we define $\delta<_m\delta'$ if $p_0\hdots p_{j-1}=p'_0\hdots p'_{j-1}$ and $m_0\hdots m_{j-1}<m'_0\hdots m'_{j-1}$ (using the induced lexicographic order) and $i_0\hdots i_{j-1}=i'_0\hdots i'_{j-1}$. We define $\delta<_i\delta'$ similarly.

Now \citet{PitermanPnueli_RabinStreett_2006} showed, based on a result of \citet{long1994improved}, that for every configuration prefix $\delta=p_0m_0i_0\hdots p_{j-1}m_{j-1}i_{j-1}$ the computation of $Y_{\delta p_j}^{0}$ can start from the \emph{minimal} set $Y_{\delta' p_j}^{m_j}$ (instead of the entire set of vertices $V$) such that
$\delta'p_jm_j<_m\delta p_j 0$. Dually, for every configuration prefix $\delta=p_0m_0i_0\hdots p_{j-1}m_{j-1}i_{j-1}$ the computation of $X_{\delta p_j}^{m_j0}$ can start from the \emph{maximal} set $X_{\delta' p_j}^{m_ji_j}$ (instead of the empty set) such that
$\delta'p_jm_ji_j<_i\delta p_jm_j0$.

Further, we see that for the innermost fixpoint, i.e.\ when $j=k$, it follows that for every computation prefix $\delta$ , there can be at most $n$ iterations over both $Y_{p_k}$ and $X_{p_k}$, where $n$ is the total number of vertices. I.e., $n$ different sets $Y_{\delta p_k}^{m_k}$ and $X_{\delta p_k}^{m_k i_k}$ have to be freshly computed for each $\delta p_k$ and $\delta p_k m_k$ respectively. We see that there are $\mathcal{O}(n^{k+1}k!)$ different such permutation sequences. As the computation of the innermost fixpoint dominates the computation time, it is shown by \citet{long1994improved} that this results in an overall worst-case computation time of $\mathcal{O}(n^{(k+1)+1}k!)=\mathcal{O}(n^{k+2}k!)$ (where $n$ is the total number of vertices and $k$ is the number of Rabin pairs).

Unfortunately, the memory requirement of this acceleration algorithm is enormous. To see this, observe that in order to warm-start the computation of $Y_{\delta p_j}^{0}$ with $\delta=p_0m_0i_0\hdots p_{j-1}m_{j-1}i_{j-1}$ we need to store the current minimal set w.r.t.\ the $m$-prefix for every combination of $p$- and $i$-prefixes that can occur in $\delta$, which are $\mathcal{O}(n^{k+1}k!)$ many. Similarly, to warm-start the computation of $X_{\delta p_j}^{m_ji_j}$ we need to store the current minimal set w.r.t. the $i$-prefix for every combination of $p$- and $m$-prefixes that can occur in $\delta$. This means that the memory required by the algorithm is $\mathcal{O}(n^{k+1}k!)$, which is prohibitively large for large values of $n$ and $k$.

We implemented a \emph{space-bounded} version of the acceleration algorithm, where for any given parameter $M$ (chosen by the user), we stored only up to $M$ values for each counter. 
Whenever the values of all the counters are less than $M$, we use the regular acceleration algorithm as outlined above.
Otherwise, if any of the counters exceeds $M$, then we fall back to the regular initialization procedure of fixpoint algorithms, i.e.\ depending on whether it is an~$Y$ or an $X$ variable, initialize it with $V$ or $\emptyset$ respectively.
As a result, the memory requirement of our accelerated fixpoint algorithm is given by $\mathcal{O}(M^{k+1}k!)$.
This space-bounded acceleration algorithm made our implementation much faster and yet practically feasible, as has been demonstrated in \REFsec{sec:experiments}.


\section{Supplementary Results for the Experiments}
\label{app:experiments}

\renewcommand{\arraystretch}{0.8}

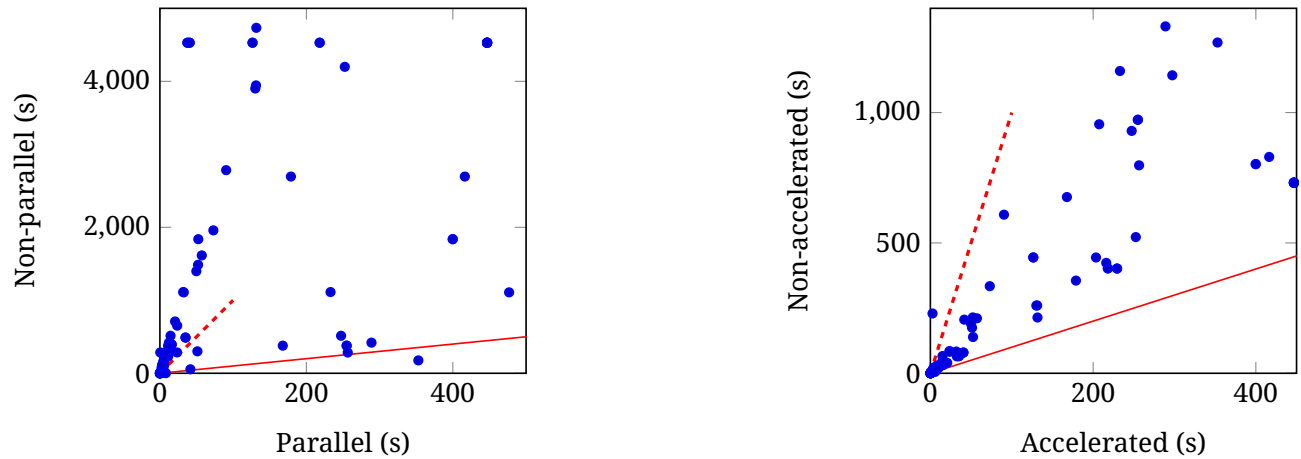
\begin{figure}[H]
	\begin{tikzpicture}[scale=0.6]
\begin{axis}[xmin=0,ymin=0,xmax=500,ymax=5000,xlabel={Parallel\ (\si{\second})}, ylabel={Non-parallel\ (\si{\second})},
					width=5\ThCScm,height=5\ThCScm,mark size=1.25\ThCSpt]


\draw[color=red]		(0,0)	--	(900,900);
\draw[color=red, dashed,  thick] 	(0,0)	--	(100,1000);

\addplot+[only marks]	table[x=acc_+_par,y=acc_+_npar] {DATA/par_vs_npar.txt};

\end{axis}

\end{tikzpicture}\qquad\qquad
	\begin{tikzpicture}[scale=0.6]
\begin{axis}[xmin=0,ymin=0,xmax=450,ymax=1400,xlabel={Accelerated\ (\si{\second})}, ylabel={Non-accelerated\ (\si{\second})},
					width=5\ThCScm,height=5\ThCScm,mark size=1.25\ThCSpt]


\draw[color=red]		(0,0)	--	(900,900);
\draw[color=red, dashed,  thick] 	(0,0)	--	(100,1000);

\addplot+[only marks]	table[x=acc_+_par,y=nacc_+_par] {DATA/accl_vs_naccl.txt};

\end{axis}

\end{tikzpicture}
	\caption{
		Zoomed-in version of \REFfig{fig:BCG scatter plot}.
		(Left) Comparison between the computation times for the non-parallel (1 worker thread) and parallel (48 worker threads) version of \fairsyn, with acceleration being enabled in both cases.
		(Right) Comparison between the computation times for the non-accelerated and the accelerated version of \fairsyn, with parallelization being enabled in both cases.
		(Both) The points on the solid red line represent the same computation time. The points on the dashed red line represent an order of magnitude improvement.
	}
	\label{fig:BCG scatter plot zoomed-in}
\end{figure}

\begin{cfigure}[H]
	\begin{tikzpicture}[scale=0.6]
	\begin{axis}[
		xmin=0,ymin=0,xmax=15,ymax=500,
		xtick={0,5,10,15},
		ymode=log,
		xlabel={$M$},
	    ylabel={Comp.\ time (\si{\second})},
	  	width=0.45\textwidth,height=3\ThCScm,
	    mark options={scale=0.25}
	    ]
	 	 \addplot table [x=M,y=c1] {DATA/vary_mem_rp_1.txt};
		\addplot table [x=M,y=c2] {DATA/vary_mem_rp_1.txt};
		\addplot table [x=M,y=c3] {DATA/vary_mem_rp_1.txt};
		\addplot table [x=M,y=c4] {DATA/vary_mem_rp_1.txt};
		\addplot table [x=M,y=c5] {DATA/vary_mem_rp_1.txt};
		\addplot table [x=M,y=c6] {DATA/vary_mem_rp_1.txt};
		\addplot table [x=M,y=c7] {DATA/vary_mem_rp_1.txt};
		\addplot table [x=M,y=c8] {DATA/vary_mem_rp_1.txt};
		\addplot table [x=M,y=c9] {DATA/vary_mem_rp_1.txt};
		\addplot table [x=M,y=c10] {DATA/vary_mem_rp_1.txt};
	\end{axis}
	
	\end{tikzpicture}
	\begin{tikzpicture}[scale=0.6]
	
	\begin{axis}[
		xmin=2,ymin=0,xmax=15,ymax=500,
		xtick={2,5,10,15},
		ymode=log,
		xlabel={$M$},
	    ylabel={Init.\ time (\si{\second})},
	    width=0.45\textwidth,height=3\ThCScm,
	    mark options={scale=0.25}
	    ]
	 	 \addplot table [x=M,y=c1] {DATA/vary_mem_init_rp_1.txt};
		\addplot table [x=M,y=c2] {DATA/vary_mem_init_rp_1.txt};
		\addplot table [x=M,y=c3] {DATA/vary_mem_init_rp_1.txt};
		\addplot table [x=M,y=c4] {DATA/vary_mem_init_rp_1.txt};
		\addplot table [x=M,y=c5] {DATA/vary_mem_init_rp_1.txt};
		\addplot table [x=M,y=c6] {DATA/vary_mem_init_rp_1.txt};
		\addplot table [x=M,y=c7] {DATA/vary_mem_init_rp_1.txt};
		\addplot table [x=M,y=c8] {DATA/vary_mem_init_rp_1.txt};
		\addplot table [x=M,y=c9] {DATA/vary_mem_init_rp_1.txt};
		\addplot table [x=M,y=c10] {DATA/vary_mem_init_rp_1.txt};
	\end{axis}
	
	\end{tikzpicture}
	\caption{(Left) Effect of variation of the acceleration parameter $M$ on the total computation time (parallelization being enabled) for the VLTS benchmark examples with $1$ Rabin pair.
		(Right) Effect of variation of the acceleration parameter $M$ on the initialization time for the VLTS benchmark examples with $1$ Rabin pair.
		The computation time (Y-axis) in both the plots are shown in the logarithmic scale.}
	\label{fig:BCG trend plots rp 1}
\end{cfigure}

\begin{table}[H]
	\begin{tabular}{
		|>{\centering\arraybackslash}m{3cm}|
		>{\centering\arraybackslash}m{3cm}|
		>{\centering\arraybackslash}m{3cm}|
		>{\centering\arraybackslash}m{3cm}|}
		\hline
		Number of Vertices	&	Number of Transitions	&	Number of Live Edges 	&	Number of BDD Variables	\\
		\hline
	\tablenum[group-separator={,},table-format=7.0]{289}	&	\tablenum[group-separator={,},table-format=7.0]{1224}	&	\tablenum[group-separator={,},table-format=7.0]{17}	&	9	\\
\tablenum[group-separator={,},table-format=7.0]{289}	&	\tablenum[group-separator={,},table-format=7.0]{1224}	&	\tablenum[group-separator={,},table-format=7.0]{25}	&	9	\\
\tablenum[group-separator={,},table-format=7.0]{289}	&	\tablenum[group-separator={,},table-format=7.0]{1224}	&	\tablenum[group-separator={,},table-format=7.0]{13}	&	9	\\
\tablenum[group-separator={,},table-format=7.0]{1952}	&	\tablenum[group-separator={,},table-format=7.0]{2387}	&	\tablenum[group-separator={,},table-format=7.0]{1}	&	11	\\
\tablenum[group-separator={,},table-format=7.0]{1952}	&	\tablenum[group-separator={,},table-format=7.0]{2387}	&	\tablenum[group-separator={,},table-format=7.0]{5}	&	11	\\
\tablenum[group-separator={,},table-format=7.0]{1952}	&	\tablenum[group-separator={,},table-format=7.0]{2387}	&	\tablenum[group-separator={,},table-format=7.0]{25}	&	11	\\
\tablenum[group-separator={,},table-format=7.0]{1183}	&	\tablenum[group-separator={,},table-format=7.0]{4464}	&	\tablenum[group-separator={,},table-format=7.0]{16}	&	11	\\
\tablenum[group-separator={,},table-format=7.0]{1183}	&	\tablenum[group-separator={,},table-format=7.0]{4464}	&	\tablenum[group-separator={,},table-format=7.0]{49}	&	11	\\
\tablenum[group-separator={,},table-format=7.0]{1183}	&	\tablenum[group-separator={,},table-format=7.0]{4464}	&	\tablenum[group-separator={,},table-format=7.0]{9}	&	11	\\
\tablenum[group-separator={,},table-format=7.0]{3995}	&	\tablenum[group-separator={,},table-format=7.0]{14552}	&	\tablenum[group-separator={,},table-format=7.0]{39}	&	12	\\
\tablenum[group-separator={,},table-format=7.0]{3995}	&	\tablenum[group-separator={,},table-format=7.0]{14552}	&	\tablenum[group-separator={,},table-format=7.0]{139}	&	12	\\
\tablenum[group-separator={,},table-format=7.0]{3995}	&	\tablenum[group-separator={,},table-format=7.0]{14552}	&	\tablenum[group-separator={,},table-format=7.0]{153}	&	12	\\
\tablenum[group-separator={,},table-format=7.0]{5121}	&	\tablenum[group-separator={,},table-format=7.0]{9392}	&	\tablenum[group-separator={,},table-format=7.0]{1}	&	13	\\
\tablenum[group-separator={,},table-format=7.0]{5121}	&	\tablenum[group-separator={,},table-format=7.0]{9392}	&	\tablenum[group-separator={,},table-format=7.0]{54}	&	13	\\
\tablenum[group-separator={,},table-format=7.0]{5121}	&	\tablenum[group-separator={,},table-format=7.0]{9392}	&	\tablenum[group-separator={,},table-format=7.0]{73}	&	13	\\
\tablenum[group-separator={,},table-format=7.0]{8879}	&	\tablenum[group-separator={,},table-format=7.0]{24411}	&	\tablenum[group-separator={,},table-format=7.0]{473}	&	14	\\
\tablenum[group-separator={,},table-format=7.0]{8879}	&	\tablenum[group-separator={,},table-format=7.0]{24411}	&	\tablenum[group-separator={,},table-format=7.0]{397}	&	14	\\
\tablenum[group-separator={,},table-format=7.0]{7119}	&	\tablenum[group-separator={,},table-format=7.0]{38424}	&	\tablenum[group-separator={,},table-format=7.0]{626}	&	14	\\
\tablenum[group-separator={,},table-format=7.0]{7119}	&	\tablenum[group-separator={,},table-format=7.0]{38424}	&	\tablenum[group-separator={,},table-format=7.0]{835}	&	14	\\
\tablenum[group-separator={,},table-format=7.0]{7119}	&	\tablenum[group-separator={,},table-format=7.0]{38424}	&	\tablenum[group-separator={,},table-format=7.0]{597}	&	14	\\
\tablenum[group-separator={,},table-format=7.0]{10849}	&	\tablenum[group-separator={,},table-format=7.0]{56156}	&	\tablenum[group-separator={,},table-format=7.0]{241}	&	14	\\
\tablenum[group-separator={,},table-format=7.0]{10849}	&	\tablenum[group-separator={,},table-format=7.0]{56156}	&	\tablenum[group-separator={,},table-format=7.0]{482}	&	14	\\
\tablenum[group-separator={,},table-format=7.0]{18746}	&	\tablenum[group-separator={,},table-format=7.0]{73043}	&	\tablenum[group-separator={,},table-format=7.0]{1585}	&	15	\\
\tablenum[group-separator={,},table-format=7.0]{18746}	&	\tablenum[group-separator={,},table-format=7.0]{73043}	&	\tablenum[group-separator={,},table-format=7.0]{1729}	&	15	\\
\tablenum[group-separator={,},table-format=7.0]{18746}	&	\tablenum[group-separator={,},table-format=7.0]{73043}	&	\tablenum[group-separator={,},table-format=7.0]{575}	&	15	\\
\tablenum[group-separator={,},table-format=7.0]{25216}	&	\tablenum[group-separator={,},table-format=7.0]{25216}	&	\tablenum[group-separator={,},table-format=7.0]{137}	&	15	\\
\tablenum[group-separator={,},table-format=7.0]{25216}	&	\tablenum[group-separator={,},table-format=7.0]{25216}	&	\tablenum[group-separator={,},table-format=7.0]{595}	&	15	\\
\tablenum[group-separator={,},table-format=7.0]{25216}	&	\tablenum[group-separator={,},table-format=7.0]{25216}	&	\tablenum[group-separator={,},table-format=7.0]{373}	&	15	\\
\tablenum[group-separator={,},table-format=7.0]{40006}	&	\tablenum[group-separator={,},table-format=7.0]{60007}	&	\tablenum[group-separator={,},table-format=7.0]{1130}	&	16	\\
\tablenum[group-separator={,},table-format=7.0]{40006}	&	\tablenum[group-separator={,},table-format=7.0]{60007}	&	\tablenum[group-separator={,},table-format=7.0]{865}	&	16	\\
\tablenum[group-separator={,},table-format=7.0]{52268}	&	\tablenum[group-separator={,},table-format=7.0]{292823}	&	\tablenum[group-separator={,},table-format=7.0]{107}	&	16	\\
\tablenum[group-separator={,},table-format=7.0]{52268}	&	\tablenum[group-separator={,},table-format=7.0]{292823}	&	\tablenum[group-separator={,},table-format=7.0]{3254}	&	16	\\
\tablenum[group-separator={,},table-format=7.0]{65537}	&	\tablenum[group-separator={,},table-format=7.0]{524293}	&	\tablenum[group-separator={,},table-format=7.0]{13727}	&	17	\\
\tablenum[group-separator={,},table-format=7.0]{65537}	&	\tablenum[group-separator={,},table-format=7.0]{524293}	&	\tablenum[group-separator={,},table-format=7.0]{25229}	&	17	\\
\tablenum[group-separator={,},table-format=7.0]{66929}	&	\tablenum[group-separator={,},table-format=7.0]{569322}	&	\tablenum[group-separator={,},table-format=7.0]{23290}	&	17	\\
\tablenum[group-separator={,},table-format=7.0]{66929}	&	\tablenum[group-separator={,},table-format=7.0]{569322}	&	\tablenum[group-separator={,},table-format=7.0]{13698}	&	17	\\
\tablenum[group-separator={,},table-format=7.0]{69753}	&	\tablenum[group-separator={,},table-format=7.0]{359575}	&	\tablenum[group-separator={,},table-format=7.0]{11071}	&	17	\\
\tablenum[group-separator={,},table-format=7.0]{69753}	&	\tablenum[group-separator={,},table-format=7.0]{359575}	&	\tablenum[group-separator={,},table-format=7.0]{5058}	&	17	\\
\tablenum[group-separator={,},table-format=7.0]{83435}	&	\tablenum[group-separator={,},table-format=7.0]{259488}	&	\tablenum[group-separator={,},table-format=7.0]{1682}	&	17	\\
\tablenum[group-separator={,},table-format=7.0]{83435}	&	\tablenum[group-separator={,},table-format=7.0]{259488}	&	\tablenum[group-separator={,},table-format=7.0]{2707}	&	17	\\
\tablenum[group-separator={,},table-format=7.0]{96878}	&	\tablenum[group-separator={,},table-format=7.0]{282880}	&	\tablenum[group-separator={,},table-format=7.0]{6225}	&	18	\\
\tablenum[group-separator={,},table-format=7.0]{96878}	&	\tablenum[group-separator={,},table-format=7.0]{282880}	&	\tablenum[group-separator={,},table-format=7.0]{585}	&	18	\\
		\hline
	\end{tabular}
	\caption{Details of the fair adversarial Rabin games randomly generated  from the VLTS benchmark suite. Continued to \REFtab{tab:vlts benchmark data 2}.}
	\label{tab:vlts benchmark data 1}
\end{table}

\begin{table}[H]
	\begin{tabular}{
		|>{\centering\arraybackslash}m{3cm}|
		>{\centering\arraybackslash}m{3cm}|
		>{\centering\arraybackslash}m{3cm}|
		>{\centering\arraybackslash}m{3cm}|}
		\hline
		Number of Vertices	&	Number of Transitions	&	Number of Live Edges 	&	Number of BDD Variables	\\
		\hline
		\tablenum[group-separator={,},table-format=7.0]{116456}	&	\tablenum[group-separator={,},table-format=7.0]{364596}	&	\tablenum[group-separator={,},table-format=7.0]{8316}	&	17	\\
\tablenum[group-separator={,},table-format=7.0]{116456}	&	\tablenum[group-separator={,},table-format=7.0]{364596}	&	\tablenum[group-separator={,},table-format=7.0]{7774}	&	17	\\
		\tablenum[group-separator={,},table-format=7.0]{142471}	&	\tablenum[group-separator={,},table-format=7.0]{925429}	&	\tablenum[group-separator={,},table-format=7.0]{19259}	&	18	\\
\tablenum[group-separator={,},table-format=7.0]{142471}	&	\tablenum[group-separator={,},table-format=7.0]{925429}	&	\tablenum[group-separator={,},table-format=7.0]{3304}	&	18	\\
		\tablenum[group-separator={,},table-format=7.0]{164865}	&	\tablenum[group-separator={,},table-format=7.0]{1619200}	&	\tablenum[group-separator={,},table-format=7.0]{13407}	&	18	\\
		\tablenum[group-separator={,},table-format=7.0]{164865}	&	\tablenum[group-separator={,},table-format=7.0]{1619200}	&	\tablenum[group-separator={,},table-format=7.0]{24868}	&	18	\\
		\tablenum[group-separator={,},table-format=7.0]{166463}	&	\tablenum[group-separator={,},table-format=7.0]{518976}	&	\tablenum[group-separator={,},table-format=7.0]{13633}	&	18	\\
		\tablenum[group-separator={,},table-format=7.0]{166463}	&	\tablenum[group-separator={,},table-format=7.0]{518976}	&	\tablenum[group-separator={,},table-format=7.0]{4155}	&	18	\\
		\tablenum[group-separator={,},table-format=7.0]{214140}	&	\tablenum[group-separator={,},table-format=7.0]{683205}	&	\tablenum[group-separator={,},table-format=7.0]{13588}	&	18	\\
		\tablenum[group-separator={,},table-format=7.0]{214140}	&	\tablenum[group-separator={,},table-format=7.0]{683205}	&	\tablenum[group-separator={,},table-format=7.0]{12113}	&	18	\\
		\tablenum[group-separator={,},table-format=7.0]{371804}	&	\tablenum[group-separator={,},table-format=7.0]{641565}	&	\tablenum[group-separator={,},table-format=7.0]{3413}	&	19	\\
		\tablenum[group-separator={,},table-format=7.0]{371804}	&	\tablenum[group-separator={,},table-format=7.0]{641565}	&	\tablenum[group-separator={,},table-format=7.0]{12151}	&	19	\\
		\tablenum[group-separator={,},table-format=7.0]{386496}	&	\tablenum[group-separator={,},table-format=7.0]{1171870}	&	\tablenum[group-separator={,},table-format=7.0]{26247}	&	19	\\
		\tablenum[group-separator={,},table-format=7.0]{386496}	&	\tablenum[group-separator={,},table-format=7.0]{1171870}	&	\tablenum[group-separator={,},table-format=7.0]{17823}	&	19	\\
		\tablenum[group-separator={,},table-format=7.0]{566639}	&	\tablenum[group-separator={,},table-format=7.0]{3984160}	&	\tablenum[group-separator={,},table-format=7.0]{7109}	&	20	\\
		\tablenum[group-separator={,},table-format=7.0]{566639}	&	\tablenum[group-separator={,},table-format=7.0]{3984160}	&	\tablenum[group-separator={,},table-format=7.0]{42757}	&	20	\\
		\hline
	\end{tabular}
	\caption{Continued from \REFtab{tab:vlts benchmark data 1}. Details of the fair adversarial Rabin games randomly generated  from the VLTS benchmark suite.}
	\label{tab:vlts benchmark data 2}
\end{table}

\begin{table}[H]
	\begin{tabular}{
		|>{\centering\arraybackslash}m{2cm}|
		>{\centering\arraybackslash}m{1.9cm}|
		>{\centering\arraybackslash}m{3cm}|
		>{\centering\arraybackslash}m{3cm}|
		>{\centering\arraybackslash}m{3cm}|
		>{\centering\arraybackslash}m{2.1cm}|
		>{\centering\arraybackslash}m{2.1cm}|}
		\hline
		Broadcast Queue Capacity	&	Output Queue Capacity	&	Number of Vertices	&	Number of Transitions	&	Number of Live Edges	&	Number of BDD Variables	&	Time (seconds)	\\
		\hline
1	&	1	&	\tablenum[group-separator={,},table-format=10.0]{5307840}	&	\tablenum[group-separator={,},table-format=10.0]{10135300}	&	\tablenum[group-separator={,},table-format=10.0]{5124100}	&	25	&	\tablenum[table-format=5.2]{7.37}	\\
2	&	1	&	\tablenum[group-separator={,},table-format=10.0]{21231400}	&	\tablenum[group-separator={,},table-format=10.0]{40541200}	&	\tablenum[group-separator={,},table-format=10.0]{20496400}	&	27	&	\tablenum[table-format=5.2]{24.90}	\\
3	&	1	&	\tablenum[group-separator={,},table-format=10.0]{21414100}	&	\tablenum[group-separator={,},table-format=10.0]{42080300}	&	\tablenum[group-separator={,},table-format=10.0]{21265900}	&	27	&	\tablenum[table-format=5.2]{28.97}	\\
1	&	2	&	\tablenum[group-separator={,},table-format=10.0]{21340800}	&	\tablenum[group-separator={,},table-format=10.0]{40879100}	&	\tablenum[group-separator={,},table-format=10.0]{20834300}	&	27	&	\tablenum[table-format=5.2]{38.25}	\\
1	&	3	&	\tablenum[group-separator={,},table-format=10.0]{21559400}	&	\tablenum[group-separator={,},table-format=10.0]{42756100}	&	\tablenum[group-separator={,},table-format=10.0]{21772800}	&	27	&	\tablenum[table-format=5.2]{51.55}	\\
4	&	1	&	\tablenum[group-separator={,},table-format=10.0]{84925400}	&	\tablenum[group-separator={,},table-format=10.0]{162165000}	&	\tablenum[group-separator={,},table-format=10.0]{81985500}	&	29	&	\tablenum[table-format=5.2]{57.70}	\\
5	&	1	&	\tablenum[group-separator={,},table-format=10.0]{85295700}	&	\tablenum[group-separator={,},table-format=10.0]{165243000}	&	\tablenum[group-separator={,},table-format=10.0]{83524600}	&	29	&	\tablenum[table-format=5.2]{65.01}	\\
6	&	1	&	\tablenum[group-separator={,},table-format=10.0]{85656300}	&	\tablenum[group-separator={,},table-format=10.0]{168321000}	&	\tablenum[group-separator={,},table-format=10.0]{85063700}	&	29	&	\tablenum[table-format=5.2]{73.19}	\\
7	&	1	&	\tablenum[group-separator={,},table-format=10.0]{86007400}	&	\tablenum[group-separator={,},table-format=10.0]{171399000}	&	\tablenum[group-separator={,},table-format=10.0]{86602800}	&	29	&	\tablenum[table-format=5.2]{77.97}	\\
1	&	4	&	\tablenum[group-separator={,},table-format=10.0]{85363200}	&	\tablenum[group-separator={,},table-format=10.0]{163516000}	&	\tablenum[group-separator={,},table-format=10.0]{83337200}	&	29	&	\tablenum[table-format=5.2]{92.56}	\\
1	&	5	&	\tablenum[group-separator={,},table-format=10.0]{85808000}	&	\tablenum[group-separator={,},table-format=10.0]{167270000}	&	\tablenum[group-separator={,},table-format=10.0]{85214200}	&	29	&	\tablenum[table-format=5.2]{113.18}	\\
2	&	2	&	\tablenum[group-separator={,},table-format=10.0]{85363200}	&	\tablenum[group-separator={,},table-format=10.0]{163516000}	&	\tablenum[group-separator={,},table-format=10.0]{83337200}	&	29	&	\tablenum[table-format=5.2]{133.20}	\\
1	&	6	&	\tablenum[group-separator={,},table-format=10.0]{86237400}	&	\tablenum[group-separator={,},table-format=10.0]{171024000}	&	\tablenum[group-separator={,},table-format=10.0]{87091200}	&	29	&	\tablenum[table-format=5.2]{135.67}	\\
3	&	2	&	\tablenum[group-separator={,},table-format=10.0]{86061400}	&	\tablenum[group-separator={,},table-format=10.0]{169673000}	&	\tablenum[group-separator={,},table-format=10.0]{86415400}	&	29	&	\tablenum[table-format=5.2]{144.27}	\\
1	&	7	&	\tablenum[group-separator={,},table-format=10.0]{86651500}	&	\tablenum[group-separator={,},table-format=10.0]{174778000}	&	\tablenum[group-separator={,},table-format=10.0]{88968200}	&	29	&	\tablenum[table-format=5.2]{145.76}	\\
8	&	1	&	\tablenum[group-separator={,},table-format=10.0]{339702000}	&	\tablenum[group-separator={,},table-format=10.0]{648659000}	&	\tablenum[group-separator={,},table-format=10.0]{327942000}	&	31	&	\tablenum[table-format=5.2]{149.68}	\\
2	&	3	&	\tablenum[group-separator={,},table-format=10.0]{86237400}	&	\tablenum[group-separator={,},table-format=10.0]{171024000}	&	\tablenum[group-separator={,},table-format=10.0]{87091200}	&	29	&	\tablenum[table-format=5.2]{163.62}	\\
9	&	1	&	\tablenum[group-separator={,},table-format=10.0]{340447000}	&	\tablenum[group-separator={,},table-format=10.0]{654815000}	&	\tablenum[group-separator={,},table-format=10.0]{331020000}	&	31	&	\tablenum[table-format=5.2]{174.29}	\\
10	&	1	&	\tablenum[group-separator={,},table-format=10.0]{341183000}	&	\tablenum[group-separator={,},table-format=10.0]{660972000}	&	\tablenum[group-separator={,},table-format=10.0]{334098000}	&	31	&	\tablenum[table-format=5.2]{197.02}	\\
3	&	3	&	\tablenum[group-separator={,},table-format=10.0]{86870100}	&	\tablenum[group-separator={,},table-format=10.0]{177181000}	&	\tablenum[group-separator={,},table-format=10.0]{90169300}	&	29	&	\tablenum[table-format=5.2]{203.15}	\\
1	&	8	&	\tablenum[group-separator={,},table-format=10.0]{341453000}	&	\tablenum[group-separator={,},table-format=10.0]{654066000}	&	\tablenum[group-separator={,},table-format=10.0]{333349000}	&	31	&	\tablenum[table-format=5.2]{248.38}	\\
1	&	9	&	\tablenum[group-separator={,},table-format=10.0]{342350000}	&	\tablenum[group-separator={,},table-format=10.0]{661574000}	&	\tablenum[group-separator={,},table-format=10.0]{337103000}	&	31	&	\tablenum[table-format=5.2]{283.85}	\\
1	&	10	&	\tablenum[group-separator={,},table-format=10.0]{343232000}	&	\tablenum[group-separator={,},table-format=10.0]{669082000}	&	\tablenum[group-separator={,},table-format=10.0]{340857000}	&	31	&	\tablenum[table-format=5.2]{331.78}	\\
7	&	2	&	\tablenum[group-separator={,},table-format=10.0]{345587000}	&	\tablenum[group-separator={,},table-format=10.0]{691003000}	&	\tablenum[group-separator={,},table-format=10.0]{351818000}	&	31	&	\tablenum[table-format=5.2]{567.26}	\\
4	&	2	&	\tablenum[group-separator={,},table-format=10.0]{341453000}	&	\tablenum[group-separator={,},table-format=10.0]{654066000}	&	\tablenum[group-separator={,},table-format=10.0]{333349000}	&	31	&	\tablenum[table-format=5.2]{710.78}	\\
2	&	4	&	\tablenum[group-separator={,},table-format=10.0]{341453000}	&	\tablenum[group-separator={,},table-format=10.0]{654066000}	&	\tablenum[group-separator={,},table-format=10.0]{333349000}	&	31	&	\tablenum[table-format=5.2]{806.74}	\\
5	&	2	&	\tablenum[group-separator={,},table-format=10.0]{342868000}	&	\tablenum[group-separator={,},table-format=10.0]{666378000}	&	\tablenum[group-separator={,},table-format=10.0]{339505000}	&	31	&	\tablenum[table-format=5.2]{852.37}	\\
6	&	2	&	\tablenum[group-separator={,},table-format=10.0]{344246000}	&	\tablenum[group-separator={,},table-format=10.0]{678691000}	&	\tablenum[group-separator={,},table-format=10.0]{345661000}	&	31	&	\tablenum[table-format=5.2]{936.04}	\\
2	&	5	&	\tablenum[group-separator={,},table-format=10.0]{343232000}	&	\tablenum[group-separator={,},table-format=10.0]{669082000}	&	\tablenum[group-separator={,},table-format=10.0]{340857000}	&	31	&	\tablenum[table-format=5.2]{1034.57}	\\
4	&	3	&	\tablenum[group-separator={,},table-format=10.0]{344950000}	&	\tablenum[group-separator={,},table-format=10.0]{684098000}	&	\tablenum[group-separator={,},table-format=10.0]{348365000}	&	31	&	\tablenum[table-format=5.2]{1071.52}	\\
2	&	7	&	\tablenum[group-separator={,},table-format=10.0]{346606000}	&	\tablenum[group-separator={,},table-format=10.0]{699113000}	&	\tablenum[group-separator={,},table-format=10.0]{355873000}	&	31	&	\tablenum[table-format=5.2]{1111.64}	\\
7	&	3	&	\tablenum[group-separator={,},table-format=10.0]{348693000}	&	\tablenum[group-separator={,},table-format=10.0]{721035000}	&	\tablenum[group-separator={,},table-format=10.0]{366834000}	&	31	&	\tablenum[table-format=5.2]{1312.88}	\\
2	&	6	&	\tablenum[group-separator={,},table-format=10.0]{344950000}	&	\tablenum[group-separator={,},table-format=10.0]{684098000}	&	\tablenum[group-separator={,},table-format=10.0]{348365000}	&	31	&	\tablenum[table-format=5.2]{1336.35}	\\
5	&	3	&	\tablenum[group-separator={,},table-format=10.0]{346233000}	&	\tablenum[group-separator={,},table-format=10.0]{696410000}	&	\tablenum[group-separator={,},table-format=10.0]{354521000}	&	31	&	\tablenum[table-format=5.2]{1351.31}	\\
3	&	4	&	\tablenum[group-separator={,},table-format=10.0]{344246000}	&	\tablenum[group-separator={,},table-format=10.0]{678691000}	&	\tablenum[group-separator={,},table-format=10.0]{345661000}	&	31	&	\tablenum[table-format=5.2]{1632.63}	\\
6	&	3	&	\tablenum[group-separator={,},table-format=10.0]{347480000}	&	\tablenum[group-separator={,},table-format=10.0]{708723000}	&	\tablenum[group-separator={,},table-format=10.0]{360677000}	&	31	&	\tablenum[table-format=5.2]{1667.54}	\\
8	&	2	&	\tablenum[group-separator={,},table-format=10.0]{1365810000}	&	\tablenum[group-separator={,},table-format=10.0]{2616260000}	&	\tablenum[group-separator={,},table-format=10.0]{1333400000}	&	33	&	\tablenum[table-format=5.2]{2478.13}	\\
9	&	2	&	\tablenum[group-separator={,},table-format=10.0]{1368660000}	&	\tablenum[group-separator={,},table-format=10.0]{2640890000}	&	\tablenum[group-separator={,},table-format=10.0]{1345710000}	&	33	&	\tablenum[table-format=5.2]{2783.77}	\\
		\hline
	\end{tabular}
	\caption{Experimental evaluation for the code-aware resource management case study (extended table).}
	\label{tab:resource management extended}
\end{table}

	
\end{document}